%% file: main-set-interpretations.tex
\documentclass[a4paper,USenglish,cleveref, autoref, thm-restate, numberwithinsect]{lipics-v2021}
\hideLIPIcs  %uncomment to remove references to LIPIcs series (logo, DOI, ...), e.g. when preparing a pre-final version to be uploaded to arXiv or another public repository
\nolinenumbers

\usepackage{stmaryrd}
\usepackage{todonotes}
\usepackage[notion,quotation, paper]{knowledge}
\knowledgeconfigure{diagnose line=true}
\usepackage{lowco2}

\usepackage{mathtools}
\usepackage{csquotes}
\usepackage{cmll}
\usepackage{booktabs}
\usepackage{comment}
\usepackage{xspace}
\usepackage[oldenum,olditem]{paralist}

\input{kl-set-interpretations.kl}

\input{macros}


\title{Expregular functions}
\author{Thomas Colcombet}{CNRS, IRIF, Universit\'e Paris Cité, France \and \url{https://www.irif.fr/~colcombe}}{thomas.colcombet@irif.fr}{https://orcid.org/0000-0001-6529-6963}{}
\author{Nathan Lhote}{LIS, Université Aix-Marseille, France \and \url{https://pageperso.lis-lab.fr/nathan.lhote/}}{nathan.lhote@lis-lab.fr}{https://orcid.org/0000-0003-3303-5368}{This work was partly supported by ANR QUASY 23-CE48-0008.}
\author{Pierre Ohlmann}{CNRS, LIS, Université Aix-Marseille, France \and \url{https://pageperso.lis-lab.fr/pierre.ohlmann/}}{pierre.ohlmann@lis-lab.fr}{https://orcid.org/0000-0002-4685-5253}{}

\authorrunning{T. Colcombet, N. Lhote, P. Ohlmann}
\acknowledgements{We express our gratitude to Lê Thành Dũng (Tito) Nguy{\~{\^{e}}}n for contributing several crucial insights to this project, and in particular the use of yield-Hennie machines.
Nathan Lhote would like to thank Pierre-Alain Reynier, Emmanuel Filiot, Rafa\l\  Stefa\'nski and Miko\l aj Boja\'nczyk for valuable discussions at the early stages of this line of research.}

\begin{document}

\maketitle
\begin{abstract}
	Polyregular functions form a robust class of string-to-string functions with polynomial growth, as evidenced by Boja\'{n}czyk (2018).
	This class admits numerous descriptions and enjoys several closure properties.
	Most notably, polyregular functions are regularity reflecting (\ie the inverse image of a regular language is regular).
   
	In this work, we propose a robust class of string-to-string functions with exponential growth which we call expregular functions. We consider the following three models for describing them:
	\begin{itemize}
		\item MSO set interpretations, which extend MSO interpretations (one of the models capturing polyregular functions), by operating on monadic variables instead of tuples of first-order variables;
		\item yield-Hennie machines, which are branching one-tape Turing machines with bounded visit; and
		\item Ariadne transducers, a new model of 2-way pushdown machines with a bounded visit restriction.
	\end{itemize}
	Our main contribution is a translation from MSO set interpretations to yield-Hennie machines, which are known to be regularity reflecting (Dartois, Nguy\~{\^{e}}n, Peyrat 2026).
	In particular this establishes that MSO set interpretations are regularity reflecting, which in turn settles a major conjecture about automatic structures: every automatic $\omega$-word has a decidable MSO theory.

	Yield-Hennie machine directly translate to Ariadne transducers, and our second contribution is to prove that Ariadne transducers also translate to MSO set interpretations, thus establishing the equivalence of the three models. This is obtained by showing that that Ariadne automata --- the automaton model corresponding to Ariadne transducers --- recognise regular languages.
\end{abstract}

\vskip\topmattervskip\baselineskip
\noindent\textcolor{lipicsGray}{\fontsize{9}{12}\sffamily\bfseries
                      \lowcotwobold\ Low-co2 research paper} (\lowcotwourl[v1]) This research
was developed, written, submitted and presented without the use of air travel.

\bigskip

\textit{The document contains internal hyperlinks for helping the reader on electronic devices.}

\newpage
\tableofcontents
\newpage

% \newpage % Tito: même sans ma todonote, la table des matières déborde, et ça va empirer en rajoutant des sous-sections…

\knowledgenewrobustcmd\universe{\cmdkl{\mathit{U}}}

\section{Introduction}
\label{sec:introduction}
\input{introduction.tex}
\section{Three formalisms for expregular functions}
\label{sec:three-formalisms}
\input{machine-models.tex}

\section{Equivalence of the three models}
\label{sec:equivalence}
\input{equivalence-of-models.tex}

\section{Discussion}
\label{sec:discussion}
\input{discussion.tex}

\section{From MSO set interpretations to yield-Hennie machines}
\label{sec:msosi-to-hennie}
\input{msosi-to-hennie}

\section{From yield-Hennie machines to Ariadne transducers}
\label{sec:hennie-to-ariadne}
\input{hennie-to-ariadne}

\section{From Ariadne transducers to MSO set interpretations}
\label{sec:ariadne-to-msosi}

\input{ariadne-to-msosi}

\section{Regularity of Ariadne automata}
\label{sec:reflection}

\input{ariadne-reflection.tex}

\bibliography{biblio}

\end{document}

%% file: macros.tex
\knowledgenewcommand{\re}[1]{\mathrel{\cmdkl{\xrightarrow{#1}}}}

\newcommand\thomas[1]{}
   
  \newcommand\nathan[1]{}
  \newcommand\nat[1]{}
  \newcommand\pierre[1]{}
  \newcommand\tito[1]{}

\newcommand{\N}{\mathbb N}

\newcommand{\fin}[1]{W_{\mathrm{fin}}}

\knowledgenewcommand{\loo}{\mathrm{\cmdkl{loop}}}

\knowledgenewcommand{\ord}{\mathrm{\cmdkl{order}}}

\knowledgenewcommand\pleq{\mathrel{\cmdkl\preccurlyeq}}

\knowledgenewcommand\convexhull{\mathrm{\cmdkl{hull}}}

\newcommand\axisleq{\mathrel{\kl[\axisleq]{\leqslant}}}
\knowledge{\axisleq}{notion}

\newcommand{\pos}{\mathrm{conf}}
\newcommand{\tup}[1]{\mathbf{#1}}

\newcommand{\la}{\leftarrow}
\newcommand{\ra}{\rightarrow}
\newcommand{\stay}{\circlearrowright}

\newcommand{\Hh}{\mathcal{H}}

\knowledgenewrobustcmd{\typ}[3]{\mathrm{\cmdkl{type}}_{#1}(#2,#3)}

\newcommand{\lef}[1]{\overleftarrow{#1}}
\newcommand{\rig}[1]{\overrightarrow{#1}}

\newcommand{\Xf}[2]{#1^{(#2)}}
\newcommand{\existsf}[1]{\exists^{(#1)}}

\knowledgenewcommand{\qr}{\mathrm{\cmdkl{rank}}}
\newcommand{\midd}{\mathrm{middle}}

\newcommand{\prods}{\otimes_s}

\newcommand{\phimwf}{\phi^{\mathrm{mwf}}}

\newcommand{\mem}{\mathrm{mem}}
\newcommand{\cont}{\mathrm{sec}}
\newcommand{\cov}{\mathrm{cov}}

\newcommand{\dir}{\mathrm{dir}}

\newcommand{\step}{\mathrm{step}}

\newcommand{\childsplopp}{%
  {\operatorname{child\text{-}spl\text{-}opp}}%
}
\newcommand{\childsplsame}{%
  {\operatorname{child\text{-}spl\text{-}same}}%
}
\newcommand{\childsplsucc}{%
  {\operatorname{child\text{-}spl\text{-}suc}}%
}

\newcommand{\childsplx}{%
  {\operatorname{child\text{-}spl\text{-}}\!x}%
}

\newcommand{\childspltilex}{%
  {\operatorname{child\text{-}spl\text{-}tile\text{-}}\!x}%
}

\newcommand{\childint}{%
  {\operatorname{child\text{-}int}}%
}

\newcommand{\dirsucc}{%
  {\operatorname{inter\text{-}int}}%
}

\newcommand{\dirsuccsplprime}{%
  {\operatorname{inter\text{-}spl}}%
}

\newcommand{\suc}{\mathrm{suc}}
\newcommand{\opp}{\mathrm{opp}}
\newcommand{\same}{\mathrm{same}}

\newcommand{\out}{\mathrm{out}}

\newcommand{\borders}{{\bowtie}}
\newcommand{\before}{\mathrm{bef}}

\newcommand{\closed}{\mathrm{closed}}

\newcommand{\triangles}[1]{\triangleright #1 \triangleleft}

\newcommand{\pre}{\mathrm{pre}}

\newcommand{\Typ}{\mathrm{Typ}}

\newcommand{\tile}{\lambda}

\newcommand{\produce}[1]{\mathrm{produce}(#1)}
\newcommand{\producenothing}{\mathrm{no-prod}}

\newcommand{\move}{\mathrm{move}}

\newcommand{\state}{\mathrm{state}}
\newcommand{\tape}{\mathrm{tape}}

\newcommand{\valid}{\mathrm{valid}}

\newcommand{\set}[1]{\{ #1 \}}
\newcommand{\rprefix}{\mathsf{rev{-}prefix}}

\newcommand{\ie}{\textit{i.e.~}}
\newcommand{\eg}{\textit{e.g.~}}
\newcommand{\cf}{\textit{cf.~}}
\newcommand{\etal}{\textit{et al.}~}
\newcommand{\etc}{\textit{etc}}

\newcommand{\type}[2]{\mathrm{type}(#1,#2)}

\knowledgenewrobustcmd{\push}{\mathtt{\cmdkl{push}}}
\knowledgenewrobustcmd{\pop}{\mathtt{\cmdkl{pop}}}
\knowledgenewrobustcmd{\update}{\mathrm{\cmdkl{Update}}}

\newcommand{\view}{\mathrm{view}}
\newcommand{\dest}{\mathrm{dest}}
\newcommand{\action}{\mathsf{action}}

\knowledgenewrobustcmd{\ttop}{\mathrm{\cmdkl{top}}}
\knowledgenewrobustcmd{\untop}{\mathrm{\cmdkl{untop}}}
\newcommand{\init}{q_0}

\knowledgenewcommand{\nextchild}{\mathsf{\cmdkl{nextchild}}}
\newcommand{\finish}{\textsf{finish}}
\newcommand{\tower}{\textsc{Tower}\xspace}

\newcommand{\mypiccore}[1]{
	\includegraphics[page=#1,scale=0.3]{msosi}
}

\newcommand{\mypic}[1]{
	\begin{center}
		\mypiccore{#1}
	\end{center}
}
\newcommand{\mypicsmaller}[1]{
	\begin{center}
		\includegraphics[page=#1,scale=0.261]{msosi}
	\end{center}
}

\newcommand{\true}{\mathtt{true}}
\newcommand{\false}{\mathtt{false}}

\knowledgenewrobustcmd\locview[2]{#1\cmdkl{|}_{#2}}

\newrobustcmd\newfreshHennie[2]{%
	\def\Hennie{\kl[\Hennie{#1}]{\mathcal{H}}}%
	\knowledge{\Hennie{#1}}{notion}%
	\def\stateshennie{\kl[\stateshennie{#1}]{Q#2}}%
	\knowledge{\stateshennie{#1}}{notion}
	\def\transhennie{\kl[\transhennie{#1}]{\delta#2}}%
	\knowledge{\transhennie{#1}}{notion}%
	\def\tapehennie{\kl[\tapehennie{#1}]{\Theta#2}}%
	\knowledge{\tapehennie{#1}}{notion}%
	}
\newfreshHennie{}{}
\newrobustcmd\newfreshAriadne[2]{%
	\def\Ariadne{\kl[\Ariadne{#1}]{\mathcal A}}%
	\knowledge{\Ariadne{#1}}{notion}%
	\def\statesariadne{\kl[\statesariadne{#1}]{Q#2}}%
	\knowledge{\statesariadne{#1}}{notion}%
	\def\transariadne{\kl[\transariadne{#1}]{\delta#2}}%
	\knowledge{\transariadne{#1}}{notion}%
	\def\acceptariadne{\kl[\acceptariadne{#1}]{F#2}}%
	\knowledge{\acceptariadne{#1}}{notion}%
	}
\newfreshAriadne{}{}
\def\aria{{\Ariadne}}

\knowledgenewcommand{\True}{\cmdkl{\mathtt T}}

%% file: introduction.tex
% !TEX root =  main-set-interpretations.tex

\reintro{MSO set interpretations} are a logical formalism introduced by Colcombet and Löding~\cite{ColcombetL07} for defining functions on relational structures. They extend the usual monadic second-order (MSO) transductions (\cf\eg\cite{DBLP:conf/icla/Filiot15}) with the use of set variables, making the output size grow up to exponentially in the input size. A basic feature of "MSO set interpretations" is that the inverse image of a first-order definable predicate on outputs is an MSO-definable (\ie regular) predicate on inputs. This feature is closely connected to the decidability of the first-order theory of automatic structures (\cf the survey~\cite{Gradel20}).
In the \emph{string-to-string} case, we prove a recent conjecture of Filiot, Lhote and Reynier~\cite[Conjecture~2]{Lex}:
\begin{theorem}\label{thm:reg-refl}
  String-to-string "MSO set interpretations" are ""regularity reflecting"":\footnote{We borrow from~\cite{BojanczykN23} this terminology inspired by preservation vs.\ reflection of limits in category theory. We avoid the term \enquote{regularity preserving} which evokes direct images. Some other authors such as~\cite{Pin2005,lmcs:4336} refer to \emph{regular-continuous} functions, alluding to the topology of profinite words.} the inverse image of a regular language is effectively regular.
\end{theorem}

In fact, Filiot \etal have shown~\cite[Proposition~14]{Lex} that this has the following consequence, conjectured by Bárány~\cite[Chapter~5]{BaranyPhD} (see also~\cite[Section 9]{Barany08RAIRO}) nearly two decades ago:

\begin{corollary}\label{cor:automatic-words}
  Every automatic $\omega$-word has a decidable MSO theory.
\end{corollary}

That said, this paper is not about automatic structures. Rather, we focus on the class of string-to-string functions defined by "MSO set interpretations". We call them the \textbf{"expregular functions"}, and we believe that they capture a canonical notion of \enquote{exponential finite-state computation} --- in the same way that the \emph{"polyregular functions"} (surveyed in~\cite{PolyregSurvey}) are considered the finite-state counterpart of the polynomial-time computable functions.
In this work, we provide several arguments that support this intuition, and in particular several equivalent machine models that characterise the class.

\subparagraph*{Context: (poly)regular functions.}

The informal concept of \enquote{finite-state computation} is to generalise the \emph{regular languages} to settings beyond string-to-boolean functions. \emph{Transducer theory} studies such generalisations for functions that may output strings, or even trees. There is a wide variety of transducer models, \ie automata with output; while they are not all equivalent, many of them cluster into a few function classes that can then be considered as \enquote{robust} and \enquote{canonical}.

\AP In the late 1990s, Engelfriet and Hoogeboom~\cite{EngelfrietHoogeboom} proved that \emph{two-way transducers} and \emph{string-to-string MSO transductions} define the same class of functions, which they called the \enquote{regular string transductions}. Many other characterisations were later found, see \eg \cite[\S1]{BojanczykN23} for a non-exhaustive enumeration. By the mid-2010s, the name ""regular function"" had become standard, and the properties of these functions had been thoroughly investigated, \cf \cite{siglog,MuschollPuppis} --- though a few major open problems remained, one of which was settled recently~\cite{theoretics:13747}. Regular functions have \emph{linear growth}: the output size is bounded by a linear function of the input size. Conversely and strikingly, for some superlinear string transducer models, if a function that they compute happens to have linear growth, then it must be regular~\cite{LinearCompositionTWT,StructPoly}.

\AP The \emph{$k$-pebble transducers}~\cite{Pebble,PebbleString,PebbleComposition} extend two-way transducers with up to $O(n^k)$ growth. 
Starting from the late 2010s, Bojańczyk promoted the functions computed by pebble transducers as a canonical notion of polynomial finite-state computation, while introducing the term ""polyregular functions"" and proving a list of several equivalent characterisations~\cite{polyregular,Bojanczyk23}. To this list, Bojańczyk, Kiefer and Lhote~\cite{msoInterpretations} soon added \emph{MSO interpretations}, a formalism that sits in-between MSO transductions and "MSO set interpretations".

This can be summarised as follows:
\begin{center}
\begin{tabular}{r|l|l}
  transductions & $\text{output positions} \subseteq \text{input positions} \times \text{finite set}$&regular\\
  \hline
  interpretations & $\text{output positions} \subseteq \text{finite set} \to \text{input positions}$ & polyregular\\
  \hline
  set interpretations & $\text{output positions} \subseteq \text{input positions} \to \text{finite set}$ & expregular
\end{tabular}
\end{center}
Polyregular functions have been an active topic of study in the past few years, with investigations of interesting subclasses~\cite{NguyenNP21,gaetanPhD} as well as applications to \eg linguistics~\cite{Reduplication} and string compression~\cite{JordonPhD}. It is informally conjectured (\cf\eg\cite[\S1.4.6]{titoPhD}) that for \enquote{reasonable} string transducer models, polynomial growth implies polyregular.
% \nathan{Par qui?}
% \tito{Par moi lol (hélas je pense pas que Mikołaj ou Thomas l'aient écrit quelque part)}

\subparagraph*{Exponential growth.}

Functions with super-polynomial growth rates have been part of transducer theory for half a century. Very simple devices, such as variants of L-systems or top-down tree transducers (see \eg\cite{ERS}), already exhibit exponential growth. However, beyond polynomial growth, the story has not been as satisfying as for (poly)regular functions: there is a proliferation of non-equivalent transducer models. HDT0L-systems are expressively equivalent to several other formalisms~\cite{FerteMarinSenizergues,CopyfulSST,gaetanPhD}, but they are incomparable with polyregular functions~\cite[\S8]{NguyenNP21}. As a common generalisation of HDT0L-systems and polyregular functions, Filiot \etal \cite{Lex} recently introduced the lexicographic string transductions, defined by an extension of pebble transducers or equivalently by a restriction of "MSO set interpretations"; they suspect that this restriction strictly decreases the expressive power~\cite[\S7]{Lex}. 

"MSO set interpretations" thus emerge as a candidate for the largest class of finite-state functions of exponential growth… provided that they are, indeed, \enquote{finite-state} in some sense. In the context of transducer theory, "regularity reflection" is a usual requirement, and this is why we prove \Cref{thm:reg-refl}.

\subparagraph*{Contribution: two new machine models.}

\emph{Hennie machines} are one of Engelfriet and Hoogeboom's equivalent characterisations of regular functions~\cite[\S7]{EngelfrietHoogeboom}. As language recognisers (studied by Hennie~\cite[\S{}II.C]{Hennie65}), they are one-tape Turing machines subject to a bounded-visit restriction\footnote{Hennie's main result is that linear-time one-tape Turing machines have the bounded-visit property.}. As string transducers, they produce their output from left to right, just like two-way and pebble transducers.
Our two machine models for "expregular functions" are two orthogonal %\footnote{One could consider \enquote{yield-Ariadne-transducers} which would have doubly exponential growth.}
 extensions of Hennie machines:
\begin{itemize}
  \item \reintro{yield-Hennie machines} allow the output to be produced in a \emph{branching} fashion: they can spawn parallel subprocesses that are in charge of producing different factors of the output. The name comes from the yield operation~\cite[\S2]{ERS} that reads the string formed by the leaves of a tree --- in this case, the tree of forking processes.
  In a recent article inspired by this work, Dartois, Nguy\~{\^{e}}n, and Peyrat studied the tree-to-tree version of the same model, for which they proved a translation towards higher-order transducers, establishing that they are regularity reflecting~\cite[Lemma 7.7 and Corollary 8.1]{DTP26}.
  \item \reintro{Ariadne transducers}, in contrast, produce their output sequentially but relax the bounded-visit condition by adding some structure to the memory.
  They easily simulate other transducer models with sequential output production, such as pebble transducers (possibly extended with invisible pebbles~\cite{InvisiblePebbles}) or (nested) marble transducers~\cite{gaetanPhD,Lex}. 
\end{itemize}

\subparagraph*{Contribution: from MSO set interpretations to yield-Hennie machines.}

The main technical contribution of this paper is the compilation from "MSO set interpretations" to "yield-Hennie machines".
Our intricate proof is based on a subtle understanding of MSO definable orders; the techniques we use depart from those used in previous work on polyregular functions, which are based on factorisation theorems in adequate semigroups~\cite{msoInterpretations}.
A high-level overview of our proof is given in Section~\ref{sec:equivalence} and a more detailed one in Section~\ref{sec:overview}.

Combined with the aforementioned result of~\cite{DTP26}, this proves that "MSO set interpretations" are "regularity reflecting" (Theorem~\ref{thm:reg-refl}), and in turns implies decidability of the MSO theory of automatic $\omega$-words (Corollary~\ref{cor:automatic-words}) thanks to~\cite[Proposition~14]{Lex}.

\subparagraph*{Contribution: equivalence of the three models.}

The translation from "yield-Hennie machines" to "Ariadne transducers" is relatively straightforward.
To obtain the remaining translation, from "Ariadne transducers" back to "MSO set interpretations", we prove that "Ariadne automata" (the automaton model underlying "Ariadne transducers") recognise regular languages.
To achieve this, we provide a translation from "Ariadne automata" to "alternating Hennie automata", which are equivalent to one-tape linear time alternating Turing machines, and are known to recognise regular languages thanks to~\cite[Corollary 8.1]{DTP26}.

\subparagraph*{Closure properties.}

Our final argument for the robustness of "expregular functions" is their good composition properties:

\begin{theorem}\label{thm:closure-composition}
  The class of "expregular functions" is closed under:
  \begin{enumerate}
    \item precomposition by "regular functions";
    \item postcomposition by "polyregular functions".
  \end{enumerate}
\end{theorem}
\begin{proof}
  We get these results for free from the literature:
  \begin{enumerate}
    \item Colcombet and Löding~\cite[Proposition~2.4]{ColcombetL07} have observed that "MSO set interpretations" are closed under precomposition by MSO transductions, for arbitrary relational structures.
    \item Filiot \etal \cite[Proposition~11]{Lex} have shown that this postcomposition property is equivalent to "regularity reflection" for "expregular functions", which is our \Cref{thm:reg-refl}. \qedhere
  \end{enumerate}
  % \nathan{un truc que je trouve dommage c'est que on mentionne pas que les ariadne sont naturellement clos par post-composition séquentielle, mais je sais pas trop où l'intégrer}
  % \tito{je remanie la discussion plus haut (\Cref{thm:equivalent-models}) pour mettre l'accent sur le côté left-to-right output, mais je trouve que la postcomp par séq est soit trivialement inférable de là si on connaît bien, soit source de surcharge cognitive si on connaît pas bien}
\end{proof}
\begin{remark}
  "Expregular functions" are not closed under precomposition by polyregular functions. Here is a counterexample (via the monoid isomorphism $\N \cong \{a\}^*$):
  \[ \underbrace{\bigl(n \mapsto 2^{n^2}\bigr)}_{\mathclap{\text{not expregular}}} = \underbrace{\bigl(n \mapsto 2^n\bigr)}_{\text{expregular}} \circ \underbrace{\bigl(n \mapsto n^2\bigr)}_{\text{polyregular}}  \]
  Indeed, on inputs of size $n$, the output size of an expregular function is at most $2^{O(n)}$ --- this is what we mean by exponential growth. This contrasts with the definition of the complexity class $\mathsf{EXPTIME} = \mathsf{TIME}\bigl(2^{n^{O(1)}}\bigr)$, chosen to make it closed under polynomial-time reductions.
  % \nathan{On pourrait dire Eregular puisque la classe $E=\text{TIME}(2^{O(n)})$, mais ça sonne comme irregular :D}
  % \tito{XD}
\end{remark}

%% file: machine-models.tex
% !TEX root =  main-set-interpretations.tex

\subsection{MSO set interpretations}\label{sec:def-msosi}

A ""$\Sigma$-string"" is a finite set of ""positions"" equipped with total order~$<$ and a map from elements to~$\Sigma$.
""Monadic second order logic"" (MSO) on strings is defined in the usual way (see \eg~\cite[Chapter~3]{Thomas97}).
It is convenient to allow the syntax to operate on colourings of the positions by a given nonempty finite set $F$ (\ie maps from positions to $F$).
This extends the usual notion of set variables\footnote{Set variables correspond to colourings of the positions by the $2$-element set.} to ""$F$-monadic variables"" which we denote $\Xf X F$; this does not affect the expressive power of MSO since "$F$-monadic variables" are easily encoded as tuples of set variables (details are given in Section~\ref{sec:msosi-to-hennie}).

\subparagraph*{MSO set interpretations.}\AP
    A ("string"-to-"string") ""MSO set interpretation"" $\Phi$ from $\Sigma^*$ to $\Gamma^*$(over a finite set $F$), where $\Sigma$ and $\Gamma$ are finite alphabets, is given by MSO$[<,\Sigma]$ "formulas" $\phi_\pos(\Xf X F)$, $\phi_<(\Xf X F \Xf Y F)$ and for each $\gamma \in \Gamma$, a formula $\phi_\gamma(\Xf X F)$, where $\Xf X F$ and $\Xf Y F$ are "$F$-variables".
    For each "$\Sigma$-string" $w$, this defines a $[<,\Gamma]$-structure $\Phi(w)$ whose universe is the set of colourings $\Xf A F$ of the "positions" of $w$ with $F$ such that $w \models \phi_\pos(\Xf A F)$ and where $<$ is interpreted as $\phi_<$ and for each $\gamma \in \Gamma$, the associated unary predicate is interpreted as $\phi_\gamma$.
    We assume that for each "$\Sigma$-string" $w$, $\Phi(w)$ defines a "$\Gamma$-string", which means that $<$ is a total order, and the unary predicates associated to the $\gamma \in \Gamma$ partition the "positions".
   % Then the transduction induced by $\Phi$ assigns $\Phi(w)$ to $w$.

% \nathan{une petite remarque sur le fait qu'on peut checker qu'une MSOSI produit bien des strings?}
% \pierre{pas d'avis fort, mais je suis plutôt contre (je préfère gagner un peu de place)}
% \tito{je me range à l'avis de pierre, la majorité gagne}

Given an "MSO set interpretation" $\Phi$, $F$-colourings of the input string that satisfy $\phi_\pos$ are called ""configurations@@msosi""; these correspond to positions in the output string (ordered by $\phi_<$).
We say that a function is ""expregular"" if it is computed by an "MSO set interpretation".

We stress the fact that $\phi_<$ defines the order over the configurations, and not the successor relation.
While being irrelevant for regular functions (the configurations are essentially input positions and therefore the order can be defined as the transitive closure of the successor relation), this distinction becomes crucial for polyregular functions (see~\cite{msoInterpretations}) as well as expregular functions.
Here are some examples.

\begin{example}[Strict prefixes reversed]\label{ex:reverse-prefixes}
    Consider the (polyregular) function from $\{a,b\}^*$ to itself which lists the reverses of the strict prefixes of a string, in order of increasing length.
    Here is a picture of an MSO set interpretation realising this function:

    \mypic{1}
    Here, the configurations are obtained by colouring the input with three colors: black, red and gray, so that there is a unique position $x$ coloured black and a unique position $y$ coloured red such that $y < x$.
    These are ordered according to the lexicographic order over $(x,y)$.
    The output of a configuration is the letter of the input at position $y$.
    
    % Let us give an "MSO set interpretation" $\Phi$ defining this function
    % \footnote{We use syntactic sugar such as $\exists x<y\in X \phi$ instead of $\exists x\exists y\,  (x<y)\wedge X(x)=1\wedge X(y)=1\wedge \phi$.}.
    % Let $F=\set{0,1}$.
    % \begin{itemize}
    %     \item $\phi_\pos(X)=\exists x_1<x_2\in X\, \forall y\in X\, (y=x_1\vee y=x_2)$ %\quad: $X$ contains precisely $2$ positions
        
    %     \item $\phi_<(XY)=\exists x_1<x_2\in X\, \exists y_1<y_2\in Y\, (y_1<y_2)\vee ((y_1=y_2)\wedge x_1>x_2)$.
        
    %     \item $\phi_\sigma(X)=\exists x_1<x_2\in X\, \sigma(x_1)$, for $\sigma\in \Sigma$.

    % \end{itemize}

    % Let us briefly describe how this realizes the $\rprefix$ function.
    % Given a $\Sigma$-string $w$, the positions of $\Phi(w)$ will be interpreted by subsets of positions of $w$ of size $2$.
    % Such a set $\set{i_1,i_2}$ with $i_1<i_2$ represents the $i_1$st letter of $w$ within the prefix stopping at $i_2$.
    % Note that the prefixes are ordered in  incresing length, but positions within a prefix are in reverse order.     
\end{example}

\begin{example}[List of subwords]\label{ex:list-of-subwords}
    Consider the function from $\{a,b,c\}^*$ to itself that lists all subwords of a given string, in lexicographic order of their support.
    Here is a picture:

    \mypic{2}
    In this case, configurations are the colourings with one red position, listed lexicographically with gray < red < black.
    The output corresponds to the input letter from the red position.

    % Let $\Sigma=\set{a,b}$. We define the function $\subwords:\Sigma^*\rightarrow \Sigma^*$, which lists, with multiplicity, the (scattered) subwords of a given string.
    %     $$\subwords(aba)=a\, b\, ba\, a\, aa\, ab\, aba$$

    % Let us give an "MSO set interpretation" $\Phi$ defining this function.
    % Let $F=\set{0,1,2}$.
    % \begin{itemize}
    %     \item $\phi_\pos(X)=\exists x X(x)=2 \wedge \forall y (X(y)=2)\rightarrow y=x$.
        
    %     \item 
    %     $\begin{array}{rl}
    %         \phi_<(XY)= \ \  \exists x X(x)=0\wedge Y(x)\neq 0\wedge \forall y<x X(x)=0 \leftrightarrow Y(x)=0\\
    %          \vee\  \forall x X(x)=0\leftrightarrow Y(x)=0\wedge \exists x<y\, X(x)=2\wedge Y(y)=2
    %     \end{array}$
        
    %     \item $\phi_\sigma(X)=\exists x X(x)=2\wedge \sigma(x)$, for $\sigma\in \Sigma$.

    % \end{itemize}

    % Let us briefly describe how this realizes the $\subwords$ function.
    % Given a $\Sigma$-string $w$, the positions of $\Phi(w)$ are intepreted by $F$-variables where exactly one position has label $2$.
    % The non-zero labels coincide with the subword being considered and the $2$ labelled position corresponds with the position within this subword.
    % The order between subwords is given by the lexicographic order between their support.
    % Let us given below this order for a string of length $3$:

    % $$002<020<021<012<200<201<102<210<120<211<121<112$$

    % $$\sm{002}{aba}<\sm{020}{aba}<\sm{021}{aba}<\sm{012}{aba}<\sm{200}{aba}<\sm{201}{aba}<\sm{102}{aba}<\sm{210}{aba}<\sm{120}{aba}<\sm{211}{aba}<\sm{121}{aba}<\sm{112}{aba}$$
\end{example}

\begin{example}[Distribute]\label{ex:distribute}
    Consider the function from $\{a,b,c,d,e,f,\sharp\}^*$ to itself which distributes $\cdot$ over $\sharp$, for instance $ab \sharp cd$ is mapped to $ac \sharp ad \sharp bc \sharp bd$.
    This function is expregular, as witnessed by the following MSO set interpretation using four colours:

    \mypic{3}

    % A product of sums $(a_1+b_1)(a_2+b_2)$ can be distributed naturally as a sum of products: $(a_1a_2+a_1b_2+b_1a_2+b_1b_2)$.
    % To capture this idea we define the function $\distr:\Sigma^*\rightarrow \Sigma^*$, with $\Sigma=\set{a,b,\sharp}$.
    %     $$\distr(ab\sharp bba\sharp a)=aba\,  aba\, aaa\, bba\, bba\, baa$$

    % Let us give informally a description of the MSO interpretation for this function.
    % As in the previous example, the set of positions of $\Phi(w)$ is given by functions $\pos(w) \rightarrow \set{0,1,2}$. with a single $2$ value.
    % The function should have exactly one non-zero value between each $\sharp$, and all the sharps should have value $0$.
    % The order is defined just as in the $\subwords$ example.
\end{example}

\begin{example}
    Consider the function from $\{a,b\}^*$ to $\{0,1,\sharp\}^*$ which, on input $w$, counts in binary with $|w|$ bits, with the bits whose positions are labelled by $a$ in $w$ stronger than those labelled by $b$.
    For instance, $aba$ is mapped to $000 \sharp 010 \sharp 001 \sharp 011 \sharp 100 \sharp 110 \sharp 101 \sharp 111$.
    This function is easily seen to be expregular.
\end{example}

\begin{example}
    Consider an expregular function $f:(A \times B)^* \to \Gamma^*$, where $B$ is totally ordered.
    Define the function $g : A \times \Gamma^*$ as follows: on input $w_A \in A^*$, list all $f(w_A \otimes w_B)$, where $w_B$ ranges over $B$-strings with the same size as $w_A$ in lexicographic order\footnote{One could also replace the lexicographic order with any MSO definable order over $B$-strings.}, and $\otimes : A^* \times B^* \to (A\times B)^*$ denotes the interleaving operation.
    Then $g$ is expregular.
\end{example}

\subsection{Yield-Hennie machines}\label{sec:def-yield-hennie}
A "yield-Hennie machine" is a one-tape Turing machine endowed with a {branching} mechanism to spawn new processes that are executed in parallel.
It is restricted by a {bounded visit} property meaning that within each such process, each cell of the input tape can only be visited a bounded (\ie~independent of the input) number of times.
Each process eventually outputs a letter, and the overall output of the computation is obtained by concatenating these letters; this corresponds to performing the yield of the tree of processes.
\smallskip

\AP
Let $\triangles ,\notin \Sigma$ be two distinct symbols that we call ""end markers"".
A ""$\Sigma$-string with end markers"" is a "$(\Sigma\cup\{\triangles ,\})$-string" of the form $\triangles w$, where $w$ is a "$\Sigma$-string". 
A ""pointed string"" is a string together with a highlighted position $i$ represented by an underline, \eg $\underline \triangleright ab \triangleleft$ or $ a\underline a b a$.

\subparagraph*{Yield-Hennie machines.}\AP
A string-to-string ""yield-Hennie machine"" $\Hh$ from $\Sigma^*$ to $\Gamma^*$ is a branching $1$-tape bounded visit Turing machine.
It consists of:
    \begin{itemize}
        \itemAP a finite set of ""internal states@@hennie"" $\intro*\stateshennie$ with an ""initial state@@hennie"" $q_0 \in Q$;
        \itemAP a finite ""tape alphabet@@hennie"" $\intro*\tapehennie$ which contains $\Sigma \cup \{\triangleright, \triangleleft\}$; and
        \itemAP a ""transition function@@hennie"" $\intro*\transhennie: \stateshennie \times \tapehennie \to (\Gamma \cup \stateshennie \times \tapehennie \times \{\la,\stay,\ra\})^*$,
    \end{itemize}
   such that $\transhennie(\_,\triangleright) \in (\Gamma \cup \stateshennie \times \{\triangleright\} \times \{\stay,\ra\})^*$ and $\transhennie(\_,\triangleleft) \in (\Gamma \cup \stateshennie \times \{\triangleleft\} \times \{\la,\stay\})^*$.

\AP To every input string $w \in \Sigma^*$, we associate a labelled ordered tree called the ""run@@hennie"" over $w$, which is defined top-down as follows:
    \begin{itemize}
        \item the root of the tree is labelled by $(q_0,\underline \triangleright w \triangleleft)$;
        \item for all nodes  labelled $(q, u)$ where $q$ is a "state@@hennie" and $u$ a "pointed $\tapehennie$-string" whose "pointed" letter is $\theta$ and such that $\transhennie(q,\theta) = \delta_1 \dots \delta_k$, the node has exactly $k$ children, and for all~$i\in\{1,\dots,k\}$, if $\delta_i=\gamma$, the $i$th child is labelled~$\gamma$, and if $\delta_i=(q'_i,\theta'_i, \dir_i) \in \stateshennie \times \tapehennie \times \{\la, \stay, \ra\}$, then the $i$'th child is labelled $(q'_i, u'_i)$ where $u'_i$ is the obtained from $u$ by replacing the "pointed" letter by $\theta'_i$ and moving the point according to $\dir_i$.
%        \item nodes labelled by $\Gamma$ are leaves.
    \end{itemize}
    Note that in the "run@@hennie" over $w$, all "pointed" strings have the same size as $\triangles w$, and therefore we may talk of the "position" of the mark in a uniform manner.
 \AP   We assume the ""bounded visit property@@hennie"": there exists some integer $k$ such that for every input string $w$, on every branch of the "run@@hennie" over $w$ and for every position $i$, there are at most $k$ labels where the pointer is at position $i$.
    This implies that the depth of the "run@@hennie" over $w$ is bounded by $k(|w|+2)$ hence the "run@@hennie" has $2^{O(|w|)}$ nodes.
    \AP The semantics of $\Hh$ is the function $\tau_\Hh: \Sigma^* \to \Gamma^*$ which assigns to $w$ the ""yield"" of its "run@@hennie", \ie the string obtained by concatenating all labels of leaves in $\Gamma$ along a traversal of the leaves in the natural order.

Here are some examples.

\begin{example}[Reverse prefixes]
    Let us describe a "yield-Hennie machine" realising the function $\rprefix$ from \Cref{ex:reverse-prefixes}.
    Here is the "run@@hennie" on input $aabb$ of output $aaabaa$:
    \mypic{5}
    There are two states: nodes in the first row are in state $q_0$ and the other nodes are in state~$q_1$.
    The tape alphabet is $\tapehennie=\Sigma \cup \{\triangles,\}$.
    The transitions are easily inferred from the picture.
\end{example}

\begin{example}[List of subwords]
    Here is a "run@@hennie" illustrating a "yield-Hennie machine" realising the function from \Cref{ex:list-of-subwords}; the 8 leftmost nodes have state $q_0$ and the rest have state $q_1$:
    \mypic{6}
    Here, the tape alphabet is $\{a,b,c\} \times \{\text{black},\text{orange},\text{blue}\}$ (with $\Sigma$ and $\Sigma \times \{\text{black}\}$ identified).
    % To realize this function a  yield-Hennie machine can behave in the following way:
    % at every position it spawns two threads, the first labels the current position by a $0$ and the second by a $1$. The two threads move to the right and repeat, spawning two threads each \etc.

    % When a thread reaches the end of the string, it rewinds to the begining and copies the positions labelled by $1$.
\end{example}

\begin{example}[Distribute]
    The function from Example~\ref{ex:distribute} can be computed by a "yield-Hennie machine" similar to the one above: first we read the word from left to write while spawning new threads which colour the node adequately, then each thread prints its output.
    Here is the beginning of a "run@@hennie"; we omit the information relative to the states:
    \mypic{7}
\end{example}

\subsection{Ariadne transducers}\label{sec:def-ariadne}
An "Ariadne transducer" is a sequential model of computation which extends virtually all pebble transducer variants that have been introduced, such as invisible pebbles, marbles, nested marbles \etc.
A configuration of an "Ariadne transducer" is a "stack" which basically records the whole history of a bounded-visit computation. This configuration can be updated in two ways: by pushing, \ie extending the computation one more step, or by popping, \ie going back in the computation --- together with an inflationary update of the top of the "stack" to avoid cyclic behaviour. Crucially an "Ariadne transducer" has the power to access not just the top of its stack but the whole substack of the history of the current input position.

\subparagraph*{Ariadne transducers.}

\AP An ""Ariadne transducer"" $\intro*\aria$ from $\Sigma^*$ to $\Gamma^*$ consists of:
\begin{itemize}
  \itemAP a partially ordered set of ""states@@ariadne"" $\intro*\statesariadne$, with an ""initial state@@ariadne"" $\init\in \statesariadne$;
  \itemAP a ""local stack bound"" $k\in \N$ --- we write $\statesariadne^{\leqslant k} = \{u \in \statesariadne^{*} \mid |u| \leqslant k\}$;
  \itemAP a partial ""transition function@@ariadne"" of the following type: 
  	\[
		\intro*\transariadne\colon \quad (\Sigma \cup \{\triangleright, \triangleleft\}){\times} \statesariadne^{\leqslant k}\quad  \rightharpoonup\quad  \Gamma^* {\times} (\{\intro*\push_\ra,\reintro*\push_\stay,\reintro*\push_\la\}{\times} \statesariadne \cup \set{\intro*\pop}{\times}\update(\statesariadne)),
	\]
  where $\intro*\update(\statesariadne)$ is the set of partial functions $f:\statesariadne\rightharpoonup \statesariadne$ which are ""inflationary"" (\ie $f(q)>q$ if it is defined).
\end{itemize}

\AP Consider an input $\Sigma$-string $w$.
A ""stack"" of $\aria$ over $w$ is a sequence $s\in (\statesariadne{\times} \set{0,\,\ldots,\, |w|+1})^*$.
\AP 
Given a non-empty "stack" $s$, we define the ""top"" and the ""untop"" of $s$ with $\intro*\ttop(s)=(p,i)$ such that $s=\intro*\untop(s) \cdot (p,i)$.
We say that $p$ is the ""top state@@ariadne"" of $s$ and $i$ is its ""top position@@ariadne"".
\AP Given a position $i\in\set{0,\ldots,|w|{+}1}$, we denote by $\intro*\locview s i\in \statesariadne^*$ the ""local view"" of $s$ at $i$, which is the sequence of top states of prefixes of $s$ with "top position@@ariadne" $i$ (in increasing prefix order).

\AP 
The ""initial stack"" over $w$ is $(\init,0)$.
Let $s$ be a "stack" over $w$ with "top position@@ariadne" $i$, let $\sigma$ be the $i$-th letter of $\triangles w$ and let $\locview s i=p_1\ldots p_\ell$ be the corresponding local view. Suppose that $(u,\action)=\transariadne(\sigma,p_1\ldots p_\ell)$ is defined, which implies $\ell\leqslant k$.
We define the ""successor stack"" $s'$ of $s$ over $w$ as follows:
\begin{itemize}
    \item if $\action=(\push_\ra,q)$ and $i+1\leq |w|+1$, then $s'=s\cdot(q,i+1)$;
    \item if $\action=(\push_\stay,q)$ then $s'=s\cdot(q,i)$;
    \item if $\action=(\push_\la,q)$ and $i-1\geq 0$, then $s'=s\cdot(q,i-1)$;
    \item if $\action=(\pop,f)$ and $\untop(s) = r\cdot(p,j) \neq \varepsilon$, then $s'=r\cdot(f(p),j)$;
    \item otherwise, $s'$ is undefined.
\end{itemize}
Whenever $s'$ is defined, the word $u$ is called the ""production"" of the "stack@@ariadne" $s$ over $w$.

\AP 
The ""run@@ariadne"" of $\aria$ over $w$ is the sequence of "stacks" $s_1,\dots,s_{n}$ such that $s_1$ is the "initial stack",  $s_{m+1}$ is the "successor stack" of $s_m$ for all $m<n$, and $s_n$ has no "successor@@stack".
"Stacks@@ariadne" are preordered lexicographically (by using the partial order over $\statesariadne$ and ignoring the positions).
Note that the transition steps imply that $s'>s$ whenever the successor stack $s'$ of $s$ is defined, and hence a "run@@ariadne" cannot contain the same "stack@@ariadne" twice. 
Note furthermore that "runs@@ariadne" are of size at most $|\statesariadne|^{|w|k}$ which is exponential in $|w|$.
\AP The ""production@@ariadne"" of $\aria$ over $w$ is the concatenation $u_1\cdots u_n$ of "productions@@ariadne" of $s_1,\ldots, s_n$.
The ""semantics@@ariadne"" of $\aria$ is the function 
$\tau_\aria:\Sigma^*\rightarrow \Gamma^*$ which maps a word to the "production@@ariadne" of $\aria$ over it.

Here are some examples.

\begin{example}[List of subwords]
    Here is a "run@@ariadne" of an "Ariadne transducer" with three "states@@ariadne" $0<1<2$, on input $ab$:
    \mypicsmaller{9}
    The transducer works as follows: first we have an initialisation phase (orange) where the first row is filled with $0$'s by using $\push_\ra$.
    Then we alternate between
    \begin{itemize}
        \item in blue: filling the second row with $0$'s by using $\push_\la$;
        \item in green: removing the second row using $\pop$ and producing an output if two $1$'s align;
        \item in black: incrementing the first row (seen as a binary number), by replacing $1$'s with $2$'s using $\pop$, until a zero is reached, replacing that $0$ with a $1$, and filling the rest with $0$'s.
    \end{itemize}
\end{example}

\begin{example}[Evaluating an MSO formula]
    \label{ex:succinct}
    Consider an MSO sentence $\phi$ over $\Sigma$-strings.
    There is an Ariadne transducer $\aria$ from $\Sigma^*$ to $\{1\}^*$ with $O(|\phi|)$ states which, on input $w$, outputs $1$ if $w \models \phi$ and outputs the empty string otherwise.
    The machine $\aria$ proceeds as follows: for every subformula of $\phi$ of the form $\exists X \psi$, we will use a row of the input, labelled $0$ or $1$, which corresponds to a possible subset $X$.
    Then by successive incrementations of the row (as in the previous example), we may test every possible subset.
    Existential quantification over first-order variables is dealt with similarly, and atomic formulas are easily handled.
    In particular, this implies that Ariadne transducers (viewed as acceptors) are \tower more succinct than finite state automata.
    This example is due to~\cite{Lex}, where a nested marble transducer is given for the same function.
\end{example}

%% file: equivalence-of-models.tex
% !TEX root =  main-set-interpretations.tex

\subparagraph*{From MSO set interpretations to yield-Hennie machines.}

This translation is our main technical contribution which is the object of Section~\ref{sec:msosi-to-hennie}; we only give a rough overview.

Consider an "MSO set interpretation" for which we aim to construct a "yield-Hennie machine", and fix an input word.
A node in the run of the machine is responsible for generating an interval of the output, therefore nodes naturally correspond to intervals $I$ of configurations.

The key challenge is to understand how such an interval $I$ interacts with a cut of the input positions into two parts, which we will call a split $s$.
We define a combinatorial notion of \emph{"simplicity"} of an interval $I$ at a split $s$, which captures how many times the machine should traverse the cut in order to generate the configurations from $I$, and show that the simplicities are bounded (Lemma~\ref{lem:all-intervals-are-d-simple}), which relates to the bounded visit property of the machine.

From there, one needs to understand how this notion evolves for an interval $I$ when shifting from a cut to the next one (\eg one position to the right).
It turns out that this can be translated into a decomposition of $I$ into a bounded number of children intervals (Lemma~\ref{lem:children-funnels}) so that the simplicity decreases when a branch visits the same cut twice.

Although these ideas give strong foundations for obtaining a Hennie machine, there is one last major difficulty: we need an encoding of the interval $I$ as a colouring of the input, so that when moving from $I$ to its children, the encoding only varies at the edge of the cut.
This is solved using definable interval bases (Section~\ref{sec:interval-bases}), thanks to which we represent intervals by guessing one of their configurations plus a bounded amount of information, together with the notion of funnels (Section~\ref{sec:funnels}), which provide the wanted encoding.
A more detailed overview is given in Section~\ref{sec:overview}.

\subparagraph*{From yield-Hennie machines to Ariadne transducers.}
This is the most straightforward translation of the article, detailed in Section~\ref{sec:hennie-to-ariadne}.
Intuitively, the "Ariadne transducer" visits the run of the "yield-Hennie machine" along the tree order and produces $\gamma \in \Gamma$ when a $\gamma$-leaf is reached, where going down and up in the tree respectively correspond to pushing and poping.

\subparagraph*{From Ariadne transducers to MSO set interpretations.}
Consider an "Ariadne transducer" $\aria$.
"Stacks" naturally correspond to colourings of the input string by a finite set.
However, it is a non-trivial task to construct a formula deciding which (encodings of) "stacks" are produced by the run.
This amounts to proving that accessible "stacks" form a regular language.

The fact that "Ariadne automata" recognise only regular languages is shown in two steps: first we show, using~\cite{DTP26} that another model of automaton, which we call \emph{"alternating Hennie automata"} recognise regular languages.
Second we give a translation from "Ariadne automata" to "alternating Hennie automata"; roughly speaking, alternation gives us a handle to model "stack" operations via the following two-player interaction.
When reaching stack $s$:
\begin{itemize}
    \item Exists proposes a sequence of states $q_1 < \dots < q_n$ which correspond to those appended to $s$ throughout the run (the sequence is necessarily increasing since these are updated by "inflationary" $\pop$ operations); and
    \item Forall chooses an index $i$ and challenges the fact that from stack $s q_i$, the run will next update $q_i$ into $q_{i+1}$. 
\end{itemize}
Once we have established that the language of reachable configurations of an Ariadne automaton is regular, it is not difficult to translate $\aria$ to a set interpretation: $\phi_\pos$ is the formula defining the reachable configurations, $\phi_<$ is the lexicographic order, and $\phi_\gamma$ is obtained easily from $\phi_\pos$.

% Before explaining how one translates an Ariadne automaton into an alternating Hennie automaton, let us describe how this automaton model operates.

% String-to-tree Hennie machines are a restriction of tree-to-tree Hennie machines introduced in \cite{TITO}, in which the authors prove that these functions are regularity reflecting.
% An alternating Hennie automaton is formally a one-tape alternating Turing machine, with a bounded-visit restriction.
% The run of an alternating Hennie automaton, similar to runs of yield-Hennie machines, is the tree of forked processes where each node is labelled by a boolean combination of the processes and leaves are labelled by $\set{\texttt{true}, \texttt{false}}$.
% The run is considered accepting if the run evaluates to $\texttt{true}$ as a boolean formula.
% Thus an alternating Hennie automaton can be seen as the inverse image of $T_{\texttt{true}}$ by a string-to-tree Hennie machine, where $T_{\texttt{true}}$ is the language of boolean formulas that evaluate to $\texttt{true}$ --- which is a regular tree language.

%% file: discussion.tex
% !TEX root =  main-set-interpretations.tex

\tito{c'est là qu'on met safe vs unsafe~? citer~\cite{SafeHORS} dans ce cas}
\tito{linear logic, linear HORS, linear higher-order transducers}
\pierre{j'ai laissé les commentaires de Tito si tu as qch à dire là dessus Thomas, sinon on peut enlever}

"Expregular functions" form, we argue, the natural and robust notion of ``exponential finite-state computation'' for string-to-string functions, as they admit several logical as well as machine-based characterisations and all the desirable closure properties, including being regularity reflecting.
Moreover, as explained in~\cite[Proposition III.7 (full version)]{Lex}, they generalise the notion of automatic $\omega$-words: an automatic $\omega$-word can be viewed as the output of an MSO set interpretation over the structure $(\N,<)$, which can in turn be represented by a string-to-string "MSO set interpretation" over a unary input alphabet.

%We believe that our main technical proof can be refined when applied to MSO interpretations. In this case the obtained "yield-Hennie machine" seems to produce a run with limited branching (more precisely bounded Strahler number, say $k$). This limited branching allows us to translate the Hennie machine to a $k$-pebble transducer. This approach may thus provide an alternative proof of the result of \cite{msoInterpretations}.

"Polyregular functions" enjoy many characterisations \cite{PolyregSurvey} and one may ask which ones can be extended to "expregular functions".
An appealing characterisation of "polyregular functions" is as the composition closure of regular\footnote{Polyregular function can actually be decomposed into even simpler functions; see \cite[Theorem 3.2]{PolyregSurvey}.} functions and the $\texttt{square}$ function (with $\texttt{square}: abc\mapsto \underline a bc\, a \underline bc\, a b\underline c$). Whether "expregular functions" can be captured by a simple family of functions in such a way would prove very useful: decomposition results indeed often come in useful as they restring one's attention to a finite set of simple cases.

Natural extensions of this work are to consider richer structures such as $\omega$-words, trees and $\omega$-trees. We believe that most of the machinery developed in this article should apply to "MSO set interpretations" from trees to strings. For instance the key notions of "splits" and "funnels" naturally extend to this setting. Additionally the automata models can be extended to tree-walking variants (with regular look-around for tree-to-string yield-Hennie machines).
In contrast, we believe that a generalisation to the tree-to-tree context will likely require new insights.

An interesting line of work concerns questions regarding the growth of the different models: % given a class of functions $\Cc$ with growth $G$ and a function $f$ from a more general class $\Dd$ such that $f$ happens to have growth $G$, 
%given classes $\Cc \subseteq \Dd$ of functions, is it the case that functions from $\Dd$ with small growth belong to $\Cc$?
In~\cite[Theorem~1.5]{StructPoly}, Gallot, Lhote and Nguy\~{\^{e}}n established that "expregular functions" with polynomial growth are "polyregular".
One can also ask whether compositions of "expregular functions" with exponential growth are expregular, or more generally whether higher-order transductions with exponential growth are expregular (this is already open for higher-order transducers of polynomial growth and polyregular functions).
Other results of this nature are known for some models of tree-to-tree transducers \eg\cite[Theorem 6.4]{StructPoly}; we refer to the discussions therein for more details.
%In the same vein, one could explore the relationship between simplicity of intervals (see Section~\ref{sec:overview}) and growth.

New questions also arise about succinct representation of regular languages.
"Alternating Hennie automata" and "Ariadne automata" are \tower more succinct than finite-state automata since an MSO formula (\eg in prenex form) can easily be translated into both automata models\footnote{Both automata models can be translated to one another with elementary complexity.} (see Example~\ref{ex:succinct}).
Surprisingly, it is not clear whether there is an elementary transformation from "alternating Hennie automata" to "MSO" (the result of~\cite{DTP26} gives us a construction in \tower).
Furthermore, the construction from MSO to "alternating Hennie automata" provides a machine with bounded alternation between universal and existential transitions (mimicking the quantifier alternation of the formula), whereas these machines can have a number of alternations which is linear in the length of the input. This argument seems to indicate that these automata models are actually more succinct than MSO.

%% file: msosi-to-hennie.tex
% !TEX root =  main-set-interpretations.tex

This section proves the following theorem which is our main technical contribution.

\begin{theorem}\label{thm:msosi-to-hennie}
    For every "MSO set interpretation", there is effectively a "yield-Hennie machine" realising the same function.
\end{theorem}

\subsection{Detailed proof overview}\label{sec:overview}

Consider an "MSO set interpretation" $\Phi$ and fix an input word $w$.

If we follow a branch of the "run@@Hennie" of a "Hennie machine", we see that to each node one can associate several pieces of informations: a position~$i$ in the input, a direction (whether the last move was left or right in the input), and an interval of the output string~$I$: the factor of the output string that is produced by the yield of the subtree rooted at that node. One can also associate a nesting depth~$d$ which is the maximal number of times over all branches from the node that the position will be traversed.
The key idea in the translation from "MSO set interpretations" to "yield-Hennie machines" is to try to derive (an approximation of) this information directly from the "MSO set interpretation".
This leads to the notion of "simplicity".

\subparagraph*{Simplicity.}
\AP 
For~$i$ a position of the input, let $\rig i$ be the set of "positions" $\leqslant i$, and $\lef i$ the set of positions $\geqslant i$.
We call sets of the form $\rig i$ or $\lef i$ "splits"\footnote{Technically, one requires that the complement of a split be a split, and hence that  the full set of positions is not allowed to be a split. Splits are redefined more formally in Section~\ref{sec:preliminaries-for-the-proof}}.
Given a split $s$, we let $\bar s$ denote its complement, which is also a split.
Recall that positions in the output string correspond to configurations (\ie $F$-colourings $A$ of $w$ such that $w \models \phi_\pos(A)$), ordered by $\phi_<$.
Consider an interval $I$ of configurations.

\AP We define what it is for $I$ to be \reintro*"$d$-simple at $s$" (for $d\geqslant 0$) recursively.
The interval $I$ is "$0$-simple at $s$" if all the "configurations@@msosi" $A\in I$ are equal when restricted to~$s$.
For the general case, an interval $I$ is "$d$-simple at $s$" if there is a set $P \subseteq I$
such that (1) all the configurations in $P$ are equal when restricted to~$s$, and (2) all sub-intervals of $I$ that do not intersect $P$ are "$(d-1)$-simple at $\bar s$".
\AP Call ""simplicity of $I$ at~$s$"" the least~$d$ such that $I$ is "$d$-simple at~$s$".

\begin{example}
	Let us consider the "MSO set interpretation" of \Cref{ex:reverse-prefixes}, applied to an input string of length~$7$. Let us fix the "split"~$s=\lef 4$. 
	Consider the following picture (we omit the input and output letters for simplicity):

	\mypic{4}

	The interval $I_1$ is "$0$-simple at $s$" since all its "configurations@@msosi" are equal when restricted to~$s$, and similarly for $I_2,I_3$ and $I_4$.
	The interval $J$ is "$1$-simple at $\bar s$", by taking $K$ to be the highlighted set of positions, since configurations in $K$ are all equal over $\bar s$, and intervals in between are $0$-simple at $s$.
	Finally, we see that the full set of output positions is "$2$-simple at~$s$" and therefore also "$3$-simple at $\overline s$".

	For this simple example, it is easy to extrapolate this reasoning and establish that for every input string $w$, every interval is $3$-simple at every split. 
\end{example}
A first crucial argument in the proof is that the above phenomenon is always true: for all string-to-string "MSO set interpretation", there exists $d_0$ such that for all input strings, all "splits" $s$, and all interval~$I$ of "configurations@@msosi", $I$ is "$d$-simple at $s$" (\Cref{lem:all-intervals-are-d-simple}).
This proof relies on the facts that the "MSO set interpretation" defines a total order (and not only successor), and that when comparing two configurations it is sufficient to consider boundedly many "MSO-types" over $s$ and $\overline s$ ("compositionality").

\subparagraph*{Structure of the run.}
Consider an interval of positions of the output string $I$ together with a "split" $s$, which corresponds to a potential node of the run of the "yield-Hennie machine" that we are trying to construct.
Now that the notion of "simplicity" is established, we should define the children of the node corresponding to $I$ and $s$.
More precisely, we want to decompose $I$ into a union of a bounded number of children intervals $I_1,\dots,I_\ell$, so that each interval $I_j$ comes together with a split $s_j$ which is obtained from $s$ by moving the cut-position either to the left of to the right.

The crucial feature of this decomposition should be that the bounded-visit property is satisfied, \ie each position is seen a bounded number of times on each branch.
In terms of simplicity, this naturally translates to the condition that if $I^0 \to I^1 \to \dots \to I^n$ defines a sequence of intervals such that for some splits $s^0,\dots,s^n$, it holds for every $k<n$ that $(I^{k+1},s^{k+1})$ is a child of $(I^k,s^k)$, and moreover $s^0$ and $s^n$ have the same cut-position, then the simplicity of $I^{n}$ at $s^n$ should be strictly lower than the simplicity of $I^0$ at $s^0$.

Defining such a decomposition requires an understanding of how simplicities evolve when moving between two successive splits (\eg from $\rig i$ to $\rig{i+1}$, or from $\rig i$ to $\lef i$).
The definition of "children@@funnel" corresponds to Lemma~\ref{lem:children-funnels}.
%However, for reasons that we will now explain, these are not just relatively to a pair $(I,s)$ composed of an interval and a split, but the object which we decompose (and which will eventually correspond to a node in the run of the yield-Hennie machine) is a bit more intricate, and we call it a funnel.

\subparagraph*{Difficulty: locally updatable encoding.} In a "yield-Hennie machine", every node contains a pointed string which corresponds to a colouring $C$ of the input together with a position $i$.
Following the above approach, the obvious choice is for $i$ to encode the current split $s$ (together with a bit of information corresponding to its direction that can be added to $C$), and for $C$ to encode the interval $I=[A,B]$, by writing the two configurations $A$ and $B$.
Unfortunately there is a major issue with this approach: when moving from a node to one of its "children@@funnel", the "yield-Hennie machine" can update the current colouring $C$ only at the pointed position~$i$.
However if $C$ is written as a pair of configurations $A,B$, there is no guarantee that its children can be obtained by such local updates of the colourings defining the configurations.
Overcoming this issue requires introducing some additional infrastructure that we describe now.

\subparagraph*{Definable bases and funnels.}
Defining an encoding of $(I,s)$ as a colouring of the input string in a way that can be updated locally when moving from $(I,s)$ to its "children@@funnel" is achieved by combining two ideas.
\begin{itemize}
    \item \emph{Definable "basis of intervals@basis":} we construct families of intervals $B_s$, one for each split $s$, with the properties that
    \begin{enumerate}[(a)]
        \item every interval which is "$d$-simple at $s$" can be covered by a bounded union of "$d$-simple at $s$" intervals from $B_s$;
        \item\label{item:thin} every configuration belongs to a bounded number of intervals from $B_s$; and
        \item given (the natural encodings of) $I$ and $s$, membership of $I$ in $B_s$ is MSO definable.
    \end{enumerate}
    This is the object of Section~\ref{sec:interval-bases}.
    From now on, we fix such a definable "basis" $(B_s)_{s \text{ split}}$.
    \item \emph{"Funnels":} we represent a node corresponding to interval $I$ and split $s$ by a sequence $I_0,s_0,\dots,I_n,s_n$ of intervals and splits, which we call a "funnel", such that\footnote{We also have a condition on the simplicities which will guarantee the bounded visit property, but we omit it in this overview.}
    \begin{itemize}
        \item $I_0$ is the interval of all positions and $s_0=\rig 0$;
        \item the $I_k$'s are basis intervals \ie $I_k \in B_{s_k}$;
        \item the splits $s_0,s_1,\dots,s_n$ are consecutive;
        \item $s_n$ is $s$ and $I$ is the intersection of the $I_k$'s;
        \item for every $k$, $(I_{k+1},s_{k+1})$ is a child of $(\bigcap_{m \leq k} I_k,s_k)$ as defined above.s
    \end{itemize}
    Intuitively, instead of storing $(I,s)$, we remember the full path from the root node to $(I,s)$, except that it is stored using basis intervals, and $I$ is obtained as an intersection of these intervals.
    This is detailed in Section~\ref{sec:funnels}.
\end{itemize}
Of course it is not at all clear how this circumvents the problem highlighted above; for this we should explain how "funnels" are encoded as colourings of the input string.

\subparagraph*{Encoding of funnels.}
To encode a "funnel" $F=I_0,s_0,\dots,I_n,s_n$ we proceed as follows (details in Section~\ref{sec:encoding-funnels}).
Let $j_0,\dots,j_{n-1}$ be such that for every $k$, it holds that $F_{\leq k+1}=I_0,s_0,\dots,I_{k+1},s_{k+1}$ is the $j_{k}$-th "child@@funnel" of $F_{\leq k}$.
Since the splits are consecutive and satisfy the bounded visit property, the sequence $s_0,\dots,s_n$ can be encoded as a colouring of the input using an adequate notion of "tilings" (see Example~\ref{ex:tiling}).
To encode the "funnel" $F$, we integrate the labels $j_0,\dots,j_{n-1}$ (these are bounded integers) in the above tiling.
Note that this encoding is indeed locally updatable: the encoding of the $j$-th child of $F$ is obtained from the encoding of $F$ by adding one step (labelled by $j$) at the current position.

The crucial idea (stated as Lemma~\ref{lem:definability-funnels}) is that given an encoding $T$ of a "funnel" as above, there exist a formula that recovers the output $I=\bigcap_k I_k$ of the "funnel" as follows.
First we guess a configuration $A$ in $I$, and guess a "labelled tiling" (corresponding to the same sequence of splits) with labels $p_0,\dots,p_n$, such that for every $k$, $I_k$ is the $p_k$-th interval from $B_{s_k}$ containing $A$ (here we rely on property~\ref{item:thin} of definable bases).
Then we verify local consistency of these guesses, and recover $I$ as the intersection of the $I_k$'s.

\subparagraph*{Defining the yield-Hennie machine.}
We are now ready for the final construction.
Here the important insight is that in the previous items, the proofs are phrased in terms of definability.
In particular, given an encoding of a "funnel", the encoding of the next "funnel" is definable.
Therefore to define the transitions from a given node (\ie, a "funnel"), it suffices to have access to the "MSO-type" (defined just below) of sufficiently large "quantifier rank" of the current labelling of the tape (which encodes the "funnel").
In a nutshell, this is achieved\footnote{This trick is stolen from~\cite[Lemma~4]{EngelfrietHoogeboom}.} by labelling the tape with one half of the type, and storing the other half in the current state (where the directions of the halves depend on the direction of the current split).
More details in Section~\ref{sec:funnels-to-hennie}.

\subsection{Preliminaries for the proof}\label{sec:preliminaries-for-the-proof}

\subparagraph*{Preliminaries on MSO.}
As explained in Section~\ref{sec:three-formalisms}, we use a slightly modified version of "MSO" which is more convenient for our needs, where for every nonempty finite set $F$, we have a notion of "$F$-variables" $\Xf X F$, which are instantiated as maps from the universe to $F$.
For example, usual set variables correspond to $2$-variables, where $2$ is the two-element set.
We call these ""monadic variables"". \AP
Therefore the syntax of MSO ""formulas"" over the above signature is given by the following grammar:
\[
    \exists x \phi \quad | \quad \existsf F \Xf X F \phi \quad | \quad \phi \vee \phi \quad | \quad \neg \phi \quad | \quad \Xf X F (x) = f \quad | \quad x < y \quad | \quad \sigma(x),
\]
where $x$ denotes a first order variable, $F$ ranges over nonempty finite sets, $\Xf X F$ denotes an "$F$-variable", $f$ ranges over $F$, and $\sigma$ ranges over $\Sigma$.
We always use lower case letters (\eg $x,y,\dots$) to indicate first order variables, and upper case letters with finite sets as exponents (\eg $\Xf X F, \Xf Y G,\dots$) to indicate second order variables.
We use the standard semantics for MSO, see for instance~\cite[Chapter~3]{Thomas97}.
The next easy lemma states that our version of MSO is syntactically equivalent to the usual one.

\begin{lemma}
    For every "formula" $\phi$, there exists a "formula" $\phi'$ which uses only first order variables and $2$-variables, such that $\phi$ and $\phi'$ are equivalent.
\end{lemma}

\begin{proof}
    For each occurrence of a monadic variable $\Xf X F$ in $\phi$, we encode it in $\phi'$ with a tuple of $2$-variables $(\Xf {X_1} 2,\dots,\Xf {X_n} 2)$ of size $n=\lceil \log |F| \rceil$, by picking an enumeration $f_1,\dots,f_{|F|}$ of $F$ and encoding $\Xf F X(x) = f_j$ as
    \[
        \bigwedge_{i=1}^n \Xf X 2(x)=b_i,
    \]
    where $b_1\dots b_n$ is the binary expansion of $j \in \{1,\dots,|F|\}$.
\end{proof}

\AP The ""quantifier rank"" $\intro*\qr(\phi)$ of a "formula" $\phi$ is defined inductively by setting:
\begin{itemize}
    \item $\qr(\Xf X F (x) \in F)=\qr(x < y)= \qr(a(x)) = 0 $;
    \item $\qr(\neg \phi) = \qr(\phi)$;
    \item $\qr(\phi_1 \vee \phi_2) = \max(\qr(\phi_1), \qr(\phi_2))$;
    \item $\qr(\exists x \phi) = \qr(\phi) +1$; and
    \item $\qr(\existsf F \Xf X F \phi) = \qr(\phi) + |F|$.
\end{itemize}
In particular, a "formula" with "quantifier rank" $q$ uses only "$F$-variables" with $|F| \leq q$.

\AP Consider a tuple of variables $\tup X=(x_1,\dots,x_n,\Xf {X_1} {F_1},\dots, \Xf{X_m}{F_m})$ and an integer $q$.
Given a string $w$, we say that a tuple $\tup A=(a_1,\dots,a_n,\Xf {A_1} {F_1},\dots, \Xf{A_m}{F_m})$ matches $\tup X$ if the $a_i$'s are "positions" in $w$ and the $A_i$'s are maps from "positions" to $F_i$.
The $\tup X$-""type"" of "quantifier rank" $q$ of a tuple $\tup A$ of variables matching $\tup X$ in a string $w$ is the collection of MSO$[<,\Sigma]$ "formulas" $\phi(\tup X)$ of "quantifier rank" $\leq q$ such that $w \models \phi(\tup A)$.

\AP For brevity, since $\tup X$ can be uniquely recovered from $\tup A$, we just call it the "$q$-type" of $\tup A$ in $w$, and sometimes write it $\intro*\typ q {\tup A} w$.
The "$q$-type" of $w$ is defined to be the "$q$-type" of $\tup A$ for $\tup A$ the empty tuple (hence, it is a collection of sentences, \ie "formulas" with no free variables).

It is easy to prove by induction on $q$ that for all $\tup X$ and for fixed alphabet $\Sigma$, up to syntactic equivalence, there are finitely many "formulas" with free variables in $\tup X$~\cite{Libkin2004}, and therefore, there are finitely many $\tup X$-"types" of "quantifier rank" $q$.
The following lemma states a fundamental property of MSO.

\begin{lemma}[""Compositionality"" of MSO on strings]
    For every $q$ and every strings $u,v \in \Sigma^*$, the $q$-"type" of their concatenation depends only on the $q$-"types" of $u$ and of $v$.
\end{lemma}

\subparagraph*{Operating on strings with markers.}
When manipulating a "string with markers" $\triangles{w}$, we number the "positions" from $0$ to $|w|+1$, so that "positions" in $w$ correspond to $1,\dots,|w|$.\AP
We say that an "MSO set interpretation" $\Psi^\borders$ over $F$ ""operates on strings with markers@@msosi"" if its input alphabet is of the form $\Sigma \cup \{\triangleright,\triangleleft\}$ and there are $f_\triangleright,f_\triangleleft \in F$ such that its universe "formula" $\phi_\pos^\borders$ satisfies for every $w \in \Sigma^*$, every map $\Xf A F$ from "positions" of $w$ to $F$ and "position"~$i$ that
\begin{align*}
    \triangles{w} \models \phi_\pos^\borders(\Xf A F) \quad \implies \qquad &\Xf A F(i)=f_{\triangleright} \text{ if and only if $i=0$,}\\
    \text{ and } &\Xf A F (i)=f_{\triangleleft}\text{ if and only if $i=|w|+1$.}
\end{align*}
For technical reasons, it will be more convenient to work with such interpretations.
The next straightforward lemma states that we may assume this without loss of generality.

\begin{lemma}
    For every "set interpretation" $\Psi$ from $\Sigma$ to $\Gamma$, there exists a "set interpretation" $\Psi^\borders$ from $\Sigma \cup \{\triangleright,\triangleleft\}$ to $\Gamma$ which "operates on strings with markers@@msosi" and such that for every input string $w \in \Sigma^*$, $\Psi(w)=\Psi^\borders(\triangles{w})$.
\end{lemma}

\begin{proof}
    We first establish the following claim, that describes an easy syntactic transformations from "formulas" about $w$ to "formulas" about $\triangles{w}$, which ignore the first and last "positions".
    
    \begin{claim}
        Let $\tup X=(x_1,\dots,x_n,\Xf {X_1} {F_1},\dots,\Xf {X_m} {F_m})$ be a tuple of variables, and for each $j \in \{1,\dots,m\}$, fix two elements $f_\triangleright^j,f_\triangleleft^j \in F_j$.
        For every MSO "formula" $\phi(\tup X)$ over $\Sigma$-strings, there exists a MSO "formula" $\phi^\borders(\tup X)$ over $(\Sigma \cup \{\triangles,\})$-strings such that for every $w \in \Sigma^*$, and every tuple $\tup A$ matching $\tup X$, we have
        \[
            {w} \models \phi(\tup A) \quad \iff \quad \triangles{w} \models \phi^\borders(\tup A^\borders),
        \]
        where $\tup A^\borders=(a_1^\borders,\dots,a_n^\borders,\Xf {{A_1^\borders}} {F_1}, \dots, \Xf {{A_m^\borders}} {F_m})$ is obtained from $\tup A=(a_1,\dots,a_n,\Xf {A_1} {F_1}, \dots, \Xf {A_m} {F_m})$ by setting $a_j^\borders=a_j$ (as an element of $\{0,\dots,|w|+1\}$) and setting $\Xf {{A_j^\borders}} {F_j}$ to coincide with $\Xf {A_j} {F_j}$ over $\{1,\dots,|w|\}$, and assign $f_\triangleright^j$ and $f_\triangleleft^j$ respectively to $0$ and $|w|+1$.
    \end{claim}

    \begin{claimproof}
        The proof is by induction on the syntactic tree of $\phi$.
        If $\phi$ is an atomic "formula", \ie it is of the form $\Xf X F(x)=f$, $x<y$ or $\sigma(x)$, then we let $\phi^\borders=\phi$.
        Otherwise, we proceed as follows in each of the cases.
        \begin{itemize}
            \item If $\phi=\neg \psi$. Then we let $\phi^\borders=\neg \psi^\borders$.
            \item If $\phi= \psi_1 \vee \psi_2$. Then we let $\phi^\borders = \psi_1^\borders \vee \psi_2^\borders$.
            \item If $\phi= \Xf \exists F \Xf X F \psi$. Then we let $\phi^\borders = \Xf \exists F \Xf X F \psi^\borders$.
            \item If $\phi=\exists x \psi$. Then we let\footnote{The first "position" $0$, the last "position" $|w|+1$, and the boolean connective $\wedge$ are easily defined using standard constructions.} $\phi^\borders = (x>0) \wedge (x<|w|+1) \wedge \psi^\borders(x)$. \claimqedhere
        \end{itemize}
    \end{claimproof}

    To construct the wanted "MSO set interpretation" $\Psi^\borders$, we proceed as follows.
    For $x \in \{<\} \cup \Gamma$, we let $\phi^\borders_x$ be obtained from $\phi_x$ by applying the claim.
    Finally, we define the universe "formula" by picking $f_\triangleright$ and $f_\triangleleft$ arbitrarily in $F$ and setting
    \[ 
       \psi^\borders_\pos(\Xf X F) = \psi'_\pos(\Xf X F) \wedge \Xf X F(0)=f_\triangleright \wedge \Xf X F(|w|+1) =  f_\triangleleft,
    \]
    where $\psi'_\pos$ is obtained from $\psi_\pos$ by applying the claim.
\end{proof}

Therefore, we assume from now on that we are given an "MSO set interpretation" from $\Sigma \cup \{\triangleright,\triangleleft\}$ to $\gamma$ that "operates on strings with markers@@msosi", and we refer to the "formulas" defining it as $\phi_\pos$, $\phi_<$ and $\phi_\gamma$ for $\gamma \in \Gamma$.

\subparagraph*{Terminology.}
Let $q$ denote the "quantifier rank" of the "MSO set interpretation"; this is defined to be the maximal "quantifier rank" of the "formulas" defining it.
All "formulas" are MSO$[<,\Sigma]$ "formulas", so we just call these "formulas".
Throughout the proof, when we say that something is ""bounded"", we mean \emph{bounded as a function of $q$ and $|F|$}.
Fix a $\Sigma$-"string with markers" $\triangles{w} \in \Sigma^*$; its "positions" are denoted $\{0,\dots,|w|+1\}$.
\AP A \reintro*"configuration@@msosi" $\Xf A F$ is a map from "positions" of the input string to $F$ such that $\triangles{w} \models \phi_\pos(\Xf A F)$.
For readability, we just denote them $A$ instead of $\Xf A F$.
Since the "MSO set interpretation" "operates on strings with markers@@msosi", there are two elements $f_\triangleright,f_\triangleleft$ such that every "configuration@@msosi" $A$ satisfies $A(0)=f_\triangleright$ and $A(1)=f_\triangleleft$.
Configurations correspond to positions in the output string, and these are ordered by $\phi_<$, which we just write as $A < A'$.
Throughout the proof, given two "configurations@@msosi" $A,B$, we let $AB$ denote the map from "positions" to $F^2$ obtained by assigning $i$ to $A(i)B(i)$.

\subparagraph*{Splits.}
\AP A ""split"" is a nonempty convex set of "positions", such that its complement is also nonempty and convex.
Stated differently, it is a set of the form $\{0,\dots,i\}$ (for $i<|w|+1$) or $\{i,\dots,|w|+1\}$ (for $i>0$).
We denote $\rig i=\{0,\dots,i\}$ and $\lef i=\{i,\dots,|w|+1\}$, so that every "split" is given by its "position" $i$ and its direction $\in \{\leftarrow,\rightarrow\}$. 
Therefore all such $\rig i$ and $\lef i$ denote valid "splits" except for $\lef 0$ and $\rig{i+1}$.

\AP
For every "split" $s$, its complement, written $\bar s$, is also a "split" which we call the ""opposite"" "split"; the "opposite" of $\rig i$ is $\lef{i+1}$ and the "opposite" of $\lef i$ is $\rig{i-1}$.
Every "split" $s$, except for $\lef 1$ and $\rig{|w|}$, also has a ""successor"", which is the "split" with the same direction $\dir \in \{\leftarrow,\rightarrow\}$ as $s$, and "position" $i+1$ if $\dir=\ \rightarrow$ and $i-1$ if $\dir=\ \leftarrow$, where $i$ is the "position" of $s$.
Stated differently, the "successor" of $s$ is the minimal "split" strictly containing $s$.

\AP
Given a "split" $s$, we let $(\triangles{w})_s$ denote the restriction of $\triangles{w}$ to "positions" in the "split".
We use the same notations for "configurations@@msosi", \eg $A_{s}$ denotes the restrition of $A$ to $(\triangles{w})_s$.
We say that two "configurations@@msosi" $A$ and $A'$ are $s$-""equivalent"" if $A_s=(A')_s$.

It is convenient to encode "splits" as pairs of first order variables as follows: the "split" $\lef i$ is encoded by the pair $(i,0)$, and the "split" $\rig i$ is encoded as the pair $(i,|w|+1)$.
From now on, we thus allow ourselves to manipulate "splits" as variables in mso "formulas", by identifying a "split" $s$ with its encoding as a pair of "positions" when necessary.

\subsection{Upper bound on the simplicity}

\subparagraph*{Simplicity.}
\AP Given two "configurations@@msosi" $A \leq B$, the ""interval"" between $A$ and $B$, denoted $[ A, B]$, is the set of "configurations@@msosi" $C$ such that $A \leq  C \leq  B$.
Note that "intervals" are assumed to be nonempty.
We naturally encode an "interval" $I=[A,B]$ by the pair of "configurations@@msosi" $A,B$, therefore for convenience when defining "formulas" we will often use $I$ as a shorthand for $AB$, \eg we will allow ourselves to write $\phi(I)$ for $\phi(AB)$.

\begin{definition}\AP
    Given a "split" $s$ and an integer $d$, we say that an "interval" $I$ is $d$-""simple"" at $s$ if there is a "configuration@@msosi" $A_0$ such that for every "interval" $J \subseteq I$, either
    \begin{itemize}
        \item there is a "configuration@@msosi" $A \in J$ which is $s$-"equivalent" to $A_0$; or
        \item $d>0$ and $J$ is $(d-1)$-"simple" at $\bar s$.
    \end{itemize}
    By convention, we also define singleton "intervals" to be $(-1)$-"simple" at every "split".
\end{definition}

We say that a "configuration@@msosi" $A_0$ is a ""witness"" for the $d$-"simplicity" of $I$ if the above holds for $A_0$.
Note that if $d<d'$ and $I$ is $d$-"simple" at $s$, then it also is $d'$-"simple" at $s$.
Likewise, if $I \supseteq I'$ and $I$ is $d$-"simple" at $s$, then so is $I'$.
Note also that every "interval" is $0$-"simple" at $\rig 0$ and at $\lef{|w|+1}$ (this is because every pair of "configuration@@msosi" are $s$-"equivalent" when $s=\rig 0$ or $\lef{|w|+1}$, since all "configurations@@msosi" $A$ satisfy $A(0)=f_\triangleright$ and $A(|w|+1)=f_\triangleleft$).

By an easy induction, for every $d$ there is a "formula" $\phi_{\text{d-simple}}$ such that for every "interval" $I$ and "split" $s$, $\triangles{w} \models \phi_{\text{d-simple}}(I s)$ if and only if $I$ is $d$-"simple" at $s$.
Likewise, there are "formulas" defining whether a given "configuration@@msosi" is a "witness" for the $d$-"simplicity" of a given "interval".

We now establish that the simplicities are bounded.

\begin{lemma}\label{lem:all-intervals-are-d-simple}
    There exists $d_0$ such that for every "split" $s$, all "intervals" are $d_0$-"simple" at $s$.
\end{lemma}

\begin{proof}
    Fix a "split" $s$. \AP
    A ""bad chain"" of length $n$ is a sequence $I_0,  A_0, I_1, \dots, I_{n-1},  A_{n-1}, I_n$ such that the $I_j$'s are "intervals" and the $A_j$'s are "configurations@@msosi", so that
    \begin{itemize}
        \item $I_0 \supseteq \dots \supseteq I_n$;
        \item for every $j$, $A_j \in I_{j}$; and
        \item for every $j<n$ all "configurations@@msosi" in $I_{j+1}$ are neither $s$-"equivalent" nor $\bar s$-"equivalent" to $A_j$.
    \end{itemize}

    We prove the lemma by establishing the two following claims.

    \begin{claim}\label{claim:not-simple-implies-long-bad-chain}
        For every "interval" $I$ which is not $2d$-"simple" at $s$, there exists a "bad chain" of length $d$ that starts with $I$.
    \end{claim}

    \begin{claim}\label{claim:bad-chains-are-short}
        The length of "bad chains" is bounded.
    \end{claim}

    \begin{claimproof}[Proof of Claim~\ref{claim:not-simple-implies-long-bad-chain}]
        We prove the result by induction on $d$; for $d=0$, the sequence restricted to $I$ gives a "bad chain" of length $0$.
        Let $d>0$ and assume the result known for smaller values.
        Let $I$ be an "interval" which is not $2d$-"simple" at $s$ and let $A_0 \in I$.

        Towards a contradiction, assume that for every "interval" $J \subseteq I$, either
        \begin{enumerate}[(a)]
            \item\label{item:simple} $J$ is $(2d-2)$-"simple" at $s$;
            \item\label{item:occcurrence_of_s} there is $A \in J$ which is $s$-"equivalent" to $A_0$; or
            \item\label{item:occurrence_of_sbar} there is $A \in J$ which is $\bar s$-"equivalent" to $A_0$.
        \end{enumerate}

        Consider an "interval" $J$ which does not satisfy item~(\ref{item:occcurrence_of_s}).
        Then every "interval" $J' \subset J$ either satisfies item~(\ref{item:simple}) or item~(\ref{item:occurrence_of_sbar}), which means that $J$ is $(2d-1)$-"simple" at $\bar s$.
        Therefore $J$ is $2d$-"simple" at $s$, a contradiction.

        Hence there is an "interval" $J \subseteq I$ which is not $(2d-2)$-"simple" at $s$, and is such that every $A \in J$ is neither $s$-"equivalent" nor $\bar s$-"equivalent" to $A_0$.
        The induction hypothesis gives a "bad chain" of length $d-1$ starting from $J$, and therefore extending this chain on the left with $I, A_0$ gives a "bad chain" of length $d$.
    \end{claimproof}

    \begin{claimproof}[Proof of Claim~\ref{claim:bad-chains-are-short}]
        Consider a "bad chain".
        Assuming it is long enough, by Ramsey's theorem, it contains three "configurations@@msosi" $A_1,  A_2,  A_3$ in this order (in the "bad chain"), such that both $\typ q {(A_j  A_{j'})_s} {{\triangles{w}}_s}$ and $\typ q {(A_j  A_{j'})_{\bar s}} {{\triangles{w}}_{\bar s}}$ are independent on the choice of $j<j'$ for $j,j' \in \{1,2,3\}$.
        In particular, by "compositionality" this implies that if for every $j,j'$ (not necessarily ordered as $j<j'$) we define $A_{j,j'}$ by $(A_{j,j'})_s=(A_j)_s$ and $(A_{j,j'})_{\bar s}=(A_{j'})_{\bar s}$,
        \begin{itemize}
            \item for every $j,j'$ it holds that $A_{j,j'}$ is a "configuration@@msosi";
            \item for every $j<j'$, every $k<k'$ and every $\ell,\ell'$, it holds that
            \[
                 A_{\ell,j} < A_{\ell,k} \iff  A_{\ell',j'} < A_{\ell',k'} \quad \text{ and } \quad  A_{j,\ell} < A_{k,\ell} \iff  A_{j',\ell'} < A_{k',\ell'}; \text{ and}
            \]
            \item for every $j<j'$ and every $k<k'$, it holds that
            \[
                 A_{j,k} < A_{k,j} \quad \iff \quad  A_{j',k'} < A_{k',j'}.
            \]
        \end{itemize}
        For example, if we have $A_{1,1}< A_{1,2}$ as well as $A_{1,1}< A_{2,1}$ and $A_{2,1}< A_{1,2}$, then the 3x3-grid is visited in the following lexicographic order:
        \begin{center}
            \begin{tikzpicture}[scale=1.2]
            \node (c11) at (0,0) {$A_{1,1}$}; \node (c21) at (1,0) {$A_{2,1}$}; \node (c31) at (2,0) {$A_{3,1}$};
            \node (c12) at (0,1) {$A_{1,2}$}; \node (c22) at (1,1) {$A_{2,2}$}; \node (c32) at (2,1) {$A_{3,2}$};
            \node (c13) at (0,2) {$A_{1,3}$};	\node (c23) at (1,2) {$A_{2,3}$}; \node (c33) at (2,2) {$A_{3,3}$};
            \draw[->,dotted,very thick] (c11) -- (c21) -- (c31) -- (c12) -- (c22) -- (c32) -- (c13) -- (c23) -- (c33);
            \end{tikzpicture}
        \end{center}

        According to the possible orderings of $A_{1,1},  A_{1,2}$ and $A_{2,1}$, eight different orderings of the 3x3-grid may occur, which correspond to the different possible lexicographic orderings.
        In the picture below, where the "interval" between $A_2$ and $A_3$ is marked in pink and $A_{j,k}$ is marked by a circle, we argue that in each of these situations there exist $i,j$ such that $A_{j,k}$ is between $A_{2}$ and $A_{3}$ (\ie $A_2 < A_{j,k} < A_3$ or $A_3 < A_{j,k} < A_2$), and either $j=1$ or $k=1$.
        \begin{center}
            \begin{tikzpicture}[scale=.6]
            \draw[pink,very thick] (0,2) circle (.2);
            \draw[pink,ultra thick] (1,1) -- (2,1) -- (0,2) -- (2,2);
            \draw[->,dotted,thick] (0,0) -- (2,0) -- (0,1) -- (2,1) -- (0,2) -- (2,2);
            \end{tikzpicture}\quad
            \begin{tikzpicture}[scale=.6]
            \draw[pink,very thick] (2,0) circle (.2);
            \draw[pink,ultra thick] (1,1) -- (1,2) -- (2,0) -- (2,2);
            \draw[->,dotted,thick] (0,0) -- (0,2) -- (1,0) -- (1,2) -- (2,0) -- (2,2);
            \end{tikzpicture}\quad
            \begin{tikzpicture}[scale=.6]
            \draw[pink,very thick] (0,1) circle (.2);
            \draw[pink,ultra thick] (1,1) -- (0,1) -- (2,2);
            \draw[->,dotted,thick] (2,0) -- (0,0) -- (2,1) -- (0,1) -- (2,2) -- (0,2);
            \end{tikzpicture}\quad
            \begin{tikzpicture}[scale=.6]
            \draw[pink,very thick] (1,0) circle (.2);
            \draw[pink,ultra thick] (2,2) -- (1,0) -- (1,1); 
            \draw[->,dotted,thick] (2,0) -- (2,2) -- (1,0) -- (1,2) -- (0,0) -- (0,2);
            \end{tikzpicture}\quad
            \begin{tikzpicture}[scale=.6]
            \draw[pink,very thick] (0,1) circle (.2);
            \draw[pink,ultra thick] (2,2) -- (0,1) -- (1,1);
            \draw[->,dotted,thick] (0,2) -- (2,2) -- (0,1) -- (2,1) -- (0,0) -- (2,0);
            \end{tikzpicture}\quad
            \begin{tikzpicture}[scale=.6]
            \draw[pink,very thick] (1,0) circle (.2);
            \draw[pink,ultra thick] (1,1) -- (1,0) -- (2,2);
            \draw[->,dotted,thick] (0,2) -- (0,0) -- (1,2) -- (1,0) -- (2,2) -- (2,0);
            \end{tikzpicture}\quad
            \begin{tikzpicture}[scale=.6]
            \draw[pink,very thick] (0,2) circle (.2);
            \draw[pink,ultra thick] (2,2) -- (0,2) -- (2,1) -- (1,1);
            \draw[->,dotted,thick] (2,2) -- (0,2) -- (2,1) -- (0,1) -- (2,0) -- (0,0);
            \end{tikzpicture}\quad
            \begin{tikzpicture}[scale=.6]
            \draw[pink,very thick] (2,0) circle (.2);
            \draw[pink,ultra thick] (2,2) -- (2,0) -- (1,2) -- (1,1);
            \draw[->,dotted,thick] (2,2) -- (2,0) -- (1,2) -- (1,0) -- (0,2) -- (0,0);
            \end{tikzpicture}\quad
        \end{center}
        Now since $A_2$ and $A_3$ belong to the "interval" $I$ just after $A_1$ in the "bad chain", so does $A_{j,k}$.
        However, since $j=1$ or $k=1$, $A_{j,k}$ is either $s$-"equivalent" or $\bar s$-"equivalent" to $A_1$, which contradicts the definition of "bad chains".
    \end{claimproof}
    The lemma immediately follows from Claims~\ref{claim:not-simple-implies-long-bad-chain} and~\ref{claim:bad-chains-are-short}.
\end{proof}

\AP 
Recall that the \reintro*"simplicity" of an "interval" $I$ at a "split" $s$ is the minimal $d$ such that $I$ is $d$-"simple" at $s$.
From now on, we let $d_0$ be the bound obtained from the above lemma, so that all "simplicities" belong to $\{-1,0,1,\dots,d_0\}$.

\subsection{Existence of a definable basis}\label{sec:interval-bases}

We now move on to definable interval bases, which provide a tool to decompose an interval in a definable fashion, while preserving its simplicity.

\AP A ""basis"" $B$ is a collection of sets $B_s$ of "intervals", one for each "split" $s$, with the following properties:
\begin{itemize}
    \item for every "split" $s$, every "interval" which is $d$-"simple" at $s$ is contained in a "bounded" union of "intervals" from $B_s$ that are $d$-"simple" at $s$;
    \item for every "split" $s$, every "configuration@@msosi" belongs to a "bounded" number of "intervals" from $B_s$.
\end{itemize}

\AP A "basis" $B$ is ""definable@@basis"" if there is a "formulas" $\phi_\mem$ such that for all "intervals" $I$ and all "splits" $s$, it holds that $\triangles{w} \models \phi_\mem(Is)$ if and only if $I$ belongs to $B_{s}$.

\begin{lemma}\label{lem:existence-of-a-basis}
    There exists a "definable@@basis" "basis".
\end{lemma}

\begin{proof}
    We first focus on constructing a "basis" $B$, then prove that it is "definable@@basis".
    We construct for each "split" $s$ and for each $d \in \{-1,0,1,\dots,d_0\}$, a set of "intervals" $B_s^d$ such that for every "split" $s$ and integer $d \leq d_0$:
    \begin{itemize}
        \item all "intervals" in $B_s^d$ are $d$-"simple" at $s$;
        \item every "interval" which is $d$-"simple" at $s$ is covered by a "bounded" number of "intervals" from $B_s^d$ together with a "bounded" number of $(d-1)$-"simple" "intervals" at $\bar s$;
        \item every "configuration@@msosi" belongs to a "bounded" number of "intervals" from $B_s^d$. 
    \end{itemize}
    
    Taking $B_s$ to be the union of the $B_s^d$'s provides the wanted "basis" $B$.
    For $d=-1$ we take $B_s^d$ to be comprised of all singletons.

    Fix a "split" $s$ and an integer $d \in \{0,1,\dots, d_0\}$.
    From now on and until the end of the proof, $d$-"simple" always means $d$-"simple" at $s$, and $(d-1)$-"simple" always means at $\bar s$, so we no longer mention $s$ and $\bar s$.
    Since the input string $\triangles{w}$ is fixed, we say that two "configurations@@msosi" $A,A'$ have the same $q$-"type" over a "split" $t \in \{s,\bar s\}$ if $A_{t}$ and $A'_t$ have the same $q$-"type" in $(\triangles{w})_t$.

    \AP
    We say that two "configurations@@msosi" are ""nearby"" if the "interval" between them is $(d-1)$-"simple".
    More generally, given an "interval" $I$, define its ""length"", to be the size of the smallest collection of $(d-1)$-"simple" "intervals" that cover $I$.
    
    \AP
    We say that an "interval" $[A,A']$ is ""well-formed"" if $A$ and $A'$ are $s$-"equivalent", and the "interval" is $d$-"simple" with "witness" $A$.
    The $q$-"type" of a "well-formed" "interval" $[A,A']$ over $s$ is the $q$-"type" of $A$ over $s$ (it is the same as the $q$-"type" of $A'$ over $s$ since $A$ and $A'$ are $s$-"equivalent").
    Note that by definition, every $d$-"simple" "interval" is covered by a $(d-1)$-"simple" prefix, a "well-formed" "interval", and a $(d-1)$-"simple" suffix. \AP
    We say that an "interval" is ""maximally"" "well-formed" if it is "well-formed" and not strictly contained in a "well-formed" "interval".
    
    In this proof we will mostly manipulate "maximally" "well-formed" "intervals".
    Given two distinct such "intervals" $I=[A,B]$ and $I'=[A',B']$, we have either $[A<A'$ and $B<B']$ or $[A'<A$ and $B'<B]$ by maximality.
    In the first case, we write $I<I'$ and otherwise $I'<I$.
    
    We say that two "maximally" "well-formed" "intervals" are "nearby" if their lower bounds are "nearby" and their upper bounds are "nearby".
    We say that they are related, denoted $I \sim I'$, if there is a sequence $I=I_0,I_1\dots,I_n=I'$ of "maximally" "well-formed" "intervals" such that for every $j$, $I_j$ and $I_{j+1}$ are "nearby".
    This defines an equivalence relation over "maximally" "well-formed" "intervals".
    We say that a $\sim$-class is ""long"" if it contains a "maximally" "well-formed" "interval" of "length" $> 4p+1$, where
    \[
        p = \text{number of $\Xf X {F^2}$-"types" of "quantifier rank" $q+9$.}
    \]
    We will prove the following statement about the structure of "long" classes.

    \begin{claim}\label{claim:long-sim-classes}
        Let $I_0=[A_0,B_0]<\cdots<I_n=[A_n,B_n]$ be a "long" $\sim$-class.
        Then $[A_0,A_n]$ and $[B_0,B_n]$ have "length" at most $2p$, which implies that the class is covered by $I_0$ and $I_n$.
    \end{claim}

    For the sake of better factorisation, we defer the proof for now.
    The above claim leads to the following definition of $B_s^d$: for every "long" $\sim$-class, we take the first and the last "interval" of the class.
    To prove that $B$ defines a "basis", we should prove the two following claims.

    \begin{claim}\label{claim:low-congestion}
        Every "configuration@@msosi" is contained in a "bounded" number of "intervals" from $B_s^d$.
    \end{claim}

    \begin{claim}\label{claim:covering}
        Every $d$-"simple" "interval" is covered by a "bounded" number of "intervals" from $B_s^d$ as well as a "bounded" number of $(d-1)$-"simple" "intervals".
    \end{claim}

    \begin{claimproof}[Proof of Claim~\ref{claim:covering} assuming Claim~\ref{claim:long-sim-classes}]
        Let $I$ be a $d$-"simple" "interval".
        Recall that $I$ is covered by two $(d-1)$-"simple" "intervals" and a "well-formed" one.
        Let $I'$ be a "maximally" "well-formed" "interval" containing that "well-formed" "interval".
        It suffices to cover $I'$.
        There are two cases.
        \begin{itemize}
            \item If $I'$ has "length" $\leq 4p+1$.
            Then $I'$ is covered by $\leq 4p+1$ "intervals" which are $(d-1)$-"simple", which concludes.
            \item If $I'$ has "length" $> 4p+1$.
            Then the $\sim$-class of $I'$ is "long", so it is covered by two "intervals" from $B_d$ thanks to Claim~\ref{claim:long-sim-classes}.\claimqedhere
        \end{itemize}
    \end{claimproof}

    To prove Claim~\ref{claim:long-sim-classes} as well as Claim~\ref{claim:low-congestion}, we will use the following technical statement.

    \begin{claim}\label{claim:technical}
        Let $I_0=[A_0,B_0] < I_1=[A_1,B_1] < I_2=[A_2,B_2]$ be "maximally" "well-formed" "intervals" with a common "configuration@@msosi", such that
        \[
            \typ {q+9} {(A_0A_1)_s} {{(\triangles{w})}_s} = \typ {q+9} {(A_0A_2)_s} {{(\triangles{w})}_s}.
        \]
        Then $A_1$ and $A_2$ are "nearby".
    \end{claim}

    \begin{claimproof}
        Given two "configurations@@msosi" $A,B$, we let $A \prods B$ denote the map from "positions" to $F$ which coincides with $A$ on $s$ and with $B$ on $\bar s$. \AP
        We say that $A$ and $B$ are ""compatible"" if $A \prods B$ is a "configuration@@msosi", and recall that by "compositionality", compatibility of $A$ and $B$ depends only on the $q$-"types" of $A_s$ in ${(\triangles{w})}_s$ and of $B_{\bar s}$ in ${(\triangles{w})}_{\bar s}$.

        Observe that there is a "formula" $\phimwf(\Xf X F \Xf {Y} F \Xf {Y'} F)$ such that for three "configurations@@msosi" $A,B,B'$,
        \[
            \triangles{w} \models \phimwf(ABB') \quad \iff \quad \begin{array}{l} \text{$A$ is "compatible" with both $B$ and $B'$, and} \\ \text{$[A \prods B,A \prods B']$ is "maximally" "well-formed".} \end{array}
        \]
        Note that the satisfaction of the above "formula" depends only on $A_s, B_{\bar s}$ and $B'_{\bar s}$, and that this "formula" can be taken with quantifier-rank less than $q+9$.

        Therefore by "compositionality", for every $A,A'$ with the same $(q+9)$-"type" over $s$, it holds that for every "configurations@@msosi" $B,B'$,
        \[
            \triangles{w} \models \phimwf(ABB')  \quad \iff \quad \triangles{w} \models \phimwf(ABB')\tag*{(1)}\label{eq:1}.
        \]

        Note that the assumption on the "types" implies that $A_1$ and $A_2$ have the same $(q+9)$-"type" over $s$, and since $(A_1)_s=(B_1)_s$ the same is true for $B_1$ and $A_2$.
        By applying~\ref{eq:1} to $AA'BB'=B_1 A_2 A_1 B_1$, we get that since $I_1$ is "maximally" "well-formed" and $B_1$ is compatible with both $A_1$ and $B_1$, then so is $A_2$, and moreover it holds that $I'=[A_2 \prods A_1, A_2 \prods B_1]$ is maximally well-formed. % which we write $I'=[A',B']$.
        Let $C_0$ denote the greatest "configuration@@msosi" in $I_0$ which is $s$-"equivalent" to $A_0$ and $<A_1$, and likewise let $D_0$ be the least "configuration@@msosi" in $I_0$ which is $s$-"equivalent" to $A_0$ and $>A_1$.
        Then $C_0$ and $D_0$ are successive $s$-"equivalent" "configurations@@msosi" which are $s$-"equivalent" to $A_0$, so they are "nearby" (\ie $[C_0,D_0]$ is $(d-1)$-"simple") by well-formedness of~$I_0$.
        %\nat{je suis d'accord mais c'est dense comme paragraphe! --> Pierre: je suis d'accord... tu peux rajouter des explications si tu veux}

        Let $A'=A_2 \prods A_1$.
        Thanks to compositionality, note that given a "configuration@@msosi" $X$ which is $s$-"equivalent" to $A_0$ and a "configuration@@msosi" $Y$, whether or not $X < Y$ depends only on the $q$-"types" of $X Y$ over $\bar s$ and of $A_0 Y$ over $s$.
        Moreover, these "types" are the same when $Y$ is $A_1$ or $A'$: the types of $XA_1$ and $XA'$ over $\bar s$ are the same since $A_1$ and $A'$ are $\bar s$-"equivalent", and the types of $A_0 A_1$ and $A_0 A'$ over $s$ are the same by the assumption of the lemma (because $A'$ is $s$-equivalent to $A_2$).
        Therefore, by instantiating $X$ to $C_0$ or $D_0$, we get that $A'$ is between $C_0$ and $D_0$, just like $A$.
       % \nat{je suis pas convaincu, je pense qu'il manque un argument --> Pierre : J'ai ajouté des détails, je pense que c'est clair maintenant}
        In particular, the "interval" formed by $A_1$ and $A'$ is contained in $[C_0,D_0]$, so $A_1$ and $A'$ are "nearby".
        
        We now prove that in fact $I'=I_2$, therefore $A'=A_2$ which proves the claim.
        Since $I'$ and $I_2$ are both "well-formed", and since $A'$ and $A_2$ are $s$-"equivalent", it suffices to prove that $I'$ and $I_2$ intersect.
        But since the lower bound $A'$ of $I'$ satisfies $A'> C_0 \geq A_0$ and both $I'$ and $I_0$ are "maximally" "well-formed", we have $B'>B_0$.
        Therefore $B_0$ intersects both $I'$ (because $A'< D_0 \leq B_0<B'$) and $I_2$ (by the assumption that $I_0$ and $I_2$ have a "configuration@@msosi" in common).
    \end{claimproof}

    We are now ready to establish Claims~\ref{claim:long-sim-classes} and \ref{claim:low-congestion}.

    \begin{claimproof}[Proof of Claim~\ref{claim:low-congestion}]
        Let $A$ be a "configuration@@msosi".
        Recall that $B_s^d$ contains at most two "intervals" from each $\sim$-class (exactly two for "long" classes, and none for short ones).
        Therefore it suffices to see that $A$ intersects a "bounded" number of $\sim$-classes.
        Let $I_0=[A_0,B_0]<\dots<I_{n+1}=[A_{n+1},B_{n+1}]$ be "maximally" "well-formed" "intervals" from pairwise different $\sim$-classes, all containing $A$, and such that $n>p^2$.
        Then there are two "intervals" $I_{j_1}<I_{j_2}$ such that
        \begin{align*}
            \typ {q+9} {(A_0A_{j_1})_s} {{\triangles{w}}_s} = \typ {q+9} {(A_0A_{j_2})_s} {{\triangles{w}}_s} \text{ and } \\ \typ {q+9} {(B_{n+1}B_{j_1})_s} {{\triangles{w}}_s} = \typ {q+9} {(B_{n+1}B_{j_2})_s} {{\triangles{w}}_s}.
        \end{align*}
        By applying Claim~\ref{claim:technical} we get that $A_{j_1}$ and $A_{j_2}$ are "nearby", and by applying its symmetrical version, we get that $B_{j_1}$ and $B_{j_2}$ are "nearby".
        Hence $I_{j_1}$ and $I_{j_2}$ are "nearby", so $I_{j_1} \sim I_{j_2}$, a contradiction.
    \end{claimproof}

    \begin{claimproof}[Proof of Claim~\ref{claim:long-sim-classes}]
        We will prove the following statement.

        \begin{subclaim*}
            Let $I_0=[A_0,B_0] < \cdots < I_k=[A_k,B_k]$ be "maximally" "well-formed" "intervals" such that the $I_{j}$ and $I_{j+1}$ are "nearby" for every $j$, and $I_0$ has "length" $> 2p+1$.
            Then $[A_0,A_k]$ has "length" $\leq 2p$.
        \end{subclaim*}

        First we explain why the subclaim implies Claim~\ref{claim:long-sim-classes}.
        Let $I_0=[A_0,B_0]<\cdots<I_n=[A_n,B_n]$ be a "long" $\sim$-class, and let $\midd \in \{0,\dots,n\}$ be such that $I_\midd$ has "length" $> 4p+1$.
        We first apply the subclaim for $I_\midd < \cdots < I_n$, which tells us that $[A_\midd,A_n]$ has "length" $\leq 2p$ and therefore $I_n$ has "length" $> 2p+1$.
        Then we apply the symmetrical version of the subclaim to $I_n > \cdots > I_0$, which tells us that $[B_0,B_n]$ has "length" $\leq 2p$.
        In particular $[B_0,B_\midd]$ has "length" $\leq 2p$ so $I_0$ has "length" $>2p+1$.
        Finally, we apply the subclaim to $I_0 < \cdots < I_n$, which tells us that $[A_0,A_n]$ has "length" $\leq 2p$, which concludes.
        Hence there remains to prove the subclaim.

        \vskip1em
        \item[\hskip\labelsep
        \color{lipicsGray}\sffamily
    Proof of the subclaim.]
     
    For $j \in \{1,\dots,k\}$, we say that $\typ {q+9} {(A_0 A_j)_s} {{\triangles{w}}_s}$ is the colour of $I_j$; we know that there are $p$ possible colours.
    Observe that if $I_j$ and $I_{j'}$ have the same colour and are such that the lengths of $[A_0,A_j]$ and $[A_0,A_{j'}]$ are $\leq 2p+1$, then they have a "configuration@@msosi" in common with $I_0$ (because $I_0$ has "length" $>2p+1$) and thus $A_j$ and $A_{j'}$ are "nearby" (thanks to Claim~\ref{claim:technical}).

    We prove the result by induction on $k$.
    If $k\leq 2p$ then the result clearly holds: since $A_j$s and $A_{j+1}$ are "nearby", the "length" of $[A_j,A_{j+1}]$ is $1$.
    Assume that $k>2p$ and that the result holds for smaller values of $k$.
    Then there are two odd indices $j<j'$ such that $I_j$ and $I_{j'}$ have the same colour and $j' \leq 2p+1$.
    Since $j$ and $j'$ are $\leq 2p+1$ it follows that the lengths of $[A_0,A_j]$ and $[A_0,A_{j'}]$ are $\leq 2p+1$, therefore by the above observation, $A_j$ and $A_{j'}$ are "nearby".
    Hence the sequence 
    \[
        I_0<I_1<\cdots<I_{j-1}<I_j<I_{j'}<I_{j'+1}<\cdots I_k
    \]
    satisfies the assumption of the lemma, and is strictly smaller, since $j$ and $j'$ are odd.
    Therefore $[A_0,A_k]$ has "length" $\leq 2p$, which concludes.
    \end{claimproof}

    Combining the above claims, we know that $B$ is a "basis".
    We now focus on its definability.
    This part presents no particular difficulty; we keep it at a relatively informal level for readability.
    We prove that for every "split", membership is definable for "intervals" which are $d$-"simple" at that "split".
    This implies the result since being $d$-"simple" at a given "split" is definable, and since all "intervals" are $d_0$-"simple".

    For $d=-1$, recall that $B_s^d$ is the set of all singletons, and $(-1)$-"simple" "intervals" are the singletons, so the three items above hold easily.
    Let $d \geq 0$.
    To define membership of an "interval" $I=[A,B]$ in $B_s^d$, we should determine if
    \begin{itemize}
        \item $I$ is "maximally" "well-formed";
        \item $I$ belongs to a "long" $\sim$-class;
        \item $I$ is either the first or the last "interval" of its $\sim$-class.
    \end{itemize}
    We have already established that being "maximally" "well-formed" is definable.
    Thanks to Claim~\ref{claim:long-sim-classes}, we know that if $I,I'$ belong to the same "long" $\sim$-class, then there are "maximally" "well-formed" "intervals" $I=I_0<I_1<\cdots<I_n=I'$ such that for every $j$, $I_{j+1}$ and $I_j$ are "nearby", and $n \leq 2p$.
    Moreover, having "length" $>4p+1$ is definable.
    Therefore, we may guess "intervals" $I_1,\dots,I_n$ as above, and verify that indeed $I_{j+1}$ and $I_j$ are "nearby" for every $j$, and that $I_n$ has "length" $>4p+2$; existence of such "intervals" is equivalent to $I$ belonging to a "long" $\sim$-class, and this is definable since $p$ is "bounded".
    For the last item, we use the same technique as above to define being in the same $\sim$-class (assuming it is "long"), and easily check that $I$ is the first or the last.
\end{proof}

From now on, we consider the definable "basis" $B=(B_s)_s$ constructed in this section, except that we add the "interval" comprised of all "configurations@@msosi" to $B_{\rig 0}$.
This is clearly still a definable "basis".

\begin{lemma}\label{lem:definability-basis-formulas}
    There are "bounded" integers $\ell_\cont$ and $\ell_\cov$ and "formulas" $\phi_\cont^1,\dots,\phi_\cont^{\ell_\cont}$ as well as $\phi_\cov^1,\dots,\phi_\cov^{\ell_\cov}$ such that:
    \begin{itemize}
        \item\emph{sections:} for all "configurations@@msosi" $A$, all "splits" $s$ and all $j \in \{1,\dots,\ell_\cont\}$, $\triangles{w} \models \phi^j_\cont(AI_js)$ holds for a single "interval" $I_j$, and that "interval" belongs to $B_{s}$ and contains $A$, so that every "interval" in $B_{s}$ containing $A$ is of the form $I_j$ for some $j$; and
        \item\emph{covering:} for all "intervals" $I$, all "splits" $s$ and all $j \in \{1,\dots,\ell_\cov\}$, $\triangles{w} \models \phi^j_\cov(II_js)$ holds for a single "interval" $I_j$ such that
        \begin{itemize}
            \item for every $j$, $I_j$ belongs to $B_{s}$;
            \item for every $j$, $I_j$ intersects $I$;
            \item for every $j$, the "simplicity" of $I_j$ at $s$ is smaller or equal to the one of $I$; and
            \item $\bigcup_j I_j$ contains~$I$.
        \end{itemize}
    \end{itemize}
\end{lemma}

\begin{proof}
    We first explain how to define the section "formulas", and then the covering ones.
    
    \vskip1em
    \noindent
    \textbf{Sections.}
    Let $\ell_\cont$ be a "bounded" upper bound to the size of the sections.
    Then for every $j\in \{1,\dots,\ell_\cont\}$ we let $\phi_\cont^j(AI_js)$ be the "formula" expressing that either
    \begin{itemize}
        \item there are $>j$ different "intervals" in $B_s$ which contain $A$, and $I_j$ is the $j$-th one lexicographically; or
        \item there are $\leq j$ different "intervals" in $B_s$ which contain $A$, and $I_j$ is the last one lexicographically.
    \end{itemize}
    This is definable because membership in $B_s$ is definable (as well as containing a "configuration@@msosi").

    \vskip1em
    \noindent
    \textbf{Covering.}
    Let $\ell_\cov$ be a "bounded" integer such that for every $d$ and for every "split" $s$, every "interval" which is $d$-"simple" at $s$ is covered by a union of $\ell_\cov$ "intervals" from $B_s$.
    Then we define $\phi_\cov^j(I I_j s)$ by expressing that for $d$ being the "simplicity" of $I$ at $s$, $I_j$ is the $j$-th "interval" in the lexicographically smallest $\ell_\cov$-tuple of "intervals" which are $d$-"simple" at $s$ and cover $I$.
\end{proof}

\subsection{Funnels and their children}\label{sec:funnels}

With a definable basis in hands, we are now ready to introduce funnels, which are meant to encode nodes of the sought yield-Hennie machine, by using (intersections of) intervals from the basis and storing all the information from the root.
Recall that all "splits" have successors except $\lef 1$ and $\rig {|w|}$.

\begin{definition} \AP
    A ""funnel"" is a sequence $I_0,s_0,I_1,s_1,\dots,I_n,s_n$ of "intervals" and "splits", such that
    \begin{enumerate}[1.]
        \item $I_0$ is the "interval" of all "configurations@@msosi";
        \item $s_0$ is $\rig 0$;
        \item\label{item:output-nonempty} $\bigcap_{k} I_k$ is nonempty;
        \item\label{item:interval-from-basis} for every $k$ it holds that $I_k \in B_{s_k}$;
        \item\label{item:directions-of-successors} for every $k$ it holds that $s_{k+1}$ is either $s_k$, the "opposite" of $s_k$ or the "successor" of $s_k$;% (assuming it is defined);
        \item\label{item:simplicity-decreases} for every $k<k'$ such that $s_{k'}$ coincides with $s_{k}$ or its "opposite", the "simplicity" of $I_{k'}$ at $s_{k'}$ is strictly less than the "simplicity" of $I_k$ at $s_k$.
    \end{enumerate}
\end{definition}

%Note that thanks to item~\ref{item:simplicity-decreases}, for every "funnel" $I_0,s_0,\dots,I_n,s_n$ and every "position" $i \in \{0,\dots,|w|+1\}$, there are at most $2(d_0+2)$ indices $k$ such that $s_k \in \{\lef i, \rig i\}$.
The ""output"" $I$ of a "funnel" $I_0,s_0,\dots,I_n,s_n$ is the intersection of the $I_k$'s, which is assumed to be nonempty by item~\ref{item:output-nonempty}.
We say that a "funnel" is ""closed"" if $I$ is a singleton, and that it is ""open"" otherwise.
The next statement provides a notion of children funnels, which will be the foundation for the tree structure of the "yield-Hennie machine".

\begin{lemma}[Definable "children@@funnel"]\label{lem:children-funnels}
    There is a "bounded" integer $\ell$ and for every $j \in \{1,\dots,\ell\}$ "formulas" $\phi_\childint^j,\phi_\childsplsame^j,\phi_\childsplopp^j, \phi_\childsplsucc^j$ such that for every "open" "funnel" $I_0,s_0,\dots,I_n,s_n$ with "output" $I$, the following hold:
    \begin{itemize}
        \item for every $j$, there is a unique "interval" $I^j$ such that $\triangles{w} \models \phi_\childint^j(I I_j s)$;
        \item for every $j$, $\triangles{w}$ models exactly one of the $\phi_{\childsplx}^j(Is)$ for $x \in \{\same,\opp,\suc\}$ and we let $s^j$ denote either the same, the "opposite", or the "successor" of $s_n$, according to $x$;
        \item for every $j$, $I_0,s_0,\dots,I_n,s_n,I^j,s^j$ defines a "funnel";
        \item $I$ is contained in $\bigcup_j I^j$; and
        \item if $I_j=[A_j,B_j]$ and $I_{j'}=[A_{j'},B_{j'}]$ are so that $A_j < A_{j'}$ then $j<j'$.
    \end{itemize}
\end{lemma}

We say that the "funnels" $I_0,s_0,\dots,I_n,s_n,I^j,s^j$ defined by the above lemma are the ""children@@funnel"" of $I_0,s_0,\dots,I_n,s_n$.

\begin{proof}
    We start by showing the following claim, which will be the foundation for defining the "children@@funnel" (below, ``inter'' stands for ``intermediate'').
    
    \begin{claim}\label{claim:splitting-intervals}
        There is a "bounded" integer $\ell$ such that for every $p \in \{1,2\}$ and every $d,d' \in \{0,\dots,d_0\}$, there exist, for every $j \in \{1,\dots,\ell\}$, "formulas" $\phi_{\dirsucc(p,d,d')}^j$ and $\phi_{\dirsuccsplprime(p,d,d')}$ such that for every "position" $i \leq |w|-p+1$, if we let $s=\rig i$ and $s'=\lef{i+p}$, for every "interval" $I$ which is $d$-"simple" at $s$ and $d'$-"simple" at $s'$, the following holds:
        \begin{itemize}
            \item for every $j$, there is a unique "interval" $I^j$ such that $\triangles{w} \models \phi_{\dirsucc(p,d,d')}^j(I I^j i)$;
            \item for every $j$, if $\triangles{w} \models \phi_{\dirsuccsplprime(p,d,d')}^j(I i)$ then $I^j$ is $(d'-1)$-"simple" at $\bar s'$, otherwise $I^j$ is $(d-1)$-"simple" at $\bar s$; and
            \item $I = \bigcup_j I^j$.
        \end{itemize}
    \end{claim}
    
    \begin{claimproof}
        Let $A$ and $A'$ be the least "configurations@@msosi" respectively witnessing that $I$ is $d$-"simple" at $s$ and $d'$-"simple" at $s'$; note that these are definable.
        Let $A_1 < A'_1 < A_2 < A'_2 < \cdots < A_n < A'_n$ be the chain of "configurations@@msosi" constructed by induction as follows:
        \begin{itemize}
            \item $A_1$ is the first "configuration@@msosi" which is $s$-"equivalent" to $A$;
            \item for every $j$, $A'_j$ is the first "configuration@@msosi" $> A_j$ which is $s'$-"equivalent" to $A'$; and
            \item for every $j>1$, $A_j$ is the first "configuration@@msosi" $> A'_{j-1}$ which is $s$-"equivalent" to $A$.
        \end{itemize}
    
        Note that the $(A_j)$'s are all $s$-"equivalent" and the $A'_{j}$'s are $s'$-"equivalent".
        This implies the following result.
    
        \begin{subclaim*}
            It holds that $n$ is "bounded".
        \end{subclaim*}

        \vskip1em
        \item[\hskip\labelsep
        \color{lipicsGray}\sffamily
    Proof of the subclaim.]
    
            Assume that $n$ is strictly greater than 
            \begin{center}
                (number of $\Xf X {F^2}$-"types" of "quantifier rank"~$q$) $\times \ |F|^{2(p-1)}$. 
            \end{center}
            Let $\midd$ be the empty set if $p=1$ and the singleton position $\{i+1\}$ if $p=2$, so that we have $s \ \cup \ \midd \ \cup \ s' = \{0,1,\dots,|w|+1\}$, the set of all positions.
            By the pigeonhole principle, there exists two indices $j<j'$ such that
            \begin{align*}
                \typ q {(A_j  A'_j)_{s'}} {{\triangles{w}}_{s'}} = \typ q {(A_{j'}  A'_{j'})_{s'}} {{\triangles{w}}_{s'}}\quad \text{ and } \quad (A_j A'_j)_\midd = (A_{j'} A'_{j'})_\midd.
            \end{align*}
            Since $A'_j$ and $A'_{j'}$ are $s'$-"equivalent", this implies
            \begin{align*}%\tag{1}\label{eq:1}
                \typ q {(A_j  A'_j)_{s'}} {{\triangles{w}}_{s'}} = \typ q {(A_{j'}  A'_{j})_{s'}} {{\triangles{w}}_{s'}},
            \end{align*}
            and since moreover, $A_j$ and $A_{j'}$ are $s$-"equivalent", we also have
            \begin{align*}%\tag{2}\label{eq:2}
                \typ q {(A_j A'_j)_{s}} {{\triangles{w}}_s}= \typ q {(A_{j'} A'_{j})_{s}} {{\triangles{w}}_s}.
            \end{align*}
            Finally, since $(A_j A'_j)_\midd = (A_{j'} A'_{j'})_\midd$ we have
            \begin{align*}%\tag{3}\label{eq:3}
                (A_j A'_j)_\midd = (A_{j'} A'_j)_\midd.
            \end{align*}
            Thanks to the "compositionality" lemma applied to the alphabet $\Sigma \times F^2$, the three above equations give %equations~\eqref{eq:1},~\eqref{eq:2} and~\eqref{eq:3} give
            \[
                \typ q {A_j A'_j} {w} = \typ q {A_{j'} A'_j} w.
            \]
            Therefore, since $\phi_< \in \typ q {A_j  A'_j} w$, we have $\phi_< \in \typ q {A_{j'}  A'_j} w$ \ie $A_{j'} < A'_j$, a contradiction.
            Hence the subclaim is proved.
            \vskip1em
    
        Note that since $n$ is "bounded", the $A_j$'s and $A'_j$'s are definable.
        Let $I^1,I'^1,\dots,I'^n,I^{n+1} \subseteq I$ be the maximal "intervals" $\subseteq I$ such that 
        \[
            I^1 < A_1 < I'^1 < A'_1 < I^2 < \cdots < I^n < A_n < I'^{n} < A'_n < I^{n+1},
        \]
        where $I<A$ means that all "configurations@@msosi" in $I$ are $<A$, and likewise for $A<I$.
        (Clearly the $I^j$'s and $I'^j$'s are definable.)
        Then by construction of the sequence $A_1, A'_1,\dots$, it holds that for every $j$, the "configurations@@msosi" in $I^{j}$ are not $s$-"equivalent" to $A$.
        Assuming $d>0$, since $A$ "witnesses" that $I$ is $d$-"simple" at $s$, this implies that for all $j$ it holds that $I^j$ is $(d-1)$-"simple" at $\bar s$.
        If $d=0$, which means that all "configurations@@msosi" in $I$ are $s$-"equivalent" to $A$, by construction of the sequence the $I^j$'s are empty and hence can be removed.
        An analogous proof holds for the $I'^j$'s, which concludes the proof of the claim.
    \end{claimproof}

    Let $I_0,s_0,\dots,I_n,s_n$ be an "open" "funnel" with "output" $I=\bigcap_k I_k$.
    To define its "children@@funnel", we will proceed in two steps.
    First, we will apply Claim~\ref{claim:splitting-intervals}, either to $s_n$ and its "opposite" (\ie with $p=1$) or to $s_n$ and the "opposite" of its "successor" (\ie $p=2$; in this case, we will need to verify that the "successor" of $s_n$ is indeed defined), depending on the "simplicities" of $I$ at $s_n$ and its "opposite".
    This will give us (definable) "intervals" $I^1,\dots,I^\ell$ which cover $I$, and we will prove that their "simplicities" are low enough so that item~\ref{item:simplicity-decreases} from the definition of "funnels" is satisfied.
    The remaining issue is that the $I^j$'s are not in the "basis", so the second step will be simply to decompose each $I^j$ into a bounded number of "basis" "intervals" (for the corresponding "split").
    We will refer to the "intervals" $I^1,\dots,I^\ell$ and the associated "splits" (see below) obtained in the first step as the intermediate "children@@funnel"; formally, these do not define "funnels" since they are not in the "basis".

    \vskip1em
    \noindent
    \textbf{First step: intermediate "children@@funnel".}
    We will need the following easy technical statement.

    \begin{claim}\label{claim:technical-opposite}
        Let $I_0,s_0,I_1,s_1,\dots,I_n,s_n$ be a "funnel", and let $k$ be such that $s_k$ is the "successor" of $s_n$.
        There exists $k' > k$ such that $s_{k'}$ is the "opposite" of $s_k$.
    \end{claim}
        
    \begin{claimproof}
        This follows immediately from item~\ref{item:directions-of-successors} from the definition of "funnels".
    \end{claimproof}

    Let $d_n$ and $\bar d_n$ denote the "simplicities" of $I$ at $s_n$ and at $\bar s_n$ (which are definable).
    We distinguish three cases (where the case "split" is definable).
    In every case, we give a definition of the intermediate "children@@funnel" $I^j,s^j$, argue that there are "formulas" defining them (as in the statement of the lemma), and argue that all items in the definition of a "funnel" hold for $I_0,s_0,\dots,I_n,s_n,I^j,s^j$, except for item~\ref{item:interval-from-basis} about the "basis", which is guaranteed in the second step.
    
    The only difficulty will be to prove item~\ref{item:simplicity-decreases}.
    For this item, since we add an "interval" $I^j$ and a "split" $s^j$ to an existing "funnel", it suffices to prove that the "simplicity" of $I^j$ at $s^j$ is less than the "simplicity" of $I_k$ at $s_k$, where $k$ is maximal such that $s_k$ is either $s^j$ or its "opposite" (if there is no such $k$, then there is nothing to prove, so we assume that $k$ is well defined).
    \begin{itemize}
        \item If $\bar d_n < d_n$.
        Then we set $I^j=I$ and $s^j=\bar s_n$ for all $j$, which are obviously definable.
        Then $k=n$ and the "simplicity" of $I^j$ at $s^j$ is $\bar d_n$, which is strictly less than the "simplicity" $d_n$ of $I$ at $s_n$, and therefore also strictly less than the "simplicity" of $I_n$ at $s_n$ because $I \subseteq I_n$.
        \item If $\bar d_n = d_n$.
        Note that since the "funnel" is "open", $d_n \geq 0$.
        In this case we apply Claim~\ref{claim:splitting-intervals} to $s_n$ and $\bar s_n$ (\ie with $p=1$), which gives "intervals" $I^1,\dots,I^\ell$ and "splits" $s^1,\dots,s^\ell \in \{s_n, \bar s_n\}$, definable from $I$, such that $I = \bigcup_j I^j$ and $I^j$ is $(d_n-1)$-"simple" at $s^j$, i.e~the "simplicity" of $I^j$ at $s^j$ is $<d_n$.
        We again have $k=n$, and the "simplicity" of $I^j$ at $s^j$ is strictly less than the "simplicity" $d_n$ of $I$ at $s_n$, so it is strictly less than the "simplicity" of $I_n$ at $s_n$.
        \item If $\bar d_n > d_n$.
        In particular, $\bar d_n>0$ hence $\bar s_n \notin {\rig 0,\lef{|w|+1}}$ and so $s_n \notin {\lef 1,\rig{|w|}}$.
        Therefore $s_n$ has a "successor", and we let $s'$ denote the "opposite" of the "successor" of $s_n$, so that $s$ and $s'$, in some order, are of the form $\rig i$ and $\lef{i+2}$.
        We apply Claim~\ref{claim:splitting-intervals} to $s_n$ and $s'$, with $p=2$.
        Let $d'$ denote the "simplicity" of $I$ at $s'$.
        Then the claim gives "intervals" $I^1,\dots,I^\ell$ and "splits" $s^1,\dots,s^\ell \in \{\bar s_n,\bar s'\}$ which are definable from $I$, such that $I = \bigcup_j I^j$, and $I^j$ is $(d_n-1)$-"simple" at $s^j$ if $s^j= \bar s_n$, and it is $(d'-1)$-"simple" at $s^j$ if $s^j= \bar s'$.
        We now distinguish two cases.
        \begin{itemize}
            \item If $s^j=\bar s_n$.
            Then $k=n$, and the result holds has previously because $I^j$ is $(d_n-1)$-"simple".
            \item Otherwise, $s^j=\bar s'$ (which is the "successor" of $s_n$).
            In this case, by Claim~\ref{claim:technical-opposite} applied to the "funnel" $I_0,s_0,\dots,I_n,s_n$, we know that $s_k=s'$.
            Since $I \subseteq I_k$, the "simplicity" $d'$ of $I$ at $s'$ is smaller or equal to the one of $I_k$.
            We conclude since moreover $I^j$ is $(d'-1)$-"simple" at $s^j$.
        \end{itemize}
    \end{itemize}

    \vskip1em
    \noindent
    \textbf{Second step: decomposition in the basis.}
    We have proved in the first step that there are definable "intervals" and "splits" $I^j,s^j$ satisfying all the wanted properties, except that $I^j$ may not belong to $B_{s_j}$.
%    (Note that we also have $I_j \subseteq I$ since $I=\bigcup_j I^j$.)
    To overcome this, it suffices to use $\phi_\cov^1,\dots,\phi_\cov^{\ell_\cov}$ from Lemma~\ref{lem:definability-basis-formulas}, to define, for each $j$, "intervals" $I^{j,m} \in B_{s_j}$ for $m \in \{1,\dots,\ell_\cov\}$ so that $I^j \subseteq \bigcup_m I^{j,m}$, the $I^{j,m}$'s intersect $I^j$ (and therefore also $I$ since $I_j \subseteq I$) and the "simplicities" of $I^{j,m}$ at $s^j$ are smaller than the one of $I^j$.
    Hence for every $j,m$, it holds that $I_0,s_0,\dots,I_n,s_n,I^{j,m},s^{j,m}$ defines a "funnel", where $s^{j,m}=s^j$.
    Since there is a "bounded" number of $(j,m)$'s, it is easy to guarantee the last item in the lemma by reordering the "children@@funnel" accordingly (and deleting those that are not maximal).
    This concludes the proof of the lemma.
\end{proof}

We say that a "funnel" $F$ is ""reachable"" if there is a sequence $F_0,F_1,\dots,F_n=F$, where $F_0=I_0,\rig 0$ is the unique "funnel" of "length" $1$, which we call the root "funnel", and for each $k$, $F_{k+1}$ is one of the $\ell$ "children@@funnel" of~$F_k$.
\subsection{Encoding reachable funnels}\label{sec:encoding-funnels}

We should now explain how reachable "funnels" are encoded as colourings of the input word, in a way that can be updated locally (we refer to the overview (Section~\ref{sec:overview}) for an explanation why this is necessary).

\subparagraph*{High level description.}
We want to encode a "reachable" "funnel" $I_0,s_0,\dots,I_n,s_n$ by using a single monadic variable.
This is a non-trivial task, because such a "funnel" contains too many "intervals" $I_0,I_1,\dots$ to be written explicitly using "configurations@@msosi" as $[A_0,B_0],[A_1,B_1]$, and so on.
However, note that thanks to item~\ref{item:simplicity-decreases}, the sequence of "splits" $s_0,s_1,\dots$ can be encoded in a monadic variable.

The idea will be to encode a "reachable" "funnel" $I_0,s_0,I_1,s_1,\dots$ as a sequence $s_0,j_0,s_1,j_1,\dots$ where each $j_k$ belongs to $\{1,\dots,\ell\}$ (where $\ell$ is from Lemma~\ref{lem:children-funnels}) and for every $k$ it holds that $I_0,s_0,\dots,I_{k+1},s_{k+1}$ is the $j_k$-th "child@@funnel" of $I_0,s_0,\dots,I_{k},s_k$.
\AP This is easily encoded as a monadic variable using ""labelled"" "tilings", defined below.
To recover the wanted "intervals" from the encoding, we proceed as follows: guess a "configuration@@msosi" $A$ from the (unknown) "output" $I$ of the "funnel", and guess, for each $k$ (using a $\ell_\cont$-"labelled" "tiling"), an index $p_k \in \{1,\dots,\ell_\cont\}$ such that $I_k$ is the $p_k$-th sectioning "interval" for $A$ and $s_k$ (\ie the only "interval" satisfying $\phi_\cont^{p_k}(A I_k s_k)$), then using the "formulas" from Lemma~\ref{lem:children-funnels}, verify the consistency of these guesses for every $k$ (by quantifying universally over "positions").

\subparagraph*{Tilings.} \AP
A ""step@@tiling"" is an element of $\{\same,\opp,\suc\}$, where $\opp$ stands for ``"opposite"'' and $\suc$ for ``successor''.
A ""tile"" is a word $\tile = \step_1 \step_2\dots $ of "steps@@tiling" of "length" $\leq 2(d_0+2)$, such that $\step_1 \neq \same$.
To disambiguate, the "positions" $1,\dots,n$ of a "tile" are called its levels, hence $\step_h$ is the step at level $h$, or the $h$-th step.
Since their "length" (\ie the number of levels) is "bounded", there are finitely many tiles, and we let $\Lambda$ denote the finite set of tiles.

\AP
A ""block"" in a "tile" is a maximal factor from $(\opp + \suc) \same^*$; note that by our assumption that the first step is not $\same$, every "tile" can be uniquely decomposed into blocks.
Blocks starting with $\opp$ are called $\opp$-blocks and blocks starting with $\suc$ are called $\suc$-blocks. \AP
Each "block" has an ""entry direction"" and an ""exit direction"", defined by induction as follows:
\begin{itemize}
    \item the "entry direction" of the first "block" is $\rightarrow$; \AP
    \item the "entry direction" of a "block" which is not the first is the opposite of the "exit direction" of the previous "block";
    \item the "exit direction" of a $\suc$-"block" matches its "entry direction", and the "exit direction" of a $\opp$-"block" is the "opposite" of its "entry direction".
\end{itemize}
We also say that a "block" is a ""left exit"" if its "exit direction" is $\leftarrow$, and similarly for ""left entries"", ""right exits"" and ""right entries"".
Finally, we say that a "tile" is a left "tile" if its last "block" is a "left exit", and that it is a right "tile" otherwise.
%Two tilings, given in some order, are "compatible", if the number of blocks with "exit direction" $\rightarrow$ in the first "tile" matches the number of blocks with "entry direction" $\rightarrow$ in the second "tile", and the number of blocks with "exit direction" $\leftarrow$ in the second "tile" matches the number of blocks with "entry direction" $\leftarrow$ in the first "tile".
\AP
A ""tiling"" (of the fixed input string $\triangles w$) is a map $T$ from "positions" $i \in \{0,1,\dots,|w|,|w+1|\}$ of $\triangles{w}$ to the finite set $\Lambda$ of tiles.
We now define a notion of validity of "tiling", which is meant to capture whether a "tiling" induces a well-behaved sequence of "splits", as in the following example.

\begin{example}\label{ex:tiling}
    An example of a "tiling" (on the left) and the associated sequence of "splits" (on the right):
\[
    \begin{array}{c|c|c|c||c|c}
       T(0)     & T(1)    & T(2)   & T(3) & k & \text{split } s_k\\\hline
       \suc     &         &        &      & 0 & \rig 0\\
                & \opp    &        &      & 1 & \lef 1\\
                & \same   &        &      & 2 & \lef 1\\
                & \same   &        &      & 3 & \lef 1\\
        \opp    &         &        &      & 4 & \rig 0\\
                & \suc    &        &      & 5 & \rig 1\\
                &         & \suc   &      & 6 & \rig 2\\
                &         & \same  &      & 7 & \rig 2\\
                &         &        & \opp & 8 & \lef 3\\
                &         & \suc   &      & 9 & \lef 2\\
    \end{array}
\]
\end{example}
\AP
Consider a "tiling" $T$ with the property that the leftmost "tile" $T(0)$ has no "left exit" and the rightmost "tile" $T(i+1)$ has no "right exit".
Consider the directed graph $G_T$ whose vertices are the set of all blocks from tiles in $T$ and edges are defined as follows.
Let $B$ be a "block" of the $i$-th "tile" $T(i)$ and assume that $B$ is a "right exit" (in particular, $i\leq w$), and let $m$ be such that $B$ is the $m$-th "right exit" in $T(i)$.
If $T(i+1)$ has $\geq m$ "left entries", then $B$ has a single outgoing edge to the $m$-th "left entry" of $T(i+1)$.
Otherwise, $B$ has no outgoing edge.
\AP
We say that $T$ is ""valid"" if the leftmost "tile" has a "left entry" and no "left exit", the rightmost "tile" has no "right exit", and moreover, $G_T$ defines a directed path.
Note that for a "valid" "tiling", the source (\ie the unique vertex with no incoming edge) is the first "block" of the leftmost "tile", since this "block" cannot have any incoming edge.
A "valid" "tiling" has a unique "block" with no outgoing edge, which is called its ""live"" "block". \AP
We refer to the "position" $i$ of the "live" "block" as the "live" "position" of the "tiling".
Note also that in a "valid" "tiling", we have a total order on the blocks, which is obtained by visiting $G_T$ from its source to its "live" "block".
\AP
A ""location"" in a "tiling" $T$ is given by a pair $(i,h)$ where $i$ is a "position" and $h$ is a level in $T(i)$.
We get a total order on the "locations" by setting $(i,h)<(i',h')$ if the two "locations" are from different blocks $B<B'$, or the two "locations" are from the same "block" (which implies $i=i'$) and $h<h'$.
The following lemma states that being "valid", as well as this total order on "locations", can be defined using "formulas".

\begin{lemma}\label{lem:definability-before}
    There are "formulas" $\phi_\valid$ and $\phi_\before^{h,h'}$, for $h,h' \in \{1,\dots,2(d_0+2)\}$ such that for every "tiling" $T$, it holds that
    \[
        \triangles w \models \phi_\valid(T) \quad \iff \quad \text{ T is valid}
    \]
    and assuming $T$ is "valid", for every pair of "locations" $(i,h),(i',h')$, it holds that
    \[
        \triangles{w} \models \phi_\before^{h,h'}(Tii') \quad \iff \quad (i,h)<(i',h') \text{ in } T.
    \]
\end{lemma}

\begin{proof}
    Since there are finitely many tiles, "formulas" may refer to any property of tiles, as well as their levels $h$ (by ``hard-coding'' them). 
    For $\phi_\valid$, we check that
    \begin{itemize}
        \item $T(0)$ has a unique left entry and no left exits, and $T(|w|+1)$ has no right exits;
        \item for every "position" $i<|w|+1$, the number of right exits of $T(i)$ is $\leq$ the number of left entries of $T(i+1)$;
        \item for every "position" $i>0$, the number of left exits of $T(i+1)$ is $\leq$ the number of right entries of $T(i)$; and
        \item among the $2(|w|+1)$ inequalities from the two above items, exactly one of them is strict.
    \end{itemize}
%    Note that assuming $T$ is "valid", this also gives us a way to define which "location" $(i,h)$ is maximal: this is obtained by letting $i$ correspond to the unique strict equality, and taking $h$ maximal (\ie $h=|T(i)|$).

    Hence to define whether $(i,h) < (i',h')$, we guess whether there is a "tiling" $T'$ such that
    \begin{itemize}
        \item $T'$ is "valid";
        \item for each $i$, the "tile" $T'(i)$ is a prefix of $T(i)$;
        \item $(i',h')$ is not a "location" in $T'$ (\ie $h'>|T'(i)|$); and
        \item $(i,h)$ is a "location" in $T'$ (\ie $h \leq |T(i)|$).\qedhere
    \end{itemize}
\end{proof}

We say that a sequence of "splits" $s_0,s_1,\dots,s_n$ is well-formed if $s_0=\rig 0$, it satisfies item~\ref{item:directions-of-successors} from the definition of "funnels", \ie for every $k<n$, it holds that $s_{k+1}$ is either $s_k$, the "opposite" of $s_k$ or the "successor" of $s_k$, and moreover, for every "position" $i$, there are at most $2(d_0+2)$ "splits" $s_k$ in $\{\lef i, \rig i\}$.
It is easy to see that item~\ref{item:simplicity-decreases} from the definition of "funnels" (together with Lemma~\ref{lem:all-intervals-are-d-simple}) implies that for every "funnel" $I_0,s_0,\dots,I_n,s_n$, the corresponding sequence of "splits" is well-formed.

Consider a "valid" "tiling" $T$.
Then every "location" $(i,h)$ induces a "split" $s$ whose "position" is~$i$ and whose direction is $\leftarrow$ if $(i,h)$ belongs to a "block" with a "left exit", and $\rightarrow$ otherwise (note that $s$ is indeed a "split", because the leftmost "tile" has no "left exit" and the rightmost "tile" has no "right exit").
By ordering them according to the above total order on "locations", $T$ induces a sequence of "splits" $s_0,\dots,s_n$; it is a direct check that this sequence is well-formed.
Conversely, a well-formed sequence of "splits" $s_0,\dots,s_n$ induces a "valid" "tiling" $T$ (which induces that sequence of "split"), which is easily constructed by induction on $n$.

Therefore, "valid" tilings can be used to encode well-formed sequences of "splits", and each "split" $s_k$ corresponds to a "location" $(i,h)$ in $T$.
We now add labels to the tilings, so as to be able to encode "funnels".

\AP Given a nonempty finite set $F$, an ""$F$-labelled tile"" is just like a "tile", except that steps (\ie the letters of the tiles) are now from $\{\same,\opp,\suc\} \times F$.
We let $\Lambda^F$ denote the finite set of "$F$-labelled tiles" which we abbreviate as $\Lambda^\ell$ for $F=\{1,\dots,\ell\}$.
We write their letters as $\same^f, \opp^f$ or $\suc^f$ for $f \in F$.
The definitions of "valid" $F$-"labelled" tilings, and the analogue of Lemma~\ref{lem:definability-before} holds.
The "tiling" induced by a "labelled" "tiling" is obtained by removing all the labels.
%We write $\Xf T F$ for $F$-"labelled" tilings, and simplify the notation to $\Xf T \ell$ for $F=\{1,\dots,\ell\}$, which we call $\ell$-"labelled" tilings.

Likewise, an $F$-"labelled" (well-formed) "split" sequence is just like a (well-formed) "split" sequence except that every "split" also comes with a label $f \in F$.
As in the case of tiles and tilings, we write labels as exponents, \eg $s_0^{f_0},\dots,s_n^{f_n}$.
There is a one-to-one correspondence between $F$-"labelled" tilings and $F$-"labelled" well-formed "split" sequences, defined just like in the unlabelled case.
When $F=\{1,\dots,\ell\}$, we speak of $\ell$-"labelled" tilings and $\ell$-"labelled" "split" sequences.

\subparagraph*{Encoding of "funnels"}
We encode a "reachable" "funnel" as an $\ell$-"labelled" "tiling", where $\ell$ is given by Lemma~\ref{lem:children-funnels}, \ie such that every "open" "funnel" has $\ell$ "children@@funnel".
Define an ""address"" to be a nonempty sequence $1=j_0,j_1,\dots,j_{n}$ of integers in $\{1,\dots,\ell\}$. \AP
An "address" induces a "funnel" (which may or may not be defined) as follows:
\begin{itemize}
    \item the "address" $j_0=1$ induces the root "funnel" $I_0,\rig 0$, where $I_0$ is the "interval" of all "configurations@@msosi";
    \item assuming $j_0,\dots,j_{n-1}$ induces an "open" "funnel" $F=I_0,s_0,\dots,I_{n-1},s_{n-1}$ with "output" $I$, the "address" $j_0,\dots,j_n$ induces the "funnel" $F^{j_n}$ which is the $j_n$-th "child@@funnel" of $F$; formally this is the unique "funnel" $I_0,s_0,\dots,I_{n-1},s_{n-1},I_n,s_n$ such that $\triangles w \models \phi_{\childint}^{j_n}(II_ns_{n-1})$ and the relationship $x \in \{\same, \opp,\suc\}$ between $s_{n-1}$ and $s_n$ is so that $\triangles w \models \phi_{\childsplx}^{j_n}(Is_{n-1})$.
\end{itemize}
Clearly every "reachable" "funnel" is induced by an "address". % (it could be that some "funnel" is induced by several addresses).
We say that a "reachable" "funnel" $F=I_0,s_0,\dots,I_n,s_n$ is encoded by an $\ell$-"labelled" "tiling" $T$ if $T$ induces a "valid" "tiling" and the $\ell$-"labelled" sequence of "split" corresponding to $T$ is $s_0^{1},s_1^{j_1},\dots,s_n^{j_n}$ so that the "address" $j_0=1,j_1,\dots,j_{n}$ induces $F$.

In particular, the root "funnel" $F=I_0,s_0$ is encoded by the "valid" $\ell$-"labelled" "tiling" whose leftmost "tile" is the one-letter $\ell$-"labelled" "tile" $\suc^1$, and whose other tiles are all empty.
We refer to this "tiling" as the root "tiling".

\begin{lemma}[Definability of "funnels"]\label{lem:definability-funnels}
    There is a "formula" $\phi_\out$ such that for every "funnel" $I_0,s_0,\dots,I_n,s_n$ and every "interval" $I$ it holds that 
    \[
        \triangles{w} \models \phi_\out(TI) \quad \iff \quad I=\bigcap_k I_k,
    \]
    where $T$ is an encoding of $I_0,s_0,\dots,I_n,s_n$ as a $\ell$-"labelled" "tiling".
\end{lemma}

\begin{proof}
    Let $s_0^{j_0},\dots,s_n^{j_n}$ denote the $\ell$-"labelled" sequence of "splits" corresponding to $T$.
    On input $TI$, the "formula" $\phi_\out$ guesses a "configuration@@msosi" $C$ together with an $\ell_\cont$-"labelled" "tiling" $T_\cont$, and checks that the following conjunction holds:
    \begin{enumerate}[(a)]
        \item $T_\cont$ and $T$ induce the same well-formed sequence of "splits" $s_0,\dots,s_n$ (\ie for every "position" $i$, the tilings obtained from $T_\cont(i)$ and from $T(i)$ by removing the labels are the same); so we let $s_0^{p_0},\dots,s_n^{p_n}$ denote the $\ell_\cont$-"labelled" sequence of "splits" corresponding to $T_\cont$, and for all $k$ we let $I'_k$ denote the unique "interval" such that $\triangles{w} \models \phi_\cont^{p_k}(C I'_k s_k)$ ($I'_k$ is definable from $T$ together with the encoding of $k$ as a "location" $(i,h)$);
        \item\label{item:base-case} the "interval" $I'_0$ matches the "interval" $I_0$ of all "configurations@@msosi";
        \item\label{item:interval} for every $k$ (encoded as a "location" $(i,h)$), the "interval" $I'_{\leq k}= \bigcap_{m \leq k} I'_m$ (this is defined using Lemma~\ref{lem:definability-before}) is such that $\triangles{w} \models \phi_{\childint}^{j_{k+1}}(I'_{\leq k} I_{k+1} s_k)$;
        \item\label{item:split} for every $k$, we have $\triangles{w} \models \phi_{\childsplx}^{j_{k+1}}(I'_{\leq k} s_k)$ if and only if $x \in \{\same,\opp,\suc\}$ corresponds to the relationship between $s_k$ and $s_{k+1}$; and
        \item $I'_{\leq n} = I$.
    \end{enumerate}

    Let us argue that the wanted equivalence holds.
    If $\triangles{w} \models \phi_\out(TI)$ then by items~(\ref{item:interval}) and~(\ref{item:split}), it holds that, $I'_0,s_0,\dots,I'_{k+1},s_{k+1}$ is the $j_{k+1}$-st "child@@funnel" of $I'_0,s_0,\dots,I'_k,s_k$. 
    Since moreover, $I_0=I'_0$ thanks to item~(\ref{item:base-case}), and since by definition, $I_0,s_0,\dots,I_{k+1},s_{k+1}$ is the $j_{k+1}$-st "child@@funnel" of $I_0,s_0,\dots,I_k,s_k$, a direct induction yields $I'_k=I_k$ for all $k$.
    Hence $I'_{\leq n}=\bigcap_k I_k = I$ as required.

    Conversely, if $I=\bigcap_k I_k$, then taking $C$ to be any "configuration@@msosi" in $I$ (which is nonempty) and letting $T_\cont$ be given by the $\ell_\cont$-"labelled" "split" sequence $s_0^{p_0},\dots,s_n^{p_n}$, where for every $k$, $p_k \in \{1,\dots,\ell_\cont\}$ is such that $\triangles{w} \models \phi_\cont^{p_k}(C I_k)$, we get that all items are validated and therefore $\triangles{w} \models \phi_\out(TI)$.
\end{proof}

We easily derive the following definability results, where the order on encodings $T$ of "funnels" is given by the lexicographic order on the "address" (in particular, it is definable).

\begin{corollary}\label{cor:definable-stuff}
    There are "formulas" $\phi_\childspltilex^j$, for $j \in \{1,\dots, \ell\}$ and $x \in \{\same, \suc, \opp\}$, a "formula" $\phi_\closed$, "formulas" $\phi_{\produce \gamma}$ for $\gamma \in \Gamma$, and a "formula" $\phi_{\producenothing}$, such that for every "reachable" "funnel" $F$ encoded as an $\ell$-"labelled" "tiling" $T$, the following hold:
    \begin{itemize}
        \item for every $j$, $\triangles{w} \models \phi_\childspltilex^j(T)$ holds if and only if $x \in \{\same,\suc,\opp\}$ corresponds to the relationship between the last "split" of $F$ and the last "split" in the $j$-th "child@@funnel" of $F$;
        \item $\triangles{w} \models \phi_\closed(T)$ if and only if $F$ is "closed";
        \item $\triangles{w} \models \phi_{\produce \gamma}(T)$ if and only if $F$ is "closed" and with "output" $[A,A]$ such that $\triangles w \models \phi_\gamma(A)$, and moreover, $T$ is the least encoding of such a "funnel"; and
        \item $\triangles w \models \phi_{\producenothing}(T)$ if and only if $F$ is "closed" but there is a smaller encoding $T'$ of a "closed" "funnel" with the same "output".
    \end{itemize} 
\end{corollary}

We stress the fact that contrary to the "formulas" $\phi_\childsplx^j$ from Lemma~\ref{lem:children-funnels} which have free variables representing an "interval" and a "split", the $\phi_\childspltilex^j$ operate on $\ell$-"labelled" tilings.

%\pierre{J'avais mis une preuve, mais je l'ai commenté dans le .tex car je trouve finalement que c'est pas nécéssaire. On peut la remettre si besoin.}

% \begin{proof}
%     For $\phi_\childspltilex^j(T)$, we first define $s$ to be the last "split" of $F$, which is definable from $T$ thanks to Lemma~\ref{lem:definability-before}, then define $I$ to be the unique "interval" such that $\triangles w \models \phi_\out(TI)$, and then check whether $\triangles w \models \phi_\childsplx^j(Is)$.
%     (Formally, we set
%     \[
%         \phi_\childspltilex^j(\Xf X {S^\ell}) = \Xf \exists {F^2} I \  \big[\phi_\out(\Xf X {S^\ell}I) \wedge \exists s\ [ \phi_{\lastsplit}(\Xf X {S^\ell}s) \wedge \phi_\childsplx^j(I s)]\big],
%     \]
%     where $\phi_{\lastsplit}(\Xf X {S^\ell}s)$ verifies that $s$ is the last "split" in the $\ell$-"labelled" "tiling" represented by the variable $\Xf X {S^\ell}$.)
%     \pierre{on peut enlever la parenthèse. Ou bien remplacer le blabla par la parenthèse, je sais pas trop. On peut aussi enlever completement la preuve du corollaire qui est assez immédiate... j'ai plus trop de jugement à ce stade !}

%     For the "formula" $\phi_\closed$, we simply check whether the "output" is a singleton, \ie we set
%     \[
%         \phi_\closed(\Xf X {S^\ell})=\Xf \exists F A \  \phi_\out(\Xf X {S^\ell} [A,A]).
%     \]
%     For the "formula" $\phi_{\produce \gamma}$ we do the obvious thing.
% \end{proof}

\subsection{From definable funnels to yield-Hennie machines}\label{sec:funnels-to-hennie}

Let $q_0$ denote the maximal "quantifier rank" of the "formulas" defined so far; note that $q_0$ is a function of $q$ and $|F|$.
In this part of the proof, we use the convention that $\rig{-1}$ and $\lef{|w|+2}$ both denote the empty set, so that every "split" has a predecessor.
% In this part of the proof, we also use the convention that $\lef 0=\{0,1,\dots,|w|+1\}$ is a "split", whose "opposite", the empty set of "positions", is also a "split" denoted $\rig{-1}$.
% \pierre{ca ferait sûrement pas de mal de prendre cette convention dès le début, mais j'aime pas trop car ca casse un peu la symmétrie entre les deux bords, et autant faire ca tard si on peut...} 

\subparagraph*{High-level overview.}
The idea will be to construct a "yield-Hennie machine" so that on input $\triangles w$, the nodes of the run correspond to the "reachable" "funnels".
Intuitively we should produce an output letter whenever a "closed" "funnel" (which corresponds to a "configuration@@msosi") is reached.
However there is a small technical caveat here, since there may be several "closed" "funnels" corresponding to the same "configuration@@msosi".
As hinted by Corollary~\ref{cor:definable-stuff}, this is easily overcome by using a definable total order on "funnels" so that an output is produced only when a "configuration@@msosi" is reached for the first time.

We should use the tape so as to be able to perform the wanted transition at each step, \ie branch, from a given "funnel", on its $\ell$-"children@@funnel", update the tape accordingly and move in the right direction (and test if the "funnel" is "closed", in which case we output).
Since these operations were shown to be definable with "quantifier rank" $q_0$, it is sufficient to have access to the $q_0$-"type" of the current "funnel", \ie the "type" of its representation as an $\ell$-"labelled" "tiling".
To access this information, we colour the tape with the current $\ell$-"labelled" "tiling", and moreover, we keep additional information regarding its "type".

More precisely, we proceed as follows.
Given a "tiling" $T$ and a "split" $s$, define the "type" of $T$ over $s$ to be the $q_0$-"type" of $T_{s}$ over $w_s$.
From now on, when we speak of "types", we always mean $\Xf X {\Lambda^\ell}$-"types" (\ie "types" of $\ell$-"labelled" tilings) of "quantifier rank" $q_0$, and we denote their composition (\ie the operation obtained from the "compositionality" lemma which is well-behaved with respect to word concatenation) multiplicatively.

Given a "tile", define its direction in $\{\leftarrow, \rightarrow\}$ to be the "exit direction" of its last "block", and $\leftarrow$ for the empty "tile".
Given a "valid" "tiling" $T$, define the latest "split" at "position" $i$ to be the "split" $s_i$ whose "position" is $i$ and whose direction is the one of $T(i)$ (this is indeed a "split" by definition of validity.) \AP 
Given a "tiling" $T$ and a "position" $i$, we define the ""back-type"" of $T$ at $i$ to be its type over the predecessor $p_i$ of $s_i$.

Then in the execution of the machine on input $\triangles w$, we will ensure that at every moment (except during an initial preprocessing phase), if $T$ denotes the $\ell$-"labelled" "tile" encoding the "funnel" corresponding to the current node, then for every $i$, the $i$-th "position" of the working tape contains $T(i)$ together with a label indicating the "back-type" of $T$ at $i$.
Moreover, we store the "type" of $T$ over $s_i$, where $i$ is the current "position" of the mark, in the state, so that whenever a step is made say from "position" $i$ to $i+1$, we have access to the "types" of $T$ over $\lef{i+1}$ and over "positions" $\rig{i+1}$ as well as the "tile" $T(i)$.
Combining these three pieces of information by "compositionality", we recover the "type" of $T$, which is just what we need to determine whether the current "funnel" is "closed", and in this case, what should be the output, and otherwise, how to define the "children@@funnel" of the node and branch on them.

\subparagraph*{Formal definition.}
We now give the formal definition of the "yield-Hennie machine" explained above.
We let $\Typ$ denote the set of $\Xf X {\Lambda^\ell}$-"types" of "quantifier rank" $q_0$ of $(\Sigma \cup \{\triangles{,}\})$-strings and we write their composition (given by the "compositionality" lemma) multiplicatively.
For a given "tile" $\tile$ and a letter $\sigma \in \Sigma\cup \{\triangles ,\}$, we let $\type \tile \sigma$ denote the quantifier-rank $q_0$ "type" of the one-letter "tiling" $\tile$ over the one-letter word $\sigma$. 

We start with input alphabet, output alphabet, states and tape alphabet, and then define the transition function.
\begin{itemize}
    \item The input alphabet is $\Sigma$ and the output alphabet is $\Gamma$.
    \item The set of internal states is 
    \[
        Q=(\{\pre\} \cup \{\same,\opp,\suc\} \times \{1,\dots,\ell\}) \times \{\leftarrow,\rightarrow\} \times \Typ
    \]
    and the initial state is $(\pre,\rightarrow,t_\varnothing)$ where $t_\varnothing \in \Typ$ is the "type" over the empty string.
    \item The tape alphabet is 
    \[
        \Theta = (\Sigma \cup \{\triangles{,}\}) \cup \big[(\Sigma \cup \{\triangles{,}\}) \times \Lambda^\ell \times \Typ \big].
    \]
\end{itemize}

%Internal states of the form $q^\pre_\dir$ are called ``preprocessing states'', and those of the form $q^\main_\dir$ are called ``main states''.
We first define the transitions corresponding to the preprocessing phase, \ie outgoing from states in $\{\pre\} \times \{\leftarrow,\rightarrow\} \times \Typ$.
The goal of this preprocessing phase is to label each "position" $i$ of the tape with the "back-type" of the empty "tiling" at "position" $i$ (recall that the direction of the empty "tile" is $\leftarrow$); this is done by first moving the head all the way to $\triangleleft$, and them moving it back while writing the wanted "types".
Formally, we get the following transitions:
\begin{itemize}
    % \item we set
    % \[    
    %   \delta(\big(\pre,\rightarrow,t_\varnothing),\triangleright\big)=\big[(\pre,\rightarrow,t_{\suc^1,\triangleright}),(\triangleright,\suc^1,(t_\varnothing,t_\varnothing)),\rightarrow\big],
    % \] 
    % where $t_{\suc^1,\triangleright}$ is the "type" of the one-letter $\ell$-"labelled" "tile" $\suc^1$ over the word $\triangleright$ (initially, we write $\suc^1$, which is the first "tile" of the root "tiling", on the $0$-th "position" of the tape, and record its "type" in the state);
    
    \item for every $\sigma \in \Sigma \cup \{\triangleright\}$, we set 
    \[
        \delta\big((\pre,\rightarrow,t_\varnothing),\sigma\big)=\big[(\pre,\rightarrow,t_\varnothing),\sigma,\rightarrow\big],
    \]
    (we move the head to the end of the input);
    
    \item we set 
    \[
        \delta((\pre,\rightarrow,t_\varnothing), \triangleleft)=\big[(\pre,\leftarrow,t_\varnothing),\triangleleft,\stay\big]
    \]
    (we switch to state $q^\pre_\leftarrow$ when reaching the end of the input);
    
    \item for every $\sigma \in \Sigma \cup \{\triangleleft\}$ and $t \in \Typ$, we set 
    \[
        \delta\big((\pre,\leftarrow,t),\sigma\big)=\big[(\pre,\leftarrow, \type {\varepsilon_\tile}{\sigma} \cdot t),(\sigma,\varepsilon_\tile,t),\leftarrow \big],
    \]
    where $\varepsilon_\tile$ is the empty "tile"
    (we come back to the start of the input while recording the correct "back-types" on the tape);
    
    \item for every $t \in \Typ$ we set 
    \[
        \delta\big((\pre,\leftarrow,t),\triangleright\big) = \big[(\suc,1,\rightarrow,t_\varnothing),(\triangleright,\varepsilon_\tile,t),\stay \big]
    \]
    (we initialise the main computation);
    
    \item every other transition from $\{\pre\} \times \{\leftarrow,\rightarrow\} \times T$ is defined arbitrarily (these are not reached).
\end{itemize}

The next claim states the correctness of the preprocessing phase; we omit its proof which is straightforward.

\begin{claim}\label{claim:preprocessing}
    Consider the run on input $\triangles w$.
    Then there is a node "labelled" by $\big[(\suc,1,\rightarrow,t_\varnothing), u\big]$ where $u$ is a pointed $\theta$-string, such that:
    \begin{itemize}
        \item in the branch leading to that node, nodes have a single child;
        \item the first letter of $u$ is pointed; and
        \item $u \in \big[(\Sigma \cup \{\triangles{,}\}) \times \Lambda^\ell \times \Typ \big]^*$, therefore it rewrites as $u=(u^\Sigma,T,u^\Typ) \in (\Sigma \cup \{\triangles ,\})^* \times (\Lambda^\ell)^* \times \Typ^*$ so that:
        \begin{itemize}
            \item $u^\Sigma=\triangles u$;
            \item $T$ is the empty "tiling" (\ie the "tiling" comprised of empty tiles); and
            \item for every $i$, the $i$-th letter of $u^\Typ$ is the "back-type" of $T$ at $i$. 
        \end{itemize}
    \end{itemize}
\end{claim}

%\pierre{Si vous trouvez ca bizarre de mettre ce claim ici alors que le reste de la machine n'est pas définie, on peut aussi le mettre à l'intérieur du lemme plus bas}

We now define the transitions for the main part of the computation.
The idea is to maintain the fact that when we reach a node with state $(\step,j_0,\dir,t^\state)$ and we read $(\sigma,\tile,t^\tape)$ from the tape, it holds that the "type" of the "tiling" $T$ written on the (second coordinate of the) tape can be recovered from $t^\state$ and $t^\tape$, as well as the updated values for these "types".
The precise definitions depend on the "type" of the state, \ie on $\step$ and~$\dir$.

Let $\step \in \{\same,\opp,\suc\}, j_\pre \in \{1,\dots,\ell\}, t^\state,t^\tape \in \Typ ,\sigma \in \Sigma \cup \{\triangles , \}$ and $ \tile \in \Lambda^\ell$.
We set
\[
   \delta\big((\step,j_\pre,\rightarrow,t^\state), (\sigma,\tile,t^\tape)\big) = \begin{cases}
       \varepsilon & \text{ if } \phi_{\producenothing} \in t\\
       \gamma  &\text{ if } \phi_{\produce \gamma} \in t \text{ for some\footnotemark  } \ \gamma \in \Gamma \\
       \delta_1 \dots \delta_\ell &\text{ otherwise},
   \end{cases}
\]
\footnotetext{Note that in this case, $\gamma$ is uniquely defined, because a "configuration@@msosi" can only have one output.}
where $t=t^\state \cdot \type {\tile \cdot \step^{j_\pre}}{\sigma} \cdot t^\tape$,
% \[
%     t= \begin{cases}
%         t^\state \cdot \type {\tile \cdot \step^{j_\pre}}{\sigma} \cdot t^\tape & \text { if } \dir = \ \rightarrow \\
%         t^\tape \cdot \type {\tile \cdot \step^{j_\pre}}{\sigma} \cdot t^\state & \text { if } \dir = \ \leftarrow \\ 
%     \end{cases}
% \]
and for every $j \in \{1,\dots,\ell\}$, we let
\[
    \delta_j = \big[(\step_j,j,\dir',t_j),(\sigma,\tile',t'),\move_j\big]
\]
where
\[
    \dir' = \begin{cases}
        \rightarrow \text{ if } \step \in \{\same,\suc\} \\
        \leftarrow \text{ otherwise,}
    \end{cases}
\]
$\tile' = \tile \cdot \step^{j_\pre}$,
\[
    t'= \begin{cases}
        t^\state & \text{ if } \step = \suc \\
        t^\tape & \text{ if } \step \in \{\same,\opp\},
    \end{cases}
\]
and for every $j \in \{1,\dots,\ell\}$,
\begin{itemize}
    \item $\step_j$ is the unique $x \in \{\same,\opp,\suc\}$ such that $\phi_{\childspltilex}^j \in t$;
    \item $t_j$ is defined by the following table:
%     \[
%     \begin{array}{|r ||c|c|c|}\hline
%         & \step_j=\same                              & \step_j \in \{\suc,\opp\}   \\\hline \hline
% \step = \same &   t_j = t^\state                     & t_j = t^\tape \cdot \type \tile \sigma   \\ \hline

% \step = \opp & t_j =  t^\state & t_j = \type \tile \sigma \cdot t^\tape \\ \hline

% \step = \suc & t_j = t^\tape  &      t_j = t^\state \cdot \type \tile \sigma \\ \hline
        
%     \end{array}
% \]
\vspace{-0.6\baselineskip}
\[
\renewcommand{\arraystretch}{1.4}
\begin{array}{r|cc}
  & \step_j=\same & \step_j\in\{\suc,\opp\} \\ \hline
  \step=\same
    & t_j = t^{\state}
    & t_j = t^{\tape}\,\cdot\,\type {\tile'} \sigma \\[2pt]
  \step=\opp
    & t_j = t^{\state}
    & t_j = \type {\tile'} \sigma \,\cdot\,t^{\tape} \\[2pt]
  \step=\suc
    & t_j = t^{\tape}
    & t_j = t^{\state}\,\cdot\,\type {\tile'} \sigma 
\end{array}
\]
    \item $\move_j$ is defined by 
    \[
        \move_j=\begin{cases}
            \stay &\text{ if } \step_j=\same \\
            \dir' &\text{ otherwise}.
        \end{cases}
    \]
\end{itemize}
Transitions of the form $\delta\big((\step,j_\pre,\leftarrow,t^\state), (\sigma,\tile,t^\tape)\big)$ are defined the same way, except that products of "types" are performed in the opposite way.
The next lemma concludes the proof of Theorem~\ref{thm:msosi-to-hennie}.

\begin{lemma}
    The above transition function defines a "yield-Hennie machine", which, on input $\triangles w$, outputs $\Phi(w)$.
\end{lemma}

\begin{proof}
    Consider the run on input $\triangles w$. \AP
    Consider the\footnote{It is easy to see that this node is uniquely defined, for instance because none of the descendants nodes have the empty $\ell$-"labelled" "tiling" in their label.} node %"labelled" by $\big[(\suc,1,\rig,\type{\varepsilon_\tile}{\triangleright}), u\big]$  
    obtained from Claim~\ref{claim:preprocessing}, and call it the ""main root"" of the run. \AP
    The descendants of the "main root" are called the ""main nodes"" of the run.

    It follows immediately from the definition of the machine that there are four types of "main nodes":
    \begin{itemize}
        \item "main nodes" with $\ell$ non-leaf "children@@funnel", that we call open nodes;
        \item "main nodes" with a single child which is a leaf, that we call "closed" nodes;
        \item "main nodes" which are leaves "labelled" by $\Gamma$, that we call output leaves; and
        \item "main nodes" which are leaves "labelled" by $\varepsilon$, that we call skip leaves.
    \end{itemize}

    Hence to every "main node" which is not a leaf, we may associate an address $1=j_0,\dots,j_{n}$, so that
    \begin{itemize}
        \item the address of the "main root" is $j_0=1$; and
        \item for every open node with address $j_0,\dots,j_{n}$ and every $j \in \{1,\dots,\ell\}$, the address of its $j$-th child is $j_0,\dots,j_{n},j$.
    \end{itemize}
    
    Consider a "main node" which is not a leaf.
    Its label is of the form $\big[(\step,j_\pre,\dir,t^\state),u \big]$, and $u$ belongs to $\big[(\Sigma \cup \{\triangles ,\}) \times \Lambda^\ell \times \Typ \big]^*$, so it can be written $(u^\Sigma,T,u^\Typ) \in (\Sigma \cup \{\triangles ,\})^* \times (\Lambda^\ell)^* \times \Typ^*$.
    Here comes the main inductive statement.
    
    \begin{claim}
        Consider a "main node" which is not a leaf, with an address $j_0,\dots,j_{n}$ that encodes a "reachable" "funnel".
        Its label is of the form $\big[(\step,j_n,\dir,t^\state),u \big]$ where $u$ rewrites as $(u^\Sigma,T^\pre,u^\Typ)$.
        Let $i_0$ denote the "position" of the pointer in $u$, and let $s$ be the "split" with "position" $i_0$ and direction $\dir$.
        Then the following hold:
        \begin{enumerate}[(a)]
            \item\label{item:correct-sigma-word} $u^\Sigma=\triangles w$;
            \item\label{item:correct-state-type} if $\step \in \{\opp,\suc\}$ then $t^\state$ is the "type" of $T^\pre$ over the predecessor of $s$, otherwise $t^\state$ is the "type" of $T^\pre$ over $\bar s$;
            \item\label{item:correct-tape-type} for every "position" $i$, the $i$-th letter of $u_{\Typ}$ is the "back-type" of $T^\pre$ at "position" $i$;
            \item\label{item:correct-tiling} the $\ell$-"labelled" "tiling" $T$ obtained from $T^\pre$ by appending $\step^j$ to $T^\pre(i_0)$ encodes the "funnel" induced by the "address" $j_0,\dots,j_{n}$;
            \item\label{item:correct-position} $i_0$ is the "live" "position" in $T$; and
            \item\label{item:correct-direction} $\dir$ is the "entry direction" of the "live" "block" of $T$.
        \end{enumerate}
    \end{claim}

    \begin{claimproof}
        We prove the claim by (top-down) induction on the run, starting from the "main root".
        
        \vskip1em
        \noindent
        \textbf{Base case.} Claim~\ref{claim:preprocessing} tells us that the "main root", which has address $j_0=1$, has label $\big[(\suc,1,\rightarrow,t_\varnothing),u \big]$, and that  $u$ is of the form $(u^\Sigma,T^\pre,u^\Typ)$ so that:
        \begin{itemize}
            \item $T^\pre$ is the empty "tiling";
            \item $u^\Sigma=\triangles u$ so item~(\ref{item:correct-sigma-word}) holds;
            \item we have $s=\rig 0$, whose predecessor is the empty set, so $t_\varnothing$ is the "type" of $T^\pre$ over the predecessor of $s$ therefore item~(\ref{item:correct-state-type}) holds;
            \item for every $i$, the $i$-th letter of $u^\Typ$ is the "back-type" of $T^\pre$ at $i$, so item~(\ref{item:correct-tape-type}) holds;
            \item the "tiling" $T$ obtain by appending $\suc^1$ to $T^\pre(0)$ is the root "tiling", which encodes the "funnel" with "address" $j_0=1$ (\ie the root "funnel") so item~(\ref{item:correct-tiling}) holds;
            \item $i_0=0$ is both the pointed "position" of $u$ and the "live" "position" in $T$, so item~(\ref{item:correct-position}) holds;
            \item $\dir=\ \rightarrow$ which is the "entry direction" of the unique "block" in $T(0)$, so item~(\ref{item:correct-direction}) holds.
        \end{itemize}

        \vskip1em
        \noindent
        \textbf{Inductive case.} 
        Consider a node with address $j_0,\dots,j_{n}$, with pointed "position" $i_0$ and with label $\big[(\step,j_n,\dir,t^\state),u \big]$ with $u=(u^\Sigma,T^\pre,u^\Typ)$ and so that the claim holds for this node, and an integer $j \in \{1,\dots,\ell\}$.
        Assume that the node encodes a "reachable" "funnel" $F$ which is "open", so that its $j$-th "child@@funnel" is also a "reachable" "funnel", with "address" $j_0,\dots,j_n,j$.
        We assume for concreteness that $\dir = \ \rightarrow$.
        Let $(\sigma,\tile,t^\tape)$ denote the pointed letter in $u$ and let $t=t^\state \cdot \type {\tile \cdot \step^{j_n}}{\sigma} \cdot t^\tape$.
        First, we should prove that the node is open, so that its $j$-th child in the run is well-defined.

        We know by induction that the $\ell$-"labelled" "tiling" $T$ obtained from $T^\pre$ by appending $\step^{j_n}$ to $T^\pre(i_0)$ encodes $F$.
        Note that the "type" of $T$ is $t$, and therefore by our assumption that $F$ is "open", we know that $\triangles w \models \phi_\producenothing(T)$ does not hold, and neither does $\triangles w \models \phi_{\produce \gamma}$ for any $\gamma \in \Gamma$.
        Hence the node is indeed open, and we let $\big[(\step_j,j,\dir',t_j),(\sigma,\tile',t'),\move_j\big]$ be the $j$-th letter $\delta_j$ of $\delta\big((\step,j_n,\dir,t^\state), (\sigma,\tile,t^\tape)\big)$.
        
        Consider the node with address $j_0,\dots,j_{n},j$.
        By definition of the machine, it is the node with label $\big[(\step_j,j,\dir',t_j), u_j\big]$, where $u_j$ is obtained from $u$ by replacing the pointed letter by $(\sigma,\tile',t')$, and moving the pointer according to $\move_j$.
        More precisely, the "position" $i_j$ in the child node is either $i_0-1,i_0$ or $i_0+1$ according to whether $\move_j$ is $\leftarrow,\rightarrow$ or $\circ$.
        We let $(u_j^\Sigma,T,u_j^\Typ)$ denote $u_j$.

        We prove that the required items hold for the child node.
        \begin{enumerate}[(a)]
            \item By induction we know that $\sigma_{i_0}$ is the $i_0$-th letter of $\triangles w$, and therefore the first item follows.
            
            \item \begin{itemize}
                \item If $\step_j \in \{\opp,\suc\}$.
                Then we should prove that $t_j$ is the "type" of $T$ over the predecessor of $s$.
                Then we have three different cases.
                \begin{itemize}
                    \item If $\step=\same$.
                    Then $\dir'=\dir=\ \rightarrow$, hence $\move_j=\dir'=\ \rightarrow$ so $i_j=i_0+1$ hence $s=\rig{i_0+1}$ and its predecessor is $\rig{i_0}$.
                    We know by induction that $t^\tape$ is the "type" of $T^\pre$ over $\rig{i_0-1}$.
                    Therefore the "type" of $T$ over $\rig{i_0}$ is $t^\tape \cdot \type {\tile'} \sigma$, which indeed matches the definition of $t_j$ in this case.

                    \item If $\step=\opp$. 
                    Then $\dir'= \ \leftarrow$ hence $\move_j=\ \leftarrow$ so $i_j=i_0-1$ hence $s=\lef{i_0-1}$ and its predecessor is $\lef{i_0}$.
                    We know by induction that $t^\tape$ is the "type" of $T^\pre$ over $\lef{i_0+1}$.
                    Therefore the "type" of $T$ over $\lef{i_0}$ is $\type {\tile'} \sigma \cdot t^\tape$, which matches the definition of $t_j$.

                    \item If $\step=\suc$.
                    Then $\dir'=\ \rightarrow$ hence $\move_j=\ \rightarrow$ so $i_j=i_0+1$ hence $s=\rig{i_0+1}$ and its predecessor is $\rig{i_0}$.
                    Since $\step \in \{\opp,\suc\}$, we know by induction that $t^\state$ is the "type" of $T^\pre$ over the predecessor $\rig{i_0-1}$ of $\rig{i_0}$.
                    Therefore the "type" of $T$ over $\rig{i_0}$ is $t^\state \cdot \type{\tile'}{\sigma}$, which matches the definition of $t_j$.
                \end{itemize}
                \item If $\step_j = \same$.
                Then $\move_j=\circ$ so $i_j=i_0$.
                We should prove that $t_j$ is the "type" of $T$ over $\bar s$.
                Again, we have three cases.
                \begin{itemize}
                    \item If $\step=\same$.
                    Then $\dir'=\ \rightarrow$ so $s=\rig{i_0}$ hence $\bar s=\lef{i_0+1}$.
                    Since $\step =\same$ we know by induction that $t^\state$ is the "type" of $T^\pre$ over $\lef{i_0+1}$, therefore we conclude because $t_j=t^\state$ and $i_0 \notin \lef{i_0+1}$ (thus the "types" of $T$ and $T^\pre$ are the same over $\lef{i_0+1}$).
                    \item If $\step=\opp$.
                    Then $\dir'=\ \leftarrow$ so $s=\lef{i_0}$ hence $\bar s=\rig{i_0-1}$.
                    Since $\step \in \{\opp,\suc\}$, we know by induction that $t^\state$ is the "type" of $T^\pre$ over the predecessor $\rig{i_0-1}$ of $\rig{i_0}$, therefore we conclude because $t_j=t^\state$ and $i_0 \notin \rig{i_0-1}$.

                    \item If $\step=\suc$.
                    Then $\dir'=\ \rightarrow$ so $s=\rig{i_0}$ hence $\bar s=\lef{i_0+1}$.
                    We know by induction that $t^\tape$ is the "type" of $T^\pre$ over $\lef{i_0+1}$ therefore we conclude because $t_j=t^\tape$ and $i_0 \notin \lef{i_0+1}$.
                \end{itemize}
            \end{itemize}
            \item Let $i$ be a "position".
            We know by induction that the $i$-th letter of $u^\Typ$ is the "back-type" of $T^\pre$ at "position" $i$.
            If $i \neq i_0$, then this coincides with the $i$-th letter of $u_j^\Typ$ as well as the "back-type" of $T$ at "position" $i$.
            Hence we focus on $i=i_0$; recall that the $i_0$-th letter of $u_j^\Typ$ is $(\sigma,\tile',t')$, so we should prove that $t'$ is the "back-type" of $T$ at $i_0$.
            \begin{itemize}
                \item If $\step=\suc$.
                Then $t'=t^\state$, and we know by induction that $t^\state$ is the "type" of $T^\pre$ over the predecessor $\rig{i_0-1}$ of $\rig{i_0}$.
                We conclude because $i_0\notin \rig{i_0-1}$ and because the direction of $T$ at $i_0$ in this case is $\rightarrow$, so the "back-type" at $i_0$ indeed corresponds to $\rig{i_0-1}$.
                \item If $\step =\same$.
                Then $t'=t^\tape$ and we know by induction that $t^\tape$ is the "back-type" of $T^\pre$ at $i_0$.
                In this case, the "back-type" of $T^\pre$ at $i_0$ is its "type" over the predecessor $\rig{i_0-1}$ of $\rig{i_0}$, and since $i \notin \rig{i_0-1}$ this is also the "back-type" of $T$ at $i_0$.
                \item If $\step = \opp$.
                Then $t'=t^\tape$ and we know by induction that $t^\tape$ is the "back-type" of $T^\pre$ at $i_0$, which in this case is its "type" over the predecessor $\lef{i_0+1}$ of $\lef{i_0}$.
                This is also the "back-type" of $T$ at $i_0$.
            \end{itemize}

            \item As observed at the beginning of the proof, $t$ is the "type" of $T$, and therefore by definition $\step_j$ is the unique $x$ so that
            $
                \triangles w \models \phi_{\childspltilex}(T).
            $
            Therefore the $j$-th "child@@funnel" $F_j=I_0,s_0,\dots,I_n,s_n,I^j,s^j$ of $F$ is such that $\step_j$ is the relationship between $s_n$ and $s^j$.
            Hence the $\ell$-labelling which encodes $F_j$ is obtained from $T$ by appending $\step_j^j$ to $T(i)$, where $i$ is $i_0-1,i_0$ or $i_0+1$ respectively according to whether $\step$ is $\opp,\same$ or $\suc$.
            Unravelling the definition of $\move_j$, we get that this labelling is the labelling $T_j$ obtained from $T$ by appending $\step^j$ to $T(i_j)$, as required.
            
            \item The fact that $i_j$ is the "live" "position" in $T_j$ follows immediately from the analysis above.
            \item Similarly, we indeed get that $\dir'$ is the "entry direction" of the "live" "block" of $T$ (which is the "block" 
            $i_0$).\claimqedhere
        \end{enumerate}

    \end{claimproof}

        Therefore by item~\ref{item:simplicity-decreases} in the definition of "funnels", the machine under consideration indeed has the "bounded" visit property (with bound $2(d_0+2)$).
        Moreover, the facts that $\delta(\_,\_,\triangleright) \in (\Gamma \cup Q \times \{\triangleright\} \times \{\stay,\ra\})^*$ and $\delta(\_,\_,\triangleleft) \in (\Gamma \cup Q \times \{\triangleleft\} \times \{\la,\stay\})^*$ follow immediately from the definition of "funnels" (this is because the "children@@funnel" of "funnels" are "funnels", and in particular, the associated "splits" cannot exceed the bounds of the input string).
        Hence it is indeed a "yield-Hennie machine"; there remains to prove that its output coincides with $\Phi(\triangles{w})$.

        Consider a "closed" "reachable" "funnel" $F$ with "output" $[A,A]$ for some "configuration@@msosi" $A$, with "address" $j_0,\dots,j_n$.
        Consider the node in the run with the same address.
        Using the notations from the claim, we know by item~\ref{item:correct-tiling} that $T$ encodes $F$, and that therefore the "type" $t$ given by $t^\state \cdot \type {\tile \cdot \step^{j_n}}{\sigma} \cdot t^\tape$ if $\dir=\ \rightarrow$ and the reversed multiplication otherwise, is the "type" of $T$.
        Since $F$ is "closed", we have that either $\triangles w \models \phi_\producenothing(T)$ or that exactly one of the $\triangles w \models \phi_{\produce \gamma}$ holds; therefore the node is a "closed" node, and its unique child is a skip leaf in the former case, and an output leaf (with label $\gamma$) in the later.

        Thanks to the definitions of $\phi_\producenothing$ and $\phi_{\produce \gamma}$, for every "configuration@@msosi" $A$ there is a unique "reachable" "funnel" $F$ with "output" $[A,A]$ whose unique child is an output leaf, and its label is the letter produced by $A$, \ie the unique $\gamma$ such that $\triangles w \models \phi_\gamma(A)$.
        We say that $F$ is the "funnel" which outputs $A$; by definition it is the least "closed" "funnel" with "output" $[A,A]$.
        Therefore performing the yield produces all the required outputs, and there only remains to prove that these are produced in the correct order, as stated in the following claim.

        \begin{claim}
            Let $A<A'$ be two "configurations@@msosi" and let $F$ and $F'$ be the "reachable" "closed" "funnels" which "output" $A$ and $A'$ respectively.
            Then $F$ is smaller than $F'$.
        \end{claim}

        \begin{claimproof}
            Let $F'=I'_0,s'_0,\dots,I_n',s_n'$.
            Let $k$ be minimal so that $A \notin I'_k$.
            Since $A' \in I_k=[A_k,B_k]$, we have $A < A_k$ and we conclude thanks to the last property in Lemma~\ref{lem:children-funnels}.
        \end{claimproof}

        This concludes the proof.
\end{proof}

%% file: hennie-to-ariadne.tex
% !TEX root = main-set-interpretations.tex

This section proves the following translation.

\begin{lemma}\label{lemma:yHm2At}
    For every "yield-Hennie machine", there is an equivalent "Ariadne transducer".
\end{lemma}

The translation is rather straightforward: the "Ariadne transducer" traverses the whole run of the "yield-Hennie machine" in a depth-first fashion, where going down the tree corresponds to $\push$ and going back up corresponds to $\pop$.

\begin{proof}\AP%
\newfreshHennie{yHm2At}{}%
\newfreshAriadne{yHm2At}{_{\Ariadne}}%
\knowledgenewrobustcmd\nDelta{\cmdkl{\Delta}}%
Consider a "yield-Hennie machine" $\intro*\Hennie$ from $\Sigma^*$ to $\Gamma^*$, with states $\intro*\stateshennie$, initial state $q_0$, "tape alphabet@@hennie" $\intro*\tapehennie$ and "transition function@@hennie" $\intro*\transhennie\colon \stateshennie \times \tapehennie \to \nDelta^*$, where $\intro*\nDelta=\Gamma \cup \stateshennie \times \tapehennie \times \{\la,\stay,\ra\}$.

\AP We define an "Ariadne transducer" $\intro*\aria$ with states $\intro*\statesariadne:=\set{\top} \cup \stateshennie \cup \stateshennie \times \nDelta^*$.
The initial state of $\aria$ is $q_0$.
Here is the intuition.
To simulate visiting a subtree in the run of $\Hennie$ rooted at a node with position $i$, label $u$ and state $q$, assuming that the top of the current "stack" is $(q,i)$ (which will be guaranteed by induction), we proceed as follows.
\begin{enumerate}[(1)]
    \item Since $\aria$ has access to the local view at position $i$ in the current "stack", the value of the $i$-th letter $\theta$ of $u$ is available, and therefore we have access to $\transhennie(q,\Theta) = \delta_1\dots\delta_k \in \nDelta^*$.
    \item\label{item:step2} We use $\push_\stay$ to push $((q,\delta_1\dots\delta_k),i)$ on the "stack", which initialises the next phase.
    \item When the "top" of the "stack" is of the form $(q,\delta_1\dots\delta_k)$, we proceed as follows.
    \begin{enumerate}[(a)]
        \item\label{item:step-pop-top} If $k=0$ this means that we have finished visiting the subtree (or that there was no child to begin with), so we pop using the function that replaces $q$ with $\top$ on the "stack".
        \item[]\hspace*{-1.4\labelwidth}\hspace*{-\labelsep} {We now assume $k>0$.}
        \item\label{item:step3b} If $\delta_1 \in \Gamma$ we produce a letter.
        \item\label{item:step3c} Otherwise $\delta_1=(q',\theta,\dir)$ so we push state $q'$ in the direction $\dir$.
        Then the machine will inductively visit the first child of the node, ending with a $\pop$ operation which passes a function $\nextchild$ removing $\delta_1$ from the sequence $\delta_1 \dots \delta_k$ (see next item).
    \end{enumerate} 
    \item\label{item:step-pop-nextchild} Eventually, after an occurrence of step~\ref{item:step-pop-top}, the top of the "stack" is $\top$ which means that we have finished visiting the tree below $q$; therefore we pop and pass the function $\nextchild$.
\end{enumerate}

\AP We now formally define the "transition@transition function@ariadne" table $\intro*\transariadne$ of $\aria$.
First, the partial order over $\statesariadne$ is given by
\begin{itemize}
    \item the order over $\stateshennie \times \nDelta^*$ corresponds to the reverse order of the length in the $\nDelta^*$ component, \ie $(q,\alpha) < (q',\alpha')$ if and only if $|\alpha| > |\alpha'|$;
    \item $\top$ is strictly greater than all elements in $\stateshennie \cup \stateshennie \times \nDelta^*$.
\end{itemize}

\AP
We will use two update functions which corresponds to Steps~(\ref{item:step-pop-top}) and~(\ref{item:step-pop-nextchild}) above.
Step~(\ref{item:step-pop-top}) uses the function $\finish$ which maps every "state@@hennie" $q \in \stateshennie$ to $\top$ (and is undefined otherwise).
Step~(\ref{item:step-pop-nextchild}) uses the function $\intro*\nextchild$ which maps $(p,\delta_1\alpha) \in \stateshennie\times \nDelta\nDelta^*$ to $(p,\alpha)$.
Note that both are indeed "inflationary".

Let $\nDelta(p,\theta)=\delta_1\cdots\delta_k$ be a transition in $\Hennie$.
Following the above description, we get the following definition for the "transition function@@ariadne"\footnote{For readability we don't specify the productions of the transitions when they are $\varepsilon$.} $\transariadne$ of $\aria$.
\begin{itemize}    
    \item $\transariadne(\theta,p)=(\push_\stay,(p,\delta_1\cdots\delta_k))$, if $\theta\in \Sigma \cup \{\triangles,\}$ \hfill  (base case) $ \qquad \qquad \qquad$

    \item $\transariadne(\sigma,\rho (q,\theta,\dir)p)=(\push_\stay,(p,\delta_1\cdots\delta_k))$, with $\rho\in R^*$ \hfill (see Step~(\ref{item:step2})) $ \qquad \qquad \qquad$

    \item $\transariadne(\sigma,(p,\rho \gamma))=(\gamma,\push_\stay,\top)$, with $\gamma\in \Gamma$, $\rho\in R^*$ \hfill (see Step~(\ref{item:step3b})) $ \qquad \qquad \qquad$

    \item $\transariadne(\sigma,(p,(q,\theta,\dir)\rho))=(\push_\dir,q)$ \hfill (see Step~(\ref{item:step3c})) $ \qquad \qquad \qquad$

    \item $\transariadne(\sigma,(p,\epsilon))=(\pop,\finish)$ \hfill (see Step~(\ref{item:step-pop-top})) $ \qquad \qquad \qquad$

    \item $\transariadne(\sigma,\rho\top )=(\pop,\nextchild)$.\hfill (see Step~(\ref{item:step-pop-nextchild})) $ \qquad \qquad \qquad$
\end{itemize}

We omit the proof that $\aria$ computes the same function as $\Hennie$, which is a straightforward induction.
\end{proof}

%% file: ariadne-to-msosi.tex
% !TEX root =  main-set-interpretations.tex

This section gives a translation from Ariadne transducers to MSO set interpretations.

\begin{theorem}\label{thm:ariadne-to-msosi}
    For every Ariadne transducer, there exists an MSO set interpretation realising the same function.
\end{theorem}

\AP The difficulty lies in defining in MSO which stack encodings are accessible.
This is overcome by applying \Cref{thm:regularity-ariadne}, which is proved in the next section, and states that "Ariadne automata" "accept@@aa" regular languages.

\AP Here, an ""Ariadne automaton"" is just like an "Ariadne transducer" with no "production", and with a subset of ""accepting states@@aa"" $\intro*\acceptariadne \subseteq\statesariadne$; it ""accepts@@aa"" a word $w$ if the "run@@ariadne" over $w$ contains a "stack" which contains an accepting state.%
\phantomintro*"transition function@@aa"
\phantomintro*"run@@aa"

\begin{proof}[Proof of Theorem~\ref{thm:ariadne-to-msosi} assuming Theorem~\ref{thm:regularity-ariadne}]%
\newfreshAriadne{ariadne-to-msosi}{}

\AP Consider an Ariadne transducer $\intro*\aria$ from $\Sigma$ to $\Gamma$, with states $\intro*\statesariadne$ and "local stack bound" $k$.
Let $N$ be an upper bound to the size of its "productions@@ariadne".
\nathan{we can assume that productions have size at most $1$, without loss of generality.}
\nathan{je pense qu'on peut simplifier un peu la preuve: l'automate lit une configuration et simule le ariadne transducer jusqu'à arriver à la configuration en question. Par clôture par look-around le Ariadne automaton peut vérifier à toute étape que la configuration atteinte correspond à la config en input. Un façon de voir ça c'est de dire qu'à chaque transition on push un état \texttt{check} qui lit le mot en entier pour vérifier si on a atteint la bonne configuration. sinon on pop tous les checks et on revient à la position en cours et on recommence}
We construct an "MSO set interpretation" which assigns the output of $\aria$ on $w$ to $\triangles w$; to obtain the wanted set interpretation, it suffices to precompose this one by the regular function mapping $w$ to $\triangles w$ (see item 1 in \Cref{thm:closure-composition}).

Let $F_1=(Q \times \{\la,\stay,\ra\})^{\leqslant k}$ and let $F=F_1 \times \{1,\dots,N\}$, which will be the set of colours for the MSO set interpretation.
Consider a $\Sigma$-string $w$.
Consider an $F_1$-colouring $A_1$ of $w$.
A location in $A_1$ is a tuple $\ell=(q,i,\dir,h) \in Q\times \{1,\dots,|w|\}  \times \{\la,\ra\} \times \{1,\dots,k-1\}$, such that $A_1(i)$ has size $\geq h$ and its $h$-th letter is $(q,\dir)$.
We say that $i$ is the position of the location $\ell$, and $h$ is its height; there is at most one location per position and height.

We define an enumeration of locations with the following algorithm (if it fails, the enumeration is undefined):
\begin{itemize}
    \item the first location is the one with position $1$ and height $1$;
    \item assuming that $\ell_0,\dots,\ell_n=(q,i,\dir,h)$ have been enumerated, we let $\ell_{n+1}$ be the location
    \begin{itemize}
        \item with position $i'$ equal to $i-1,i$ or $i+1$ according to whether $\dir$ is $\la,\stay$ or $\ra$ (if this is out of bounds, the algorithm fails), and
        \item with height $h'$, where $h'$ is one plus the maximal height of a location in $\ell_0,\dots,\ell_n$ with position $i'$ (the maximum of the empty set is $0$).
    \end{itemize}
\end{itemize}
We say that $A_1$ is admissible if the above algorithm does not fail, and enumerates all the locations of $A_1$.
% Provided $A_1$ is admissible, its head is defined to be the last location on this enumeration, denoted $\head(A_1)$.
% The next easy claim is left to the reader.

% \begin{claim}
%     oefizqjf
% \end{claim}
Assume $A_1$ is admissible, the stack induced by a $A_1$ is the element of $(Q \times \{0,\dots,|w|+1\})^*$ obtained from the enumeration of $A_1$ by projecting each location to its first two coordinates.

Let $s_1,\dots,s_n$ denote the run of the Ariadne machine over $w$.
We say that an $F$-colouring $A=(A_1,A_2)$ is a configuration over $w$ if:
\begin{itemize}
    \item $A_2$ is a constant labelling, \ie there is $j \in \{1,\dots,N\}$ such that $A_2(i)=j$ for every $i$;
    \item $A_1$ is admissible and the stack it induces coincides with $s_p$ for some $p \in \{1,\dots,n\}$; and
    \item $j \in \{1,\dots,m\}$, where $m$ is the size of the production of $s_p$.
\end{itemize}
If the above items hold, we say that $A$ is the $j$-th occurrence of the $p$-th stack (of the run over $w$).
We order configurations in the natural way: if $A$ and $A'$ are the $j$-th and $j'$-th occurrences of the $p$-th and $p'$-th stack, then $A<A'$ if and only if $p<p'$ or ($p=p'$ and $j<j'$).

Here comes the crucial claim; its proof is relatively straightforward.

\begin{claim}
    There is an "Ariadne automaton" $\aria'$ with alphabet $(\Sigma \cup \{\triangles, \}) \times F$, which accepts a word $(\triangles w, A)$ if and only if $A=(A_1,A_2)$ is a configuration over $w$.
\end{claim}

\begin{claimproof}
    First $\aria'$ makes a pass over the input to check whether $A_2$ is constant, and if this is the case it remembers the constant $j$ in the subsequent states (if $A_2$ is not constant, the $\aria'$ stops, and thus rejects).
    Then $\aria'$ simulates the run of $\aria$ over $A$, while storing on the stack the information whether or not the current stack is a prefix of the stack $s$ induced by $A$.
    This allows $\aria'$ to check, at every step of the computation, whether the stack matches $s$.
    If this is the case, then go to an accepting state if $j$ is indeed less or equal to the length of the production, and terminate.
\end{claimproof}

Therefore by Theorem~\ref{thm:regularity-ariadne}, we obtain that there is an MSO formula $\phi_\pos$ such that for every $F$-colouring $A$,
\[
    \triangles w \models \phi_\pos(A) \quad \iff \quad \text{ $A$ is a configuration over $w$.}
\]
Defining $\phi_<(AA')$ is simple: we just check whether $A<A'$ lexicographically; this is similar to Lemma~\ref{lem:definability-before} about comparisons of tilings, so we omit the proof.
For $\gamma \in \Gamma$, we let $\phi_\gamma(A)$ express the facts that $\phi_\pos(A)$ holds, and moreover, the $j$-th letter of the production of the corresponding stack is $\gamma$, where $j$ is such that $A_2$ is constant and equal to $j$.
\end{proof}

%% file: ariadne-reflection.tex
% !TEX root =  main-set-interpretations.tex

Recall the definition of "Ariadne automata" from the previous section.
The main theorem from this section is the following.

\begin{theorem}\label{thm:regularity-ariadne}
	The word language "accepted@@aa" by an "Ariadne automaton" is effectively regular.
\end{theorem}

\newfreshAriadne{regularity-ariadne}{}%
\newfreshHennie{regularity-ariadne}{}%
Let~$\intro*\Ariadne$ be an "Ariadne automaton".
\AP The proof strategy is to construct an ""alternating Hennie automaton"" that accepts the same language.
Syntactically, an "alternating Hennie automaton" is defined just like a "yield-Hennie machine" with empty output alphabet except that the transitions, instead of producing words, produce positive Boolean combinations.
Therefore, the ""transition function@@aha"" is of the form:
\[
	\intro*\transhennie\colon\intro*\stateshennie\times \intro*\tapehennie \to B^+(\stateshennie\times\tapehennie\times \{\la,\stay,\ra\})
\]
in which $B^+(X)$ denotes the set of positive Boolean expressions built on top of~$X$.
\nat{est-ce qu'il faudrait pas mettre $B^+(\set{\texttt{true},\texttt{false}}\cup\stateshennie\times\tapehennie\times \{\la,\stay,\ra\})$  ?}

\AP A ""run@@aha"" of an "alternating Hennie automaton" is defined just like a "run@@hennie" of a "Hennie machine", except that this time, each node is decorated by a positive Boolean formula of its children.
\AP Therefore a "run@@aha" naturally defines a large positive Boolean formula; an input string~$w$ is ""accepted@@aha"" by the "alternating Hennie automaton" if this formula evaluates to $\true$. We shall also say that that a "alternating Hennie automaton" ""accepts from@@aha""~$(q,u)$ for~$u$ some "$\stateshennie$-string with end markers", of the "run@@aha" constructed starting with $(q,u)$ as root node.

The following result is a consequence of~\cite[Corollary~8.1]{DTP26}.

\begin{lemma}\label{lem:regularity-alternating-Hennie}
	Languages recognised by alternating Hennie automata are regular\footnote{This refutes a conjecture of Niwi\'nski~\cite[top of page 19]{niwinskiexercises}, because single-tape alternating Turing linear time Turing machines are the same as alternating Hennie automata.}.	
\end{lemma}

\begin{proof}
	We say that a tree whose leaves are labelled by $\{\wedge,\vee,\true,\false\}$ such that every non-leaf node has exactly one child labelled by either $\wedge$ or $\vee$ (therefore it is either a $\wedge$-node, which computes a conjunction, or a $\vee$-node, which computes a disjunction) encodes a formula.
	Consider the language $T$ of trees that encode a formula that evaluates to true (with the obvious definition).
	This is easily seen to be a regular language of trees (by performing the natural bottom-up evaluation).

	Now consider an alternating Hennie automaton $H_\text{alt}$.
	By blowing up each positive Boolean combination defining its transitions into a finite set of states, it is easy to transform it into a string-to-tree Hennie machine $H_{\text{tree}}$ that, on input $w$, outputs a tree which encodes a formula corresponding to the run of $H_{\text{alt}}$ over $w$.
	Then by definition, the language of $H_\text{alt}$ is the inverse image of $T$ by the string-to-tree function computed by $H_\text{tree}$, therefore by~\cite[Corollary~8.1]{DTP26}, it is regular.
\end{proof}

Therefore, there remains to prove the following.

\begin{lemma}\label{lem:aa-to-aha}
	For every "Ariadne automaton" there exists (effectively) an "alternating Hennie automaton" recognising the same language.
\end{lemma}

\begin{proof}
\AP\newfreshAriadne{aa-to-aha}{_{\Ariadne}}\newfreshHennie{aa-to-aha}{_{\Hennie}}
Consider an "Ariadne automaton"~$\intro*\Ariadne$ with "states@@ariadne" $\intro*\statesariadne$, "initial state@@ariadne"~$q_0$, "local stack bound"~$k$, "transition function@@aa"
\[
	\intro*\transariadne\colon  (\Sigma \cup \{\triangleright, \triangleleft\}){\times} \statesariadne^{\leqslant k}  \rightharpoonup (\{\push_\ra,\push_\stay,\push_\la\}{\times} \statesariadne \cup \set{\pop}{\times}\update(\statesariadne))
\] and "accepting states@@aa"~$\intro*\acceptariadne$.

\AP We construct an "alternating Hennie automaton"~$\intro*\Hennie$ as follows:
\begin{itemize}
\item the "states@@hennie"~$\intro*\stateshennie$ are $\statesariadne\times\update(\statesariadne)$;
\item the "initial state@@hennie" is $(q_0,f_0)$, where $f_0\in\update(\statesariadne)$ is some update that is different from every $f$ such that $\transariadne(\triangleright,q)=(\pop,f)$ for some $q \in Q$; \footnote{This is always possible since there are strictly more possible updates than states. This is a trick that allows to not duplicate the arguments in the proof for treating the case of the length one stacks, for which it is disallowed to pop. This $f_0$ can be understood as the commitment to not pop.}
\item the "tape alphabet@@hennie" is $\intro*\tapehennie:=(\Sigma \cup \{\triangleright, \triangleleft\})\times\statesariadne^{\leqslant k}$;
\item the "transition function@@aha"
	\[	
		\intro*\transhennie\colon \stateshennie\times \tapehennie\to B^+(\stateshennie\times \tapehennie\times\{\la,\stay,\ra\}),
	\]
%	is better described, as is classical, as a game between the players Exist and Forall: Exist controls the disjunctions and aims at proving the %formula true, while Forall controls the conjunctions and aims at falsifying it.
	is described as follows.
	Given a "state@@hennie" $(q,f)\in \stateshennie =  \statesariadne\times\update(\statesariadne)$ and a "tape letter"~$(\sigma,v)\in \tapehennie$, the positive Boolean formulas is: 
	\begin{enumerate}
	\item a disjunction over all choices of sequence $q_1,q_2,\dots,q_\ell\in\statesariadne$ of distinct\footnote{Hence it is a finite disjunction.} "states@@ariadne", and updates $f_1,\dots,f_{\ell-1}\in\update(\statesariadne)$  such that 
		\begin{itemize}
		\item $q_1=q$;
		\item either $q_\ell\in F$ or $\transariadne(a,vq_\ell)=(\pop,f)$; and
		\item for all~$j\in\{1,\dots,\ell-1\}$, there is~$r_j$ and $\dir_j$ such that $\transariadne(\sigma,vq_j)=(\push_{\dir_j},r_j)$,
		\end{itemize}
	\item of the conjunction of $((r_j,f_j),(\sigma,vq_j),\dir_j)$ over $j$ ranging in $\{1,\dots, \ell-1\}$.
	\end{enumerate}
	Note that, as is usual, if there is no disjunct available at step (1), the formula evaluates to false, and if there is no conjunct available at step (2), the formula evaluates to true.
\end{itemize}

\AP For this definition of~$\Hennie$ to be a valid "alternating Hennie automaton", we have to check the "bounded visit property@@hennie".
For this, it is sufficient to remark that for each transition of the "alternating Hennie automaton", the "local view" in $\statesariadne^{\leqslant k}$ which is "pointed" gets to increase its size by one. Therefore the bounded visit assumption is satisfied: the visits are bounded by the "local stack bound"~$k$ of~$\Ariadne$.

\AP We should prove that $\Ariadne$ and $\Hennie$ accept the same language.
For this, let us fix a "$\Sigma$-string with end markers" ~$\triangles w=\sigma_0\sigma_1\dots\sigma_{|w|+1}$, and show that it is accepted by one device if and only if the other also does.
\AP For $s(q,i)$ some "stack@@ariadne" of the "Ariadne automaton" over~$w$ (with "top" $(q,i)$ and "untop" $s$),
we define $\overline {s(q,i)}$ to be the "pointed $\Theta$-string"
that has as $m$-th "tape letter" $(\sigma_m,\locview s m)$ for every position $m$ and is "pointed" at $i$.
	
We prove by downward induction on~$n$, starting from $n=k(|w|+2)+1$, the following:
	\begin{quote}\AP
 	 \textbf{""Induction hypothesis@@aa2aha""}: For all~"stacks@@ariadne"~$s(q,i)$ of~$\Ariadne$ over~$w$ of length~$n$, and all updates~$f\in\update(\stateshennie)$, the following properties are equivalent:
	\begin{itemize}
	\item $\Hennie$ "accepts~$w$ from@@aha"~$((q,f),\overline{s(q,i)})$. 
	\item There is a "run@@ariadne" of $\Ariadne$ over~$w$ that starts with "stack" $s(q,i)$,
		"stays above"\footnote{We write that a "run@@ariadne" ""stays above~$s$"" if all the "stacks" it contains have~$s$ as strict prefix.}~$s$,
		and either is "accepting@@aa"
		or reaches a "stack" $s(r,i)$ such that $\transariadne(\sigma_i,\locview s ir)=(\pop,f)$.
	\end{itemize}
	\end{quote}
	
	Since "stacks@@ariadne" of $\Ariadne$ over~$w$ have length at most $k(|w|+2)$, the "induction hypothesis@@aa2aha" is established for the base case.
	For the induction step, we fix $n<k(|w|+2)+1$ and assume that the "induction hypothesis@@aa2aha" holds for "stacks" of length $n+1$.
	Let us now consider a "stack@@ariadne"~$s(q,i)$ of length~$n$ and $f\in\update(\stateshennie)$.

	\smallskip\noindent
	\textbf{"Downward implication@induction hypothesis@aa2aha".}
	Assume that the $\Hennie$ "accepts from@@aha" $((q,f),\overline{s(q,i)})$. 
	This means at step (1) that there exists a sequence of distinct "states@@ariadne"~$q_1,\dots,q_\ell$, and updates $f_1,\dots,f_{\ell-1}$ such that the assumptions of (2) are fulfilled, and the formula is true at step (2).
	In particular for all $j \in \{1,\dots,\ell-1\}$, there are~$r_j$ and $\dir_j$ such that $\transariadne(a,vq_j)=(\push_{\dir_j},r_j)$.
	Let $i_j$ be the position reached from $i$ after a $\dir_j$ move (\ie $i_j$ is $i-1,i$ or $i+1$ according to whether $\dir_j$ is $\la,\stay$ or $\ra$).
	Since $\Hennie$ "accepts from@@aha" $((q,f),\overline{s(q,i)})$, by step (2), it has to "accept from@@aha"~$((r_j,f_j),\overline{s(q,i)(r_j,i_j)})$, for all~$j\in\{1,\dots,\ell-1\}$.
	Hence, by "induction hypothesis@@aa2aha", there exists a "run@@ariadne"~$\rho_j$ that starts in $s(q_j,i)(r_j,i_j)$, "stays above"~$s(q_j,i)$, and is either "accepting@@aa", or ends in $s(q_j,i)(r'_j,i_j)$ such that $\transariadne(\sigma_{i_j},\locview{s(q_j,i)}{i_j})=(\pop,f_j)$.
	If none of these "runs@@aa" reach an accepting state, then
	\[
		\rho = s(q_1,i),\rho_1,s(q_2,i),\rho_2,\dots,\rho_{\ell-1},s(q_\ell,i)
	\]
	is a "run@@aa" of $\Ariadne$ that starts in $s(q,i)$, "stays above"~$s$, and ends in $s(q_\ell,i)$ which is such that $\transariadne(\sigma_i,\locview s iq_\ell)=(\pop,f)$.
	Otherwise, for~$j_0\in\{1,\dots,\ell-1\}$ minimal such that~$\rho_{j_0}$ is "accepting state@@aa", the run
	\[
		\rho := s(q_1,i),\rho_1,s(q_2,i),\rho_2,\dots,s(q_{j_0},i),\rho_{j_0}
	\]
	starts in $s(q,i)$, "stays above"~$s$ and is "accepting@@aa".
	In both cases, the second condition of the "induction hypothesis@@aa2aha" is established.

	\smallskip\noindent
	\textbf{"Upward implication@induction hypothesis@aa2aha".} Assume that there exists a "run@@ariadne"~$\rho$ of~$\Ariadne$ over~$w$
	that starts in~$s(q,i)$, "stays above"~$s$, and is either "accepting@@aa" or ends in a configuration~$s(q',i)$ such that $\transariadne(\sigma_i,\locview s iq')=(\pop,f)$.
	
	This "run@@ariadne" can be decomposed:
	\begin{align*}
		\text{as}\quad\rho &:= s(q_1,i),\rho_1,s(q_2,i),\dots,s(q_{\ell-1},i),\rho_{\ell-1}&&\text{if it is "accepting@@aa",}\\
		\text{or as}\quad\rho&:= s(q_1,i),\rho_1,s(q_2,i),\dots,s(q_{\ell-1},i),\rho_{\ell-1},s(q_\ell,i)&&\text{otherwise}
	\end{align*}
	where
	\begin{itemize}
	\item $q_1=q$,
	\item $\rho_j$ "stays above" $s(q_j,i)$ for all~$j\in\{1,\dots,\ell-1\}$,
	\item $\transariadne(\sigma,\locview s i q_j)$ is of the form $(\push,\dir_j,r_j)$, for all~$j\in\{1,\dots,\ell-1\}$, 
	\item $\rho_j$ starts with a "stack" of the form $s(q_j,i)(r_j,i_j)$ where $h_j$ is the position reached from $i$ by move~$\dir_j$, for all~$j\in\{1,\dots,\ell-1\}$, 
	\item $\rho_j$ ends with a "stack"~$s(q_j,i)(r'_j,i_j)$ such that $\transariadne(\sigma_{i_j},\locview{s(q_j,i)}{i_j}r'_j)$ is of the form $(\pop,f_j)$, for all~$j\in\{1,\dots,\ell-1\}$, such that $q_{j+1}$ is defined,
	\item if $\rho$ is "accepting@@aa", then~$\rho_\ell$ is "accepting@@aa",
		otherwise $q_\ell=q'$.
	\end{itemize}
	
	We should prove that the formula produced by the "alternating Hennie automaton"~$\Hennie$ from~$(q,f),\overline{s(q,i)}$ evaluates to true.
	For this, let us choose in the definition of $\transhennie$ the disjunct corresponding to~$q_1,q_2,\dots,q_{\ell-1},q_\ell$ (if $q_\ell$ is not defined, \ie if $\rho$ is "accepting@@aa", then any $q_\ell\in\statesariadne$ can be chosen), and $f_1,\dots,f_{\ell-1}$  (if $f_{\ell-1}$ is not defined, \ie if $\rho$ is "accepting@@aa", then any $f_\ell \in\update(\statesariadne)$ can be chosen).
	By "induction hypothesis@@aa2aha" applied to $s(q_j,i)(r_j,i_j)$, and $f_j$, using~$\rho_j$ as witness, we know that $\Hennie$ "accepts from@@aha" 
	$(r_j,f_j),\overline{s(q_j,i)(r_j,i_j)}$ for every $j$.
	Hence, all conjuncts are satisfied, and thus $\Hennie$ "accepts from@@aha"~$((q,f),\overline{s(q,i)})$. The upward implication of the "induction hypothesis@@aa2aha" is established.

	By applying the induction hypothesis with $s(q,i)$ being the "initial stack" $(q_0,0)$ (\ie $s$ is the empty word), we get that $\Hennie$ accepts a word $w$ if and only if either
	\begin{itemize}
		\item $\Ariadne$ accepts $w$; or
		\item on input $w$, $\Ariadne$ reaches a stack $(r,0)$ such that $\transariadne(\triangleright,r)=(\pop,f_0)$.
	\end{itemize}
	Thanks to our choice of $f_0$, the second item cannot occur, which concludes.
\end{proof}